\documentclass[10pt,epsfig]{article}
\usepackage{times}
\usepackage{epsfig}
\usepackage{amsmath}
\usepackage{amsthm}
\usepackage{color}
\usepackage{subfigure}
\usepackage{amssymb}
\usepackage{fancyhdr}
\usepackage{enumerate}
\usepackage{mathrsfs}
\usepackage{dsfont}
\usepackage{stmaryrd}
\usepackage{units}

\newtheorem{theorem}{Theorem}[section] 
\newtheorem{lemma}{Lemma}[section]
\newtheorem{corollary}{Corollary}[section]
\theoremstyle{definition}
\newtheorem{example}{Example}[section]
\newtheorem*{example*}{Example}
\newtheorem*{remark*}{Remark}
\newtheorem{remark}{Remark}[section]

\def\argmax{\mathop{\rm argmax}}
\def\argmin{\mathop{\rm argmin}}

\newcommand{\tendsto}[2]
{\raisebox{-1.2ex}{$\stackrel{\textstyle \longrightarrow}{
\scriptscriptstyle #1\rightarrow #2}$}}

\newcommand{\dint}[1]
{\raisebox{-3.6ex}{$\stackrel{\displaystyle \int\int}{
\scriptstyle #1}$}}

\newcommand{\RealF}{\mathds{R}}
\newcommand{\NaturalF}{\mathds{N}}

\newcommand{\RatF}{\mathds{Q}}
\newcommand{\Expt}{\mathds{E}}
\newcommand{\Prob}{\mathds{P}}

\newcommand{\supp}{\text{\rm supp}}
\newcommand{\isupp}{\underline{\rm s\hspace{0.1cm}}\hspace{-0.1cm}\rm upp}
\newcommand{\eps}{\varepsilon}
\newcommand{\Borel}{\mf{B}}
\newcommand{\range}{\rm range}

\newcommand{\sst}{\scriptscriptstyle}
\newcommand{\st}{\scriptstyle}

\newcommand{\m}[1]
{\mathcal{#1}}

\newcommand{\wt}[1]
{\widetilde{#1}}

\newcommand{\mf}[1]{\mathfrak{#1}}

\newcommand{\ui}{(0,1)}

\newcommand{\ssc}[1]
{{}^{#1}\hspace{-0.07cm}}

\def\SNR{\rm SNR}

\def\ll{<\hspace{-0.15cm}<}

\def\bigo{{\mathit{O}}}
\def\lito{\mathit{o}}

\def\ind{\mathds{1}}
\def\wt{\widetilde}
\newcommand{\dfn}{\triangleq}
\def\setdiff{\hspace{-0.05cm}\setminus\hspace{-0.03cm}}
\def\dom{\prec_d}

\begin{document}

\title{Optimal Feedback Communication via Posterior Matching}

\author{{\bf Ofer Shayevitz}${}^\ast$ \and {\bf Meir Feder}${}^\dagger$}

\date{}

\maketitle

\pagestyle{fancy}

\thispagestyle{empty}

\begin{abstract}
In this paper we introduce a fundamental principle for optimal communication over general memoryless channels in the presence of noiseless feedback, termed \textit{posterior matching}. Using this principle, we devise a (simple, sequential) generic feedback transmission scheme suitable for a large class of memoryless channels and input distributions, achieving any rate below the corresponding mutual information. This provides a unified framework for optimal feedback communication in which the Horstein scheme (BSC) and the Schalkwijk-Kailath scheme (AWGN channel) are special cases. Thus, as a corollary, we prove that the Horstein scheme indeed attains the BSC capacity, settling a longstanding conjecture. We further provide closed form expressions for the error probability of the scheme over a range of rates, and derive the achievable rates in a mismatch setting where the scheme is designed according to the wrong channel model. Several illustrative examples of the posterior matching scheme for specific channels are given, and the corresponding error probability expressions are evaluated. The proof techniques employed utilize novel relations between information rates and contraction properties of iterated function systems.
\end{abstract}

\vspace{0.3cm}

{\renewcommand{\thefootnote}{} \footnotetext{The work of O. Shayevitz was supported by the Adams Fellowship Program of the Israel Academy of Sciences and Humanities. This research was supported in part by the Israel Science Foundation, grant no. 223/05. \\ ${}^\ast$ O. Shayevitz is with the Information Theory \& Applications Center, University of California, San Diego, USA \{email: ofersha@ucsd.edu\}. This work has been performed while he was with the Department of EE-Systems, Tel Aviv University, Tel Aviv, Israel. \\ ${}^\dagger$ M. Feder is with the Department of EE-Systems, Tel Aviv University, Tel Aviv, Israel \{email: meir@eng.tau.ac.il\}.}}

\section{Introduction}
Feedback cannot increase the capacity of memoryless channels~\cite{shannon_zero_error,capacity_feedback_memoryless}, but can
significantly improve error probability performance, and perhaps more importantly  - can drastically simplify capacity achieving transmission schemes. Whereas complex coding techniques strive to approach capacity in the absence of feedback, that same goal can sometimes be attained using noiseless feedback via simple deterministic schemes that work ``\textit{on the fly}''. Probably the first elegant feedback scheme in that spirit is due to Horstein~\cite{horstein} for the Binary Symmetric Channel (BSC). In that work, information is represented by a uniformly distributed \textit{message point} over the unit interval, its binary expansion representing an infinite random binary sequence. The message point is then conveyed to the receiver in an increasing resolution by always indicating whether it lies to the left or to the right of its posterior distribution's median, which is also available to the transmitter via feedback. Loosely speaking, using this strategy the transmitter always answers the most informative binary question that can be posed by the receiver based on the information the latter has. Bits from the binary representation of the message point are decoded by the receiver whenever their respective intervals accumulate a sufficient posterior probability mass. The Horstein scheme was conjectured to achieve the capacity of the BSC, but this claim was verified only for a discrete set of crossover probability values for which the medians exhibit regular behavior~\cite{Schalkwijk71,Schalkwijk-Post}, and otherwise not rigorously established hitherto\footnote{The rate and error exponent analysis in the original papers~\cite{horstein,Horstein-report}, while intuitively appealing, are widely considered to be non-rigorous.}.

A few years later, two landmark papers by Schalkwijk-Kailath~\cite{Schalkwijk-Kailath} and Schalkwijk~\cite{Schalkwijk2} presented an elegant capacity achieving feedback scheme for the Additive White Gaussian Noise (AWGN) channel with an average power constraint. The Schalkwijk-Kailath scheme is ``parameter estimation'' in spirit, and its simplest realization is described as follows: Fixing a rate $R$ and a block length $n$, the unit interval is partitioned into $2^{nR}$ equal length subintervals, and a (deterministic) message point is selected as one of the subintervals' midpoints. The transmitter first sends the message point itself, which is corrupted by the additive Gaussian noise in the channel and so received with some bias. The goal of the transmitter is now to refine the receiver's knowledge of that bias, thereby zooming-in on the message point. This is achieved by computing the Minimum Mean Square Error (MMSE) estimate of the bias given the output sequence observed thus far, and sending the error term amplified to match the permissible input power constraint, on each channel use. At the end of transmission the receiver uses a nearest neighbor decoding rule to recover the message point. This linear scheme is strikingly simple and yet achieves capacity; in fact at any rate below capacity it has an error probability decaying double-exponentially with the block length, as opposed to the single exponential attained by non-feedback schemes. A clean analysis of the Schalkwijk-Kailath scheme can be found in~\cite{Kim06onreliability} and a discussion of a sequential delay-universal variant is given in~\cite{Sahai2006}.

Since the emergence of the Horstein and the Schalkwijk-Kailath schemes, it was evident that these are similar in some fundamental sense. Both schemes use the message point representation, and both attempt to ``steer'' the receiver in the right direction by transmitting what is still missing in order to ``get it right''. However, neither the precise correspondence nor a generalization to other cases has ever been established. In this paper, we show that in fact there exists an underlying principal, which we term \textit{posterior matching}, that connects these two schemes. Applying this principle, we present a simple recursive feedback transmission scheme that can be tailored to any memoryless channel and any desired input distribution (e.g., capacity achieving under some input constraints), and is optimal in the sense of achieving the corresponding mutual information, under general conditions. Loosely speaking, the new scheme operates as follows: At each time instance, the transmitter computes the posterior distribution of the message point given the receiver's observations. According to the posterior, it ``shapes'' the message point into a random variable that is independent of the receiver's observations and has the desired input distribution, and transmits it over the channel. Intuitively, this random variable captures the information still missing at the receiver, described in a way that best matches the channel input. In the special cases of a BSC with uniform input distribution and an AWGN channel with a Gaussian input distribution, the posterior matching scheme is reduced to those of Horstein and Schalkwijk-Kailath respectively, thereby also proving the Horstein conjecture as a corollary.

The paper is organized as follows. In Section~\ref{sec:prelim}, notations and necessary mathematical background are provided. In
Section~\ref{sec:scheme}, the posterior matching principle is introduced and the corresponding transmission scheme is derived. Technical regularity conditions for channels and input distributions are discussed in Section~\ref{sec:fam}. The main result of this paper, the achievability of the mutual information via posterior matching, is presented in Section~\ref{sec:main_result_memoryless}. Error probability analysis is addressed in Section~\ref{sec:error_memoryless}, where closed-form expressions are provided for a range of rates (sometimes strictly) below the mutual information. Some extensions including variants of the baseline scheme, and the penalty in rate incurred by a channel model mismatch, are addressed in Section~\ref{sec:extensions}. A discussion and some future research items appear in Section~\ref{sec:discussion}. Several illustrative examples are discussed and revisited throughout the paper, clarifying the ideas developed.

\section{Preliminaries}\label{sec:prelim}
In this section we provide some necessary mathematical background. Notations and definitions are given in Subsection~\ref{subsec:notations}. Information theoretic notions pertaining to the setting of communication with feedback are described in Subsection~\ref{subsec:information_theoretic_notions}. An introduction to the main mathematical tools used in the paper, continuous state-space Markov chains and iterated function systems, is given in Subsections~\ref{subsec:markov_chains} and~\ref{subsec:IFS}.

\subsection{Notations and Definitions}\label{subsec:notations}

Random variables (r.v.'s) are denoted by upper-case letters, their realizations by corresponding lower-case letters. A real-valued r.v. $X$ is associated with a probability distribution $P_X(\cdot)$ defined on the usual Borel $\sigma$-algebra over $\RealF$, and we write $X\sim P_X$. The \textit{cumulative distribution function} (c.d.f.) of $X$ is given by $F_X(x)= P_X\big{(}(-\infty,x\,]\big{)}$, and the inverse c.d.f. is
defined by $F_X^{-1}(t)\dfn\inf\{x:F_X(x)>t\}$. Unless otherwise stated, we assume that any real-valued r.v. $X$ is either continuous, discrete, or a mixture of the two\footnote{This restricts $F_X$ to be the sum of an absolutely continuous function (continuous part) and a jump function (discrete part). This is to say we avoid the case of a \textit{singular part}, where $P_X$ assigns positive probability to some uncountable set of zero Lebesgue measure.}. Accordingly, $X$ admits a (wide sense) \textit{probability density function} (p.d.f.) $f_X(x)$, which can be written as a mixture of a Lebesgue integrable function (continuous part) and Dirac delta functions (discrete part). If there is only a continuous part then $X$ and its distribution/c.d.f./p.d.f. are called \textit{proper}. The \textit{support} of $X$ is the intersection of all closed sets $A$ for which $P_X(\RealF\setdiff A) = 0$, and is denoted $\supp(X)$.\footnote{This coincides with the usual definitions of support for continuous and discrete r.v.'s.} For brevity, we write $P_X(x)$ for $P_X(\{x\})$, and $x\in\supp(X)$ is called a \textit{mass point} if $P_X(x)>0$. The discrete part of the support is the set of all mass points, and the continuous part the complement set. The \textit{interior} of the support is denoted by $\isupp(X)$ for short. A vector of real-valued r.v.'s $X^n=(X_1,X_2,\ldots,X_n)$ is similarly associated with $P_{X^n}$, $F_{X^n}$, $f_{X^n}$ and with $\supp(X^n)$, where the p.d.f. is now called proper if all the scalar conditional distributions are a.s. (almost surely) proper. We write $\Expt(\cdot)$ for expectation and $\Prob(\cdot)$ for the probability of a measurable event within the parentheses. The uniform probability distribution over $\ui$ is denoted throughout by $\m{U}$. A measurable bijective function $\mu:\ui\mapsto\ui$ is called a \textit{uniformity preserving function} (u.p.f.) if $\Theta\sim\m{U}$ implies that $\mu(\Theta)\sim\m{U}$.

A scalar distribution $P_X$ is said to be (strictly) \textit{dominated} by another distribution $P_Y$ if $F_X(x)<F_Y(x)$ whenever $F_Y(x)\in\ui$, and the relation is denoted by $P_X\dom P_Y$. A distribution $P_X$
is called \textit{absolutely continuous} w.r.t. another distribution $P_Y$, if $ P_Y(A)=0$ implies $P_X(A)=0$ for every $A\in\Borel$, where $\Borel$ is the corresponding $\sigma$-algebra. This relation is denoted $P_X\ll P_Y$. If both distributions are absolutely continuous w.r.t. each other, then they are said to be \textit{equivalent}. The \textit{total variation distance} between $P_X$ and $P_Y$ is defined as
\begin{equation*}
d_{TV}(P_X,P_Y) = \sup_{A\in\Borel}\left| P_X(A)-
P_Y(A)\right|
\end{equation*}
A statement is said to be satisfied for \textit{$P_X$-a.a. (almost all)} $x$, if the set of $x$'s for which it is satisfied has probability one under $P_X$.

In what follows we use ${\rm conv}(\cdot)$ for the \textit{convex hull} operator, $|\Delta|$ for the length of an interval $\Delta\subseteq\RealF$, $\log$ for $\log_2$, $\range(f)$ for the range of a function $f$, and $\circ$ for function composition. The indicator function over a set $A$ is denoted by $\ind_A(\cdot)$. A set $A\subseteq \RealF^m$ is said to be \textit{convex in the direction $u\in\RealF^m$}, if the intersection of $A$ with any line parallel to $u$ is a connected set (possibly empty). Note that $A$ is convex if and only if it is convex in any direction.

The following simple lemma states that (up to discreteness issues) any real-valued r.v. can be shaped into a uniform r.v. or vice versa, by applying the corresponding c.d.f or its inverse, respectively. This fact is found very useful in the sequel
\begin{lemma}\label{lem:matching_trans}
Let $X\sim P_X\,,\,\Theta\sim\m{U}$ be statistically independent.
Then
\begin{enumerate}[(i)]
\item $F_X^{-1}(\Theta)\sim P_X$.\label{item:match1}

\item $F_X(X)-\Theta\cdot P_X(X)\sim\,\m{U}$. Specifically, if $X$
is proper then $F_X(X)\sim\,\m{U}$.\label{item:match2}
\end{enumerate}
\end{lemma}
\begin{proof}
See Appendix \ref{app:lemmas}.
\end{proof}

A proper real-valued r.v. $X$ is said to have \textit{a regular tail} if there exists some $\gamma\in(0,\frac{1}{2}]$ and positive constants $c_0,c_1,\alpha_0,\alpha_1$, such that
\begin{equation*}
c_0f_X^{\alpha_0}(x) \leq \min\left(F_X(x),1-F_X(x)\right) \leq c_1f_X^{\alpha_1}(x)
\end{equation*}
for any $x\in\supp(X)$ satisfying
$\min\left(F_X(x),1-F_X(x)\right)\leq \gamma$.

\begin{lemma}\label{lem:regular_dist}
Let X be proper with $\supp(X)=\RealF$ and a bounded unimodal p.d.f. $f_X$. Each of the following conditions implies that $X$ has a regular tail:
\begin{enumerate}[(i)]
\item $f_X(x)=\bigo(|x|^{-a})$ and $f_X(x)=\Omega(|x|^{-b})$ as
$|x|\rightarrow\infty$, for some $b\geq
a>1$.\label{cond:poly_tail}

\item $f_X(x)=\bigo(e^{-b|x|^a})$ and $f_X(x)=\Omega(e^{-b|x|^a})$
as $|x|\rightarrow\infty$, for some $a\geq
1\,,b>0$.\label{cond:exp_tail}

\end{enumerate}
\end{lemma}
\begin{proof}
See Appendix \ref{app:misc}.
\end{proof}
\begin{example}
If $X$ is either Gaussian, Laplace or Cauchy distributed then $X$
has a regular tail.
\end{example}

\subsection{Information Theoretic Notions}\label{subsec:information_theoretic_notions}
The relative entropy between two distributions $P_X$ and $P_Y$ is denoted by $D(P_X\|P_Y)$. The mutual information between two r.v.'s $X$ and $Y$ is denoted $I(X;Y)$, and the differential entropy of a continuous r.v. $X$ is denoted $h(X)$. A \textit{memoryless channel} is defined via (and
usually identified with) a conditional probability distribution $P_{Y|X}$ on $\RealF$. The \textit{input alphabet} $\m{X}$ of the
channel is the set of all $x\in\RealF$ for which the distribution $P_{Y|X}(\cdot|x)$ is defined, the output alphabet of the channel
is the set $\m{Y}\dfn\bigcup_{x\in\m{X}}\supp(Y|X=x)\subseteq\RealF$. A sequence of real-valued r.v. pairs $\{(X_n,Y_n)\}_{n=1}^\infty$ taking values in $\m{X}\times\m{Y}$ is said to be an \textit{input/output sequence} for the memoryless channel $P_{Y|X}$ if
\begin{equation}\label{def:memoryless_channel}
P_{Y_n|X^nY^{n-1}}(\cdot|x^n,y^{n-1}) =
P_{Y|X}(\cdot|x_n)\,,\qquad n\in\NaturalF
\end{equation}
A probability distribution $P_X$ is said to be a (memoryless) \textit{input distribution} for the channel $P_{Y|X}$ if $\supp(X)\subseteq\m{X}$. The pair $(P_X,P_{Y|X})$ induces an \textit{output distribution} $P_Y$ over the output alphabet, a joint input/output distribution $P_{XY}$, and an \textit{inverse channel} $P_{X|Y}$. Such a pair $(P_X,P_{Y|X})$ is called an \textit{input/channel pair} if $I(X;Y)<\infty$.

A channel for which both the input and output alphabets
$\m{X},\m{Y}$ are finite sets is called a \textit{discrete
memoryless channel (DMC)}. Note that the numerical values of the
inputs/outputs are practically irrelevant for a DMC, and hence in
this case one can assume without loss of generality that
$\m{X}=\{0,1,\ldots,|\m{X}|-1\}$ and
$\m{Y}=\{0,1,\ldots,|\m{Y}|-1\}$. Moreover, two input/DMC pairs
$(P_X, P_{Y|X})$ and $(P_{X^*},P_{Y^*|X^*})$ are said to be
\textit{equivalent} if one can be obtained from the other by input
and output permutations, i.e., there exist permutations
$\sigma_1:\m{X}\mapsto\m{X}$ and $\sigma_2:\m{Y}\mapsto\m{Y}$ such
that
\begin{equation*}
P_X(i) = P_{X^*}(\sigma_1(i))\,,\quad  P_{Y|X}(j|i) =
P_{Y^*|X^*}(\sigma_2(j)|\sigma_1(i))
\end{equation*}
for all $i\in\m{X},j\in\m{Y}$. In particular, equivalent pairs have the same mutual information.

Let $\Theta_0$ be a random \textit{message point} uniformly
distributed over the unit interval, with its binary expansion
representing an infinite independent-identically-distributed
(i.i.d.) ${\rm Bernoulli}\left(\frac{1}{2}\right)$ sequence to be reliably conveyed
by a transmitter to a receiver over the channel $P_{Y|X}$. A
\textit{transmission scheme} is a sequence of a-priori agreed upon
measurable \textit{transmission functions}
$g_n:\ui\times\m{Y}^{n-1}\mapsto \m{X}$, so that the input to
the channel generated by the transmitter is given by
\begin{equation*}
X_n=g_n(\Theta_0,Y^{n-1})\,,\qquad n\in\NaturalF
\end{equation*}
A transmission scheme induces a distribution
$P_{X_n|X^{n-1}Y^{n-1}}$ which together with
(\ref{def:memoryless_channel}) uniquely defines the joint
distribution of the input/output sequence. In the special case
where $g_n$ does not depend on $y^{n-1}$, the transmission scheme
is said to work \textit{without feedback} and is otherwise said to
work \textit{with feedback}.

A \textit{decoding rule} is a sequence of measurable mappings
$\left\{\Delta_n:\m{Y}^{n}\mapsto\m{E}\right\}_{n=1}^\infty$,
where $\m{E}$ is the set of all open intervals in
$\ui$. We refer to $\Delta_n(y^n)$ as the \textit{decoded interval}. The
\textit{error probability} at time $n$ associated with a
transmission scheme and a decoding rule, is defined as
\begin{equation*}
p_e(n) \dfn  \Prob (\Theta_0\not\in \Delta_n(Y^n))
\end{equation*}
and the corresponding \textit{rate} at time $n$ is defined to be
\begin{equation*}
R_n \dfn  -\frac{1}{n}\log\left|\Delta_n(Y^n)\right|
\end{equation*}
We say that a transmission scheme together with a decoding rule
\textit{achieve} a rate $R$ over a channel $P_{Y|X}$ if
\begin{equation}\label{def:achievability}
\lim_{n\rightarrow \infty}\Prob(R_n < R) = 0\,,\qquad
\lim_{n\rightarrow\infty}p_e(n) = 0
\end{equation}
The rate is achieved \textit{within an input constraint}
$(\eta,u)$, if in addition
\begin{equation}\label{def:input_constraint}
\lim_{n\rightarrow\infty}n^{-1}\sum_{k=1}^n\eta(X_k) \leq u\quad\text{a.s. \;(element-wise)}
\end{equation}
where $\eta:\m{X}\mapsto\RealF^m$ is a measurable function and
$u\in\RealF^m$. A scheme and a decoding rule are also said to
\textit{pointwise achieve} a rate $R$ if for all $\theta_0\in\ui$
\begin{align*}
\lim_{n\rightarrow \infty}\Prob(R_n < R|\Theta_0=\theta_0) = 0\,,\qquad \lim_{n\rightarrow\infty}\Prob(\Theta_0\not\in \Delta_n(Y^n)|\Theta_0=\theta_0) = 0
\end{align*}
and to do the above within an input constraint $(\eta,u)$ if
(\ref{def:input_constraint}) is also satisfied. Clearly, pointwise
achievability implies achievability but not vice versa.
Accordingly, a rate $R$ is called (pointwise) achievable over a
channel $P_{Y|X}$ within an input constraint $(\eta,u)$ if there
exist a transmission scheme and a decoding rule (pointwise)
achieving it. The \textit{capacity} (with feedback) $C(P_{Y|X},\eta,u)$ of the channel
under the input constraint is the supremum of all the
corresponding achievable rates\footnote{A pointwise capacity can be defined as well, and
may be smaller than (\ref{eq:capacity_cost}) depending on the
channel. However, we do not pursue this direction.}. It is well known that the capacity
is given by~\cite{gallager_book}
\begin{equation}\label{eq:capacity_cost}
C(P_{Y|X},\eta,u) = \sup_{\stackrel{\scriptstyle P_X:\;\Expt\eta(X)\leq u}{\hspace{22pt}\supp(X)\subseteq\m{X}}} I(X;Y)
\end{equation}
Furthermore, the capacity without feedback (i.e., considering only schemes that work without feedback) is given by the above as
well. The \textit{unconstrained capacity} (i.e., when no input constraint is imposed) is denoted $C(P_{Y|X})$ for short.

An \textit{optimal fixed rate} decoding rule with rate $R$ is one that decodes an interval of length $2^{-nR}$ whose a-posteriori
probability is maximal, i.e.,
\begin{equation*}
\Delta_n(y^n) =
\argmax_{\{J\in\m{E}\,:\,|J|=2^{-nR}\}}P_{\Theta_0|Y^n}(J|y^n)
\end{equation*}
where ties are broken arbitrarily. This decoding rule minimizes the error probability $p_e(n)$ for a fixed $R_n=R$. An \textit{optimal variable rate} decoding rule with a \textit{target error probability} $p_e(n)=\delta_n$ is one that decodes a minimal-length interval whose accumulated a-posteriori probability exceeds $1-\delta_n$, i.e.,
\begin{equation*}
\Delta_n(y^n) =
\argmin_{\{J\in\m{E}\,:\,P_{\Theta_0|Y^n}(J|y^n)\geq
1-\delta_n\}}\hspace{-1cm}|J|
\end{equation*}
where ties are broken arbitrarily, thereby maximizing the instantaneous rate for a given error probability. Both decoding rules make use of the posterior distribution of the message point $P_{\Theta_0|Y^n}(\cdot|y^n)$ which can be calculated online at both terminals.

It should be noted that the main reason we adopt the above nonstandard definitions for channel coding with feedback, is that they result in a much cleaner analysis. It may not be immediately clear how this corresponds to the standard coding framework~\cite{cover}, and in particular, how achievability as defined above translates into the actual reliable decoding of messages at a desired rate. The following Lemma justifies this alternative formalization.
\begin{lemma}\label{lem:setting_eqv}
Achievability as defined in (\ref{def:achievability}) and (\ref{def:input_constraint}) above, implies achievability in the standard framework.
\end{lemma}
\begin{proof}
See Appendix \ref{app:lemmas}. Loosely speaking, a rate $R$ is achievable in our framework if the posterior distribution $P_{\Theta_0|Y^n}$ concentrates in an interval of size $\approx 2^{-nR}$ around $\Theta_0$, as $n$ grows large. This intuitively suggests that $nR$ bits from the message point representation could be reliably decoded, or, more accurately, that the unit interval can be partitioned into $\approx 2^{nR}$ intervals such that the one containing $\Theta_0$ can be identified with high probability.
\end{proof}

\subsection{Markov Chains}\label{subsec:markov_chains}
A \textit{Markov chain} $\{\Psi_n\}_{n=1}^\infty$ over a
measurable \textit{state space} $\mf{F}$, is a stochastic process
defined via an initial distribution $P_{\Psi_1}$ on $\mf{F}$, and
a \textit{stochastic kernel} (conditional probability
distribution) $\m{P}$, such that
\begin{equation*}
P_{\Psi_n|\Psi^{n-1}}(\cdot|\psi^{n-1}) =
P_{\Psi_n|\Psi_{n-1}}(\cdot|\psi_{n-1}) \dfn \m{P}(\cdot|\psi_{n-1})
\end{equation*}
We say $s\in \mf{F}$ is the \textit{initial point} of the chain if
$P_{\Psi_1}(s)=1$, and denote the probability distribution induced
over the chain for an initial point $s$ by $\m{P}_s$. The Markov
chain generated by sampling the original chain in steps of $m$ is
called the $m$-\textit{skeleton}, and its kernel is denoted by
$\m{P}^m$. The chain is said to be $P_\Psi$-\textit{irreducible}
for a distribution $ P_\Psi$ over $\mf{F}$, if any set
$A\in\Borel$ with $P_\Psi(A)>0$ is reached in a finite number of
steps with a positive probability for any initial point, where
$\Borel$ is the corresponding $\sigma$-algebra over $\mf{F}$.
$P_\Psi$ is said to be \textit{maximal} for the chain if any other
irreducibility distribution is absolutely continuous w.r.t. $
P_\Psi$. A maximal $ P_\Psi$-irreducible chain is said to be
\textit{recurrent} if for any initial point, the expected number
of visits to any set $A\in\Borel$ with $P_\Psi(A)>0$, is infinite.
The chain is said to be \textit{Harris recurrent}, if any such set
is visited infinitely often for any initial point. Thus, Harris
recurrence implies recurrence but not vice versa. A set
$A\in\Borel$ is called \textit{invariant} if $\m{P}(A|s)=1$ for
any $s\in A$. An \textit{invariant distribution} $P_\Psi$ is one
for which $P_{\Psi_{n-1}}=P_\Psi$ implies $P_{\Psi_n}=P_\Psi$.
Such an invariant distribution is called \textit{ergodic} if for
every invariant set $A$ either $P_\Psi(A)=0$ or $P_\Psi(A)=1$. A
chain which has (at least one) invariant distribution is called
\textit{positive}. For short, we use the acronym \textit{p.h.r.} to indicate
positive Harris recurrence. A chain is said to have a
\textit{$d$-cycle} if its state space can be partitioned into $d$
disjoint sets amongst which the chain moves cyclicly a.s. The
largest $d$-cycle possible is called a \textit{period}, and a
chain is called \textit{aperiodic} if its period equals one.

The following results are taken from~\cite{meyn_tweedie} and~\cite{hernandez_lasserre}. We will assume here that $\mf{F}$ is an open/closed set of $\RealF^m$ associated with the usual Borel $\sigma$-algebra $\Borel$, although the claims hold under more general conditions.

\begin{lemma}\label{lem:irreducible_to_recurrent}
An irreducible chain that has an invariant distribution is
(positive) recurrent, and the invariant distribution is unique
(and hence ergodic).
\end{lemma}

\begin{lemma}[\it p.h.r. conditions]\label{lem:PHR_cond}
Consider a chain with a kernel $\m{P}$. Each of the following
conditions implies p.h.r.:
\begin{enumerate}[(i)]
\item The chain has a unique invariant distribution $P_{\Psi}$,
and $\m{P}(\cdot|s)\ll P_{\Psi}$ for any $s\in \mf{F}$.
\label{cond:PHR_abs_cont}

\item Some $m$-skeleton $\m{P}^m$ is p.h.r.\label{cond:PHR_skeleton}
\end{enumerate}
\end{lemma}

\begin{lemma}[\it p.h.r. convergence]\label{lem:PHR_convergence}
Consider an aperiodic p.h.r. chain with a kernel $\m{P}$ and an invariant distribution $P_\Psi$. Then for any $s\in \mf{F}$
\begin{equation*}
\lim_{n\rightarrow\infty}d_{TV}\left(\m{P}^n(\cdot|s),
P_\Psi\right) = 0
\end{equation*}
\end{lemma}

\begin{lemma}[\it Strong law of large numbers (SLLN)]\label{lem:SLLN}
If $P_\Psi$ is an ergodic invariant distribution for the Markov chain $\{\Psi_n\}_{n=1}^\infty$ with kernel $\m{P}$, then for any measurable function $\eta:\mf{F}\mapsto\RealF$ satisfying $\Expt|\eta(\Psi)|<\infty$
and $P_\Psi$-a.a. initial point $s\in \mf{F}$,
\begin{equation*}
\lim_{n\rightarrow\infty}\frac{1}{n} \sum_{k=1}^n \eta(\Psi_k) =
\Expt\eta(\Psi) \qquad  \m{P}_s\text{\rm-a.s.}
\end{equation*}
Furthermore, if the chain is p.h.r. then the above holds for
any $s\in \mf{F}$.
\end{lemma}

\subsection{Iterated Function Systems}\label{subsec:IFS}
Let $\mf{F}$ be a measurable space,
$\omega:\RealF\times\mf{F}\mapsto \mf{F}$ a measurable
function\footnote{$\RealF\,$ is equipped with the usual Borel
$\sigma$-algebra, and $\RealF\times\mf{F}$ is equipped with the
corresponding product $\sigma$-algebra.}, and write
$\omega_y(\cdot)\dfn\omega(y,\cdot)$ for any $y\in\RealF$. Let
$\{Y_n\}_{n=1}^\infty$ be an i.i.d. sequence of real-valued r.v.'s. An
\textit{Iterated Function system (IFS)}
$\,\{S_n(s)\}_{n=1}^\infty$ is a stochastic process over $\mf{F}$,
defined by\footnote{We call the process itself an IFS. In the
literature sometimes $\omega_{\sst y}$ is the IFS and the process
is defined separately.}
\begin{equation}\label{eq:IFS_def}
S_1 = s\in\mf{F}\,,\quad S_{n+1}(s) =
\omega_{\sst{Y}_n}\circ\omega_{\sst{Y}_{n-1}}\circ\cdots\circ\omega_{\sst{Y}_1}(s)\vspace{-0.12cm}
\end{equation}
A \textit{Reversed IFS (RIFS)} $\{\wt{S}_n(s)\}_{n=1}^\infty$ is a
stochastic process over $\mf{F}$, obtained by a reversed order
composition:
\begin{equation}\label{eq:RIFS_def}
\wt{S}_1 = s\in\mf{F} \,,\quad \wt{S}_{n+1}(s) =
\omega_{\sst{Y}_1}\circ\omega_{\sst{Y}_{2}}\circ\cdots\circ\omega_{\sst{Y}_n}(s)\vspace{-0.1cm}
\end{equation}
We say that the (R)IFS is \textit{generated} by the (R)IFS
\textit{kernel} $\omega_{\sst y}(\cdot)$, \textit{controlled} by
the sequence $\{Y_n\}_{n=1}^\infty$, and $s$ is its initial point.
Note that an IFS is a Markov chain over the state space $\mf{F}$,
and in fact a large class of Markov chains can be represented by a
suitable IFS~\cite{Kifer}. In contrast, an RIFS is not a Markov
chain but it is however useful in the analysis of the
corresponding IFS,\footnote{The idea is that it is relatively
simple to prove (under suitable contraction conditions) that the
RIFS converges to a unique random fixed point a.s., and since the
IFS and the RIFS have the same marginal distribution, the
distribution of that fixed point must be the unique stationary
distribution of the IFS.} see e.g.~\cite{diaconis99,Steinsaltz99,stenflo2001}. However, in what
follows the RIFS will turn out to have an independent
significance.

A function $\xi:[0,1]\mapsto [0,1]$ is called a (generally
nonlinear) \textit{contraction} if it is nonnegative,
$\cap$-convex, and $\xi(x)<x$ for any $x\in(0,1]$.
\begin{lemma}\label{lem:contraction_profile}
For any contraction $\xi(\cdot)$
\begin{equation*}
r(n) \dfn \sup_{x\in[0,1]}\xi^{(n)}(x)\,,\quad
\lim_{n\rightarrow\infty} r(n) = 0
\end{equation*}
where $\xi^{(n)}$ is the $n$-fold iteration of $\,\xi$. The
sequence $r(n)$ is called the {\em decay profile} of $\,\xi$.
\end{lemma}
\begin{proof}
See Appendix \ref{app:lemmas}.
\end{proof}

\begin{example}
The function $\xi(x)= rx\,$ is a (linear) contraction for $0<r<1$, with an exponential decay profile $r(n) = r^n$.
\end{example}

\begin{example}
The function $\xi(x)= x-\alpha x^\beta$ is a contraction for
$\alpha<\frac{1}{\beta}$ and $\beta>1$, with a polynomial decay
profile $r(n) = \bigo\left(n^{\frac{1}{1-\beta}}\right)$.
\end{example}

In what follows, a measurable and surjective function
$\psi:\mf{F}\mapsto[0,1]$ is called a \textit{length function}. We
now state some useful convergence Lemmas for (R)IFS.
\begin{lemma}\label{lem:IFS_convergence}
Consider the IFS defined in (\ref{eq:IFS_def}), and suppose there
exist a length function $\psi(\cdot)$ and a contraction
$\xi(\cdot)$ with a decay profile $r(n)$, so that
\begin{equation}\label{eq:contraction_IFS_cond}
\Expt\big{[}\psi(\omega_{\sst{Y_1}}(s))\big{]} \,\leq\;
\xi(\psi(s))\,,\quad \forall s\in\mf{F}
\end{equation}
Then for any $s\in\mf{F}$ and any $\eps>0$
\begin{equation*}
\Prob\big{(}\psi(S_n(s)) > \eps\big{)} \leq  \eps^{-1}r(n)
\end{equation*}
\end{lemma}
\begin{proof}
See Appendix \ref{app:lemmas}.
\end{proof}

In the sequel, we consider an IFS over the space $\mf{F}_{c}$ of
all c.d.f. functions over the open unit
interval\footnote{$\mf{F}_{c}$ is associated with the topology of
pointwise convergence, and the corresponding Borel
$\sigma$-algebra.}, i.e., all monotone non-decreasing functions
$h:\ui\mapsto\ui$ for which ${\rm conv}(\range(h))=\ui$.
Furthermore, we define the following family of length functions
over $\mf{F}_{c}$:
\begin{equation}\label{eq:def_lenght_func}
\psi_{\sst{\lambda}}(h) \dfn \int_0^1 \lambda(h(x))dx\,, \quad
h\in\mf{F}_{c}
\end{equation}
where $\lambda:[0,1]\mapsto[0,1]$ is surjective, $\cap$-convex and
symmetric about $\frac{1}{2}$.

For any $h:\RealF\mapsto\RealF$ and $s,t\in\RealF$, define
\begin{equation}\label{eq:lip_operator}
D_{s,t}(h) \dfn \frac{|h(s)-h(t)|}{|s-t|}\;,\quad D_s(h)
\dfn \limsup_{t\rightarrow s}D_{s,t}(h)
\end{equation}
$D_{s,t}(\cdot)$ and $D_s(\cdot)$ are called \textit{global} and
\textit{local Lipschitz operators} respectively.

\begin{lemma}\label{lem:contraction1}
Consider the RIFS in (\ref{eq:RIFS_def}) over some interval
$\mf{F}\subseteq\RealF$, and suppose the following conditions hold
for some $q>0$:
\begin{equation}\label{eq:contraction_cond}
r \dfn \sup_{s\neq t\in\mf{F}}\Expt
\left[D_{s,t}(\omega_{\sst{Y_1}})\right]^q < 1
\end{equation}
Then for any $\eps>0$
\begin{equation*}
\Prob\,\left(\left|\wt{S}_n(s)-\wt{S}_n(t)\right| >
\eps\right)\leq \eps^{-q}|s-t|^qr^n \qquad s,t\in\mf{F}
\end{equation*}
\end{lemma}
\begin{proof}
See Appendix \ref{app:lemmas}.
\end{proof}

\begin{lemma}[\it From~\cite{Steinsaltz99}]\label{lem:contraction2}
Consider the RIFS in (\ref{eq:RIFS_def}) over the interval
$\mf{F}=\ui$. Let $\rho:\ui\mapsto[1,\infty)$ be a continuous
function, and define
\begin{align*}
J(s;t) \dfn \sup\left\{\rho({\rm conv}\{s,t\})\right\}\,,\quad
K_s \dfn \Expt\left[\,J(s;\omega_{\sst{Y_1}}(s))\right]\,,\quad \Psi(x,z,\alpha) \dfn \frac{K_s+K_t}{1-r} + 2J(s;t)
\end{align*}
If
\begin{equation*}
r\dfn
\sup_{s\in\mf{F}}\Expt\left[\frac{\rho(\omega_{\sst{Y_1}}(s))}{\rho(s)}\,D_s(\omega_{\sst{Y_1}})\right]
< 1,
\end{equation*}
then for any $s,t\in\ui$ and any $\eps>0$
\begin{align*}
\Prob\,\left(\left|\wt{S}_n(s)-\wt{S}_n(t)\right| >
\eps\right) \leq \eps^{-\sst{1}}\Psi(s,t,r)\cdot r^n
\end{align*}
\end{lemma}

\section{Posterior Matching}\label{sec:scheme}
In this section, we introduce the idea of posterior matching and develop the corresponding framework. In Subsection~\ref{subsec:pm_prin}, a new fundamental principle for optimal communication with feedback is presented. This principle is applied in Subsection~\ref{subsec:pm_scheme}, to devise a general transmission scheme suitable for any given input/channel pair $(P_X,P_{Y|X})$,\footnote{For instance, $P_X$ may be selected to be capacity achieving for $P_{Y|X}$, possibly under some desirable input constraints.}. This scheme will later be shown (in Section~\ref{sec:main_result_memoryless}) to achieve any rate below the corresponding mutual information $I(X;Y)$, under general conditions. A recursive representation of the scheme in a continuous alphabet setting is developed, where the recursion rule is given as a simple function of the input/channel pair $(P_X,P_{Y|X})$. A common framework for discrete, continuous and mixed alphabets is introduced in Subsection~\ref{subsec:norm_chan}, and a corresponding unified recursive representation is provided. Several illustrative examples are discussed throughout the section, where in each the corresponding scheme is explicitly derived. In the special cases of the AWGN channel with a Gaussian input, and the BSC with a uniform input, it is demonstrated how the scheme reduces to the Schalkwijk-Kailath and Horstein schemes, respectively.

\subsection{The Basic Principle}\label{subsec:pm_prin}
Suppose the receiver has observed the output sequence $Y^n$, induced by a message point $\Theta_0$ and an arbitrary transmission scheme used so far. The receiver has possibly gained some information regarding the value of $\Theta_0$ via $Y^n$, but what is the information it is still missing? We argue that a natural candidate is any r.v. $U$ with the following properties:
\begin{enumerate}[(I)]
\item \textit{$U$ is statistically independent of $Y^n$}. \label{prop:PM_indp}
\item \textit{The message point $\Theta_0$ can be a.s. uniquely recovered from $(U,Y^n)$}. \label{prop:PM_invert}
\end{enumerate}
Intuitively, the first requirement guarantees that $U$ represents ``new information'' not yet observed by the receiver, while the second requirement makes sure this information is ``relevant'' in terms of describing the message point. Following this line of thought, we suggest a simple principle for generating the next channel input: \vspace{3pt}

\textit{The transmission function $g_{n+1}$ should be selected so that $X_{n+1}$ is $P_X$-distributed, and is a fixed function\footnote{By \textit{fixed} we mean that the function cannot depend on the outputs $y^n$, so that $X_{n+1}$ is still independent of $Y^n$.} of some r.v. $U$ satisfying properties ({\rm \ref{prop:PM_indp}}) and ({\rm \ref{prop:PM_invert})}}.

\vspace{3pt}
That way, the transmitter attempts to convey the missing information to the receiver, while at the same time satisfying the input constraints encapsulated in $P_X$\footnote{The extra degree of freedom in the form of a deterministic function is in fact significant only when $P_X$ has a discrete part, in which case a quantization of $U$ may void property (\ref{prop:PM_invert}).}. We call this the \textit{posterior matching principle} for reasons that will become clear immediately. Note that any transmission scheme adhering to the posterior matching principle, satisfies
\begin{equation}\label{eq:PM_MI}
I(\Theta_0;Y_{n+1}|Y^n) = I(\Theta_0,Y^n; Y_{n+1})  - I(Y_{n+1}; Y^n) = I(X_{n+1};Y_{n+1})  - I(Y_{n+1}; Y^n) = I(X;Y)
\end{equation}
The second equality follows from the memorylessness of the channel and the fact that $X_{n+1}$ is a function of $(\Theta_0,Y^n)$. The last equality holds since $X_{n+1}\sim P_X$, and since $Y_{n+1}$ is independent of $Y^n$, where the latter is implied by property (\ref{prop:PM_indp}) together with the memorylessness of the channel. Loosely speaking, a transmission scheme satisfying the posterior matching principle therefore conveys, on each channel use, ``new information'' pertaining to the message point that is equal to the associated one-shot mutual information. This is intuitively appealing, and gives some idea as to why such a scheme may be good. However, this property does not prove nor directly implies anything regarding achievability. It merely indicates that we have done ``information lossless'' processing when converting the one-shot channel into an $n$-shot channel, an obvious necessary condition. In fact, note we did not use property (\ref{prop:PM_invert}), which turns out to be important\footnote{One can easily come up with useless schemes for which only property (\ref{prop:PM_indp}) holds. A simple example is repetition: Transmit the binary representation of $\Theta_0$ bit by bit over a BSC, independent of the feedback.}.

The rest of this paper is dedicated to the translation of the posterior matching principle into a viable transmission scheme, and to its analysis. As we shall see shortly, there are infinitely many transmission functions that satisfy the posterior matching principle. There is however one baseline scheme which is simple to express and analyze.

\subsection{The Posterior Matching Scheme}\label{subsec:pm_scheme}
\begin{theorem}[\it Posterior Matching Scheme]\label{thrm:post_match_scheme}
The following transmission scheme satisfies the posterior matching principle for any $n$:
\begin{equation}\label{eq:gn_func}
g_{n+1}(\theta,y^n) = F_X^{-1}\circ F_{\Theta_0|Y^n}\,(\theta|y^n)
\end{equation}
Based on the above transmission functions, the input to the channel is a sequence of r.v.'s given by
\begin{equation}\label{eq:gn_funcX}
X_{n+1} = F_X^{-1}\circ F_{\Theta_0|Y^n}\,(\Theta_0|Y^n)
\end{equation}
\end{theorem}

\begin{proof}
Assume $P_{\Theta_0|Y^n}(\cdot|y^n)$ is proper for any $y^n\in\m{Y}^n$. Then Lemma \ref{lem:matching_trans} claim (\ref{item:match2}) implies that $F_{\Theta_0|Y^n}\,(\Theta_0|y^n)\sim\m{U}$, and since this holds
for all $y^n$ then $F_{\Theta_0|Y^n}\,(\Theta_0|Y^n)\sim\m{U}$ and is statistically independent of $Y^n$. It is easy to see that for any $y^n$, the mapping $F_{\Theta_0|Y^n}\,(\cdot|y^n)$ is injective when its domain is restricted to $\supp\left(P_{\Theta_0|Y^n}\,(\cdot |y^n)\right)$, thus $\Theta_0$ can be a.s. uniquely recovered from $(F_{\Theta_0|Y^n}\,(\Theta_0|Y^n),Y^n)$. Hence, we conclude that $F_{\Theta_0|Y^n}\,(\Theta_0|y^n)$ satisfies properties (\ref{prop:PM_indp}) and (\ref{prop:PM_invert}) required by the posterior matching principle . By Lemma \ref{lem:matching_trans} claim (\ref{item:match1}), applying the inverse c.d.f. $F_X^{-1}$ merely shapes the uniform distribution into the distribution $P_X$. Therefore, $X_{n+1}$ is $P_X$-distributed and since it is also a deterministic function of $F_{\Theta_0|Y^n}\,(\Theta_0|Y^n)$, the posterior matching principle is satisfied. See Appendix \ref{app:lemmas} to eliminate the properness assumption.
\end{proof}

Following the above, it is now easy to derive a plethora of schemes satisfying the posterior matching principle.
\begin{corollary}\label{cor:eq_schemes}
Let $\{\mu_n\}_{n=1}^\infty$ be a sequences of u.p.f's, and let $\{\varsigma_n:\ui\mapsto\ui\}_{n=1}^\infty$  be a sequence of measurable bijective functions. The transmission scheme given by
\begin{equation*}
g_{n+1}(\theta,y^n) = F_X^{-1}\circ \mu_n \circ P_{\Theta_0|Y^n}\,\left(\varsigma_n^{-1}\left(\left(0,\varsigma_n\left(\theta\right)\right]\right)|y^n\right)
\end{equation*}
satisfies the posterior matching principle for any $n$. In particular, a scheme obtained by fixing $\mu_n=\mu$ and $\varsigma_n$ to be the identity function\footnote{In fact, letting $\varsigma_n$ be any sequence of monotonically increasing functions results in the same scheme. This fact is used in the error probability analysis on Section~\ref{sec:error_memoryless}, to obtain tighter bounds.} for all $n$, is called a {\em $\mu$-variant}. The transmission scheme corresponding to a $\mu$-variant is thus given by
\begin{equation}\label{eq:PM_mu_variants}
g_{n+1}(\theta,y^n) = F_X^{-1}\circ \mu \circ F_{\Theta_0|Y^n}\,(\theta|y^n)
\end{equation}
Finally, the baseline scheme (\ref{eq:gn_func}) is recovered by setting $\mu$ to be the identity function.
\end{corollary}

We note that the different schemes described above have a similar flavor. Loosely speaking, the message point is described each time at a resolution determined by the current uncertainty at the receiver, by somehow stretching and redistributing the posterior probability mass so that it matches the desired input distribution (we will later see that the ``stretching rate'' corresponds to the mutual information). This interpretation explains the posterior matching moniker. From this point forward we mostly limit our discussion to the baseline scheme described by (\ref{eq:gn_func}) or (\ref{eq:gn_funcX}), which is henceforth called the \textit{posterior matching scheme}. The \textit{$\mu$-variants} (\ref{eq:PM_mu_variants}) of the scheme will be discussed in more detail on Section~\ref{sec:extensions}-\ref{subsec:mu_var}.

As it turns out, the posterior matching scheme may sometimes admit a simple recursive form.
\begin{theorem}[\it Recursive representation I]\label{thrm:recursion} If $P_{XY}$ is proper, then the posterior matching scheme (\ref{eq:gn_func}) is also given by
\begin{equation}\label{eq:gn_rec}
g_1(\theta) = F_X^{-1}(\theta)\,,\quad g_{n+1}(\theta|y^n) = \Big{(}F_X^{-1}\circ F_{X|Y}(\cdot|y_n)\Big{)}\circ g_n(\theta|y^{n-1})
\end{equation}
Moreover, the corresponding sequence of input/output pairs
$\left\{(X_n,Y_n)\right\}_{n=1}^\infty$ constitute a Markov chain
over a state space $\isupp(X,Y)\subseteq\RealF^2$, with an
invariant distribution $P_{XY}$, and satisfy the recursion rule
\begin{equation}\label{eq:gn_recX}
X_1 = F_X^{-1}(\Theta_0)\,,\quad  X_{n+1} = F_X^{-1}\circ F_{X|Y}(X_n|Y_n)
\end{equation}
\end{theorem}

\begin{proof}
The initialization $g_1(\theta) = F_X^{-1}(\theta)$ results
immediately from (\ref{eq:gn_func}), recalling that $\Theta_0$ is
uniform over the unit interval. To prove the recursion relation,
we notice that since $P_{XY}$ is proper then the transmission
functions $g_n(\theta,y^{n-1})$ are continuous when restricted to the support of
the posterior, and strictly increasing in $\theta$ for any fixed
$y^{n-1}$. Therefore, we have the following set of equalities:
{\allowdisplaybreaks
\begin{align}\label{eq:rec_derivation}
\nonumber F_{\sst{\Theta_0}|\sst{Y}^n}(\theta|\,y^n) &=
\Prob(\Theta_0 \leq \theta|\,Y^n = y^n) \stackrel{(\rm a)}{=} \Prob(g_n(\Theta_0,y^{n-1})\leq g_n(\theta,y^{n-1})|\,Y^n=y^n)
\\
\nonumber & = \Prob(X_n\leq g_n(\theta,y^{n-1})|\,Y^n=y^n) \stackrel{(\rm b)}{=} \Prob(X_n\leq
g_n(\theta,y^{n-1})|\,Y_n=y_n)
\\
&= F_{X|Y}(g_n(\theta,y^{n-1})|\,y_n)
\end{align}}
where in (a) we used the continuity and monotonicity of
the transmission functions, and in (b) we used the
facts that the channel is memoryless and that by construction
$X_n$ is statistically independent of $Y^{n-1}$, which also imply
that $Y^n$ is an i.i.d. sequence. The recursive rule
(\ref{eq:gn_rec}) now results immediately by combining
(\ref{eq:gn_func}) and (\ref{eq:rec_derivation}).

Now, using (\ref{eq:gn_funcX}) we obtain
\begin{equation*}
X_{n+1} = F_X^{-1}\circ F_{\sst{\Theta_0}|\sst{Y}^n}(\Theta_0|\,Y^n) = F_X^{-1}\circ
F_{X|Y}(g_n(\Theta_0,Y^{n-1})|\,Y_n) = F_X^{-1}\circ F_{X|Y}(X_n|\,Y_n)
\end{equation*}
yielding relation (\ref{eq:gn_recX}). Since $Y_n$ is generated from $X_n$ via a memoryless channel,
the Markovity of $\left\{(X_n,Y_n)\right\}_{n=1}^\infty$ is
established. The distribution $P_{XY}$ is invariant since by
construction $(X_n,Y_n)\sim P_{XY}$ implies $X_{n+1}\sim P_X$, and
then $Y_{n+1}$ is generated via the memoryless channel $P_{Y|X}$.
Taking the state space to be $\isupp(X,Y)$ is artificial here
since $P_{XY}\big{(}\supp(X,Y)\setminus \isupp(X,Y)\big{)} = 0$,
and is done for reasons of mathematical convenience to avoid
having trivial invariant distributions (this is not true when
$P_{XY}$ is not proper). Note that the chain \textit{emulates} the
``correct'' input marginal and the ``correct'' joint (i.i.d.)
output distribution; this interpretation is further discussed in Section~\ref{sec:discussion}.
\end{proof}

In the sequel, we refer to the function $F_X^{-1}\circ F_{X|Y}$
appearing in the recursive representation as the \textit{posterior
matching kernel}. Let us now turn to consider several examples,
which are frequently revisited throughout the paper.
\begin{example}[\it AWGN channel]\label{ex:AWGN}
Let $P_{Y|X}$ be an AWGN channel with noise variance $\rm N$, and
let us set a Gaussian input distribution $X\sim\mathcal{N}(0,{\rm
P})$, which is capacity achieving for an input power constraint
$\rm P$. We now derive the posterior matching scheme in this case,
and show it reduces to the Schalkwijk-Kailath scheme. Let
$\SNR\dfn\frac{\rm P}{\rm N}$. Standard manipulations yield the
following posterior distribution
\begin{equation}\label{eq:gauss_post}
X|Y=y\;\sim\;\m{N}\left(\frac{\SNR}{1+\SNR}\cdot
y\;,\;\frac{1}{1+\SNR}\cdot P\right)
\end{equation}
The joint p.d.f. $f_{XY}$ is Gaussian and hence proper, so the
recursive representation of Theorem \ref{thrm:recursion} is valid.
By definition, the corresponding posterior matching kernel satisfies
\begin{equation}\label{eq:gauss1}
F_X^{-1}\circ F_{X|Y}(x|y) = \{z\,:\, F_X(z) = F_{X|Y}(x|y)\}
\end{equation}
However, from Gaussianity and (\ref{eq:gauss_post}) we know that
\begin{equation}\label{eq:gauss2}
F_{X|Y}(x|y) =
F_X\left(\sqrt{1+\SNR}\left(x-\frac{\SNR}{1+\SNR}\cdot
y\right)\right)
\end{equation}
Combining (\ref{eq:gauss1}) and (\ref{eq:gauss2}), the posterior
matching kernel for the AWGN channel setting is given by
\begin{equation}\label{eq:AWGN_kernel}
F_X^{-1}\circ F_{X|Y}(x|y)
=\sqrt{1+\SNR}\left(x-\frac{\SNR}{1+\SNR}\cdot y\right)
\end{equation}
and hence the posterior matching scheme is given by
\begin{equation}\label{eq:AWGN_scheme}
X_1=F_X^{-1}(\Theta_0)\,,\;\; X_{n+1} =
\sqrt{1+\SNR}\left(X_n-\frac{\SNR}{1+\SNR}\,Y_n\right)
\end{equation}
From the above we see that at time $n+1$, the transmitter sends
the error term pertaining to the MMSE estimate of $X_n$ from
$Y_n$, scaled to match the permissible input power $P$. In fact,
it can be verified either by directly or using the equivalence
stated in Theorem \ref{thrm:recursion} that $X_{n+1}$ is the scaled
MMSE term of $X_n$ given the entire output sequence $Y^n$.
Therefore, the posterior matching scheme in this case is an
infinite-horizon, variable-rate variant of the Schalkwijk-Kailath
scheme. This variant is in fact even somewhat simpler than the
original scheme~\cite{Schalkwijk2}, since the initial matching step of the random message
point makes transmission start at a steady-state. The fundamental difference between the posterior matching principle and the Schalkwijk-Kailath ``parameter estimation'' approach in a non-Gaussian setting, is now evident. According to Schalkwijk-Kailath one should transmit a scaled linear MMSE term given past observations, which is \textit{uncorrelated} with these observations but \textit{not independent} of them as dictated by
the posterior matching principle; the two notions thus coincide only in the AWGN case. In fact, it can be shown that following the Schalkwijk-Kailath approach when the additive noise is not Gaussian results in achieving only the corresponding "Gaussian equivalent" capacity, see Example \ref{ex:AWGN_mis}.

\end{example}

\begin{example}[\it BSC]\label{ex:BSC}
Let $P_{Y|X}$ be a BSC with crossover probability $p$, and set a capacity achieving input distribution $X\sim{\rm Bernoulli}\left(\frac{1}{2}\right)$,
i.e., $f_X(x)=\frac{1}{2}\left(\delta(x)+\delta(x-1)\right)$. We now derive the posterior matching scheme for this setting, and
show it reduces to the Horstein scheme~\cite{horstein}. The conditions of Theorem \ref{thrm:recursion} are not satisfied since
the input distribution is discrete, and we therefore use the original non-recursive representation (\ref{eq:gn_func}) for now. It is
easy to see that the matching step $F_X^{-1}$ acts as a quantizer above/below $\frac{1}{2}$, and so we get
\begin{equation*}
X_{n+1} = F_X^{-1}\circ F_{\sst\Theta_0|Y^n}(\Theta_0|Y^n) = \left
\{\begin{array}{cc} 0 & \Theta_0<{\rm
median}\{f_{\sst\Theta_0|Y^n}(\theta|Y^n)\}
\\ 1 & o.w.
\end{array} \right.
\end{equation*}
which is precisely the Horstein scheme. The posterior matching
principle is evident in this case, since slicing the posterior
distribution at its median results in an input $X_{n+1}\sim{\rm Bernoulli}\left(\frac{1}{2}\right)$ given any possible output $Y^n=y^n$, and is
hence independent of $Y^n$ and ${\rm Bernoulli}(\frac{1}{2})$-distributed. We return to the BSC example later in this section, after we develop the necessary tools to provide an alternative (and more useful) recursive representation for the Horstein scheme.
\end{example}

\begin{example}[\it Uniform Input/Noise]\label{ex:uniform}
Let $ P_{Y|X}$ be an additive noise channel with noise uniformly
distributed over the unit interval, and set the input
$X\sim\m{U}$, i.e., uniform over the unit interval as well. Let us
derive the posterior matching scheme in this case. It is easy to
verify that the inverse channel's p.d.f. is given by
\begin{equation*}
f_{X|Y}(x|y) = \left\{\begin{array}{lc} y^{-1}\ind_{(0,y)}(x) &
y\in(0,1]
\\ (2-y)^{-1}\ind_{(y-1,1)}(x) & y\in(1,2) \end{array}\right.
\end{equation*}
Since the conditions of Theorem \ref{thrm:recursion} are
satisfied, we can use the recursive representation. We note that
since the input distribution is $\m{U}$, the matching step is
trivial and the posterior matching kernel is given by
\begin{equation}\label{eq:uniform_kernel}
F_X^{-1}\circ F_{X|Y}(x|y) = F_{X|Y}(x|y) = \left\{\begin{array}{lc} \frac{x}{y}\cdot\ind_{(0,y)}(x)+\ind_{[y,\infty)}(x) & y\in(0,1]
\\ \frac{x-y+1}{2-y}\cdot\ind_{(y-1,1)}(x)+\ind_{[1,\infty)}(x) &
y\in(1,2)
\end{array}\right.
\end{equation}
and therefore the posterior matching scheme is given by
\begin{equation}\label{eq:uniform_scheme}
X_1 = \Theta_0\,,\; X_{n+1} =
\frac{X_n}{Y_n}\cdot\ind_{(0,1]}(Y_n) +
\frac{X_n-Y_n+1}{2-Y_n}\cdot\ind_{(1,2)}(Y_n)
\end{equation}

The above has in fact a very simple interpretation. The desired
input distribution is uniform, so we start by transmitting the
message point $X_1=\Theta_0$. Then, given $Y_1$ we determine the
range of inputs that could have generated this output value, and
find an affine transformation that stretches this range to fill
the entire unit interval. Applying this transformation to $X_1$
generates $X_2$. We now determine the range of possible inputs
given $Y_2$, and apply the corresponding affine transformation to
$X_2$, and so on. This is intuitively appealing since what we do
in each iteration is just \textit{zoom-in} on the remaining
uncertainty region for $\Theta_0$. Since the posterior
distribution is always uniform, this zooming-in is linear.

The posterior distribution induced by this transmission strategy
is uniform in an ever shrinking sequence of intervals. Therefore,
a zero-error variable-rate decoding rule would be to simply decode
at time $n$ the (random) maximal interval $J_n$ within which the posterior
is uniform. The size of that interval is
\begin{equation*}
|J_n| = \prod_{k\in J}Y_k \prod_{k\not\in J}(2-Y_k)
\end{equation*}
where $J = \left\{k\,:\,1\leq k\leq n\,,Y_k<1\right\}$. Denoting the
channel noise sequence by $Z_n\sim P_Z$, the corresponding rate is
{\allowdisplaybreaks
\begin{align*}
R_n &= -\frac{1}{n}\,\log|J_n| = \frac{1}{n}\sum_{k\in J}
\log{\frac{1}{Y_k}} + \frac{1}{n}\sum_{k\not\in
J}\log\frac{1}{2-Y_k} = \frac{1}{n}\sum_{k=1}^n \log \frac{f_{X|Y}(X_k|Y_k)}{f_X(X_k)}
\\
&= \frac{1}{n}\sum_{k=1}^n \log \frac{ f_{Z}(Z_k)}{f_Y(Y_k)}
\quad\tendsto{n}{\infty} \quad \Expt\log{f_Z(Z)}-\Expt\log{f_Y(Y)} = I(X;Y) = \frac{1}{2}\log{e} \qquad \text{a.s.}
\end{align*}}
where we have used the SLLN for the i.i.d. sequences $Z^n,Y^n$.
Therefore, in this simple case we were able to directly show that
the posterior matching scheme, in conjunction with a simple
variable rate decoding rule, achieves the mutual information with
zero error probability. In the sequel, the achievability of the
mutual information and the tradeoff between rate, error
probability and transmission period obtained by the posterior
matching scheme are derived for a general setting. We then revisit
this example and provide the same results as above from this more
general viewpoint.
\end{example}

\begin{example}[\it Exponential Input/Noise]\label{ex:exp}
Consider an additive noise channel $ P_{Y|X}$ with $\sim{\rm
Exponential(1)}$ noise, and set the input $X\sim{\rm
Exponential(1)}$ as well. This selection is not claimed to be capacity achieving under any reasonable input constraints, yet it is instructive to study due to the simplicity of the resulting scheme. We will return to the exponential noise channel in Example \ref{ex:exp_noise_mean} after developing the necessary tools, and analyze it using the capacity achieving distribution under an input mean constraint.

It is easy to verify that for the above simple selection, the input given the output is uniformly distributed, i.e., the inverse channel p.d.f./c.d.f. are given by
\begin{equation*}
f_{X|Y}(x|y) = \frac{1}{y}\cdot\ind_{(0,y)}(x)\,,\quad F_{X|Y}(x|y) = \frac{x}{y}\cdot\ind_{(0,y)}(x) +
\ind_{[y,\infty)}(x)
\end{equation*}
The input's inverse c.d.f. is given by
\begin{equation*}
F_X^{-1}(s) = \ln\left(\frac{1}{1-s}\right)
\end{equation*}
Therefore, the posterior matching kernel is given by
\begin{equation}\label{eq:exp_kernel}
F_X^{-1}\circ F_{X|Y}(x|y) = \ln\left(\frac{y}{y-x}\right)
\end{equation}
and the posterior matching scheme in this case is simply given by
\begin{equation}\label{eq:exp_scheme}
X_1 = \ln\left(\frac{1}{1-\Theta_0}\right)\,,\quad X_{n+1} =
\ln\left(\frac{Y_n}{Y_n-X_n}\right)
\end{equation}
\end{example}
\vspace{0.5cm}
\begin{figure}
  \begin{center}
    \leavevmode
    \epsfig{file=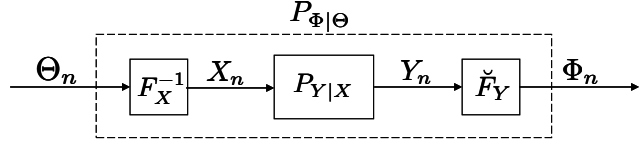, scale = 0.5}
    \caption{The normalized channel $P_{\Phi|\Theta}$}
    \label{fig:norm_chan}
  \end{center}
\end{figure}

\subsection{The Normalized Channel}\label{subsec:norm_chan}

The recursive representation provided in Theorem \ref{thrm:recursion}
is inapplicable in many interesting cases, including DMCs in
particular. In order to treat discrete, continuous and mixed
alphabet inputs/channels within a common framework, we define for any
input/channel pair $(P_X, P_{Y|X})$ a corresponding
\textit{normalized channel} $P_{\Phi|\Theta}$ with $\ui$ as a
common input/output alphabet, and a uniform input distribution
$\Theta\sim\m{U}$. The normalized channel is obtained by viewing
the matching operator $F^{-1}_X(\cdot)$ as part of the original
channel, and applying the output c.d.f. operator $F_Y(\cdot)$ to
the channel's output, with the technical exception that whenever
$F_Y(\cdot)$ has a jump discontinuity the output is randomly
selected uniformly over the jump span.\footnote{The output mapping is of a lesser importance, and is introduced mainly to provide a common framework.} This is depicted in Figure
\ref{fig:norm_chan}, where $\breve{F}_Y(\cdot)$ stands for the
aforementioned possibly random mapping. This construction is most
simply formalized by
\begin{equation}\label{def:formal_norm_chan}
P_{Y|\Theta}(\cdot|\theta) =
P_{Y|X}(\cdot|F_X^{-1}(\theta)),\;\;\Phi = F_Y(Y) -P_Y(Y)\cdot
\Lambda
\end{equation}
where $\Theta\sim\m{U}$, and $\Lambda\sim\m{U}$ is statistically independent of $(\Theta,Y)$.

\begin{lemma}[\it Normalized Channel Properties]\label{lem:norm_chan_prop} Let
$(P_\Theta,P_{\Phi|\Theta})$ be the normalized input/channel pair
corresponding to the pair $(P_X, P_{Y|X})$. The following
properties are satisfied:
\begin{enumerate}[(i)]
\item $\Phi\sim\m{U}$, i.e., $P_{\Phi|\Theta}$ preserves the uniform distribution over the unit interval. \label{claim:norm1}

\item The mutual information is preserved, i.e., \label{claim:norm2}
\begin{equation*}
I(\Theta;\Phi)=I(X;Y)
\end{equation*}

\item The joint distribution $P_{\Theta\Phi}$ is proper.\label{claim:norm3}

\item The normalized kernel $F_{\Theta|\Phi}(\theta|\phi)$ is continuous in $\theta$ for $P_\Phi$-a.a. $\phi\in\ui$. \label{claim:norm4}
\end{enumerate}
\end{lemma}

\begin{proof}
\begin{enumerate}[(i)]
\item By Lemma \ref{lem:matching_trans} claim (\ref{item:match1})
we have $F_X^{-1}(\Theta)\sim P_X$, and so $Y\sim P_Y$ in
(\ref{def:formal_norm_chan}). The result now follows from Lemma
\ref{lem:matching_trans} claim (\ref{item:match2}).

\item An easy exercise using the relations in Figure
\ref{fig:norm_chan}, and noting that $X,Y$ are always uniquely
recoverable from $\Theta,\Phi$ respectively.

\item See Appendix \ref{app:lemmas}.

\item Follows easily from (\ref{claim:norm3}).
\end{enumerate}

\end{proof}

The posterior matching scheme over the normalized channel with a
uniform input, is given by
\begin{equation*}
\bar{g}_{n+1}(\theta,\phi^n) = F_{\Theta_0|\Phi^n}\,(\theta|\phi^n)
\end{equation*}
The properties of the normalized channel allows for a unified
recursive representation of the above scheme via the inverse
normalized channel $P_{\Theta|\Phi}$ corresponding to $( P_\Theta,
P_{\Phi|\Theta}) = (\m{U}, P_{\Phi|\Theta})$, i.e., in terms of the \textit{normalized posterior matching kernel} $F_{\Theta|\Phi}$.
\begin{theorem}[\it Recursive representation II]\label{thrm:recursion2}
The posterior matching scheme for the normalized channel is given by the recursive relation:
\begin{equation}\label{eq:gn_rec_norm}
\bar{g}_1(\theta) = \theta ,\;\; \bar{g}_{n+1}(\theta|\phi^n) = F_{\Theta|\Phi}(\cdot|\phi_n)\circ \bar{g}_n(\theta|\phi^{n-1})
\end{equation}
The corresponding sequence of input/output pairs $\left\{(\Theta_n,\Phi_n)\right\}_{n=1}^\infty$ constitutes a Markov chain over a state space
$\isupp(\theta,\Phi)\subseteq\ui^2$, with an invariant distribution $ P_{\Theta\Phi}$, and satisfy the recursion rule
\begin{equation}\label{eq:gn_rec_normX}
\Theta_1 = \Theta_0\,,\quad \Theta_{n+1} = F_{\Theta|\Phi} (\Theta_n|\Phi_n)
\end{equation}
Furthermore, (\ref{eq:gn_rec_normX}) is equivalent to the posterior matching scheme (\ref{eq:gn_funcX}) in the sense that the distribution of the sequence $\left\{F_X^{-1}(\Theta_n),F_Y^{-1}(\Phi_n)\right\}_{n=1}^\infty$ coincides with the distribution of the sequence $\{(X_n,Y_n)\}_{n=1}^\infty$.
\end{theorem}
\begin{proof}
By Lemma \ref{lem:norm_chan_prop} the joint distribution $P_{\Theta\Phi}$ is proper, hence Theorem \ref{thrm:recursion} is applicable and the recursive representations and Markovity follow immediately. Once again, taking the state space to be $\isupp(\Theta,\Phi)$ and not $\supp(\Theta,\Phi)$ is
artificial and is done for reasons of mathematical convenience, to avoid having the trivial invariant distributions $P_0\times
P_{\Phi|\Theta}(\cdot|0)$ and $P_1\times P_{\Phi|\Theta}(\cdot|1)$, where $P_0(0)=1,P_1(1)=1$. The distribution $P_{\Theta\Phi}$ is invariant by construction, and the equivalence to the original scheme is by definition.
\end{proof}

In the sequel, an initial point for the aforementioned Markov
chain will be given by a fixed value $\theta_0\in\ui$ of the
message point only\footnote{This is an abuse of notations, since
an initial point is properly given by a pair
$(\theta_1,\phi_1)$. However, it can be justified since
$\Theta_1=\Theta_0$ and $\Phi_1$ is generated via a memoryless
channel. Hence, any statement that holds for a.a/all initial
points $(\theta_1,\phi_1)$ also holds in particular for a.a./all
$\theta_0$.}. Notice also that the Theorem above reveals an
interesting fact: Whenever $F_X^{-1}$ is not injective, the sequence of input/output pairs pertaining to the
original posterior matching scheme (\ref{eq:gn_funcX}) is a \textit{hidden
Markov process}. In particular, this is true for the BSC and the Horstein scheme.

\vspace{9pt}\noindent{\bf Example \ref{ex:BSC} (BSC, continued)}.
The normalized channel's p.d.f. corresponding to a BSC with
crossover probability $p$ and a ${\rm Bernoulli}\left(\frac{1}{2}\right)$ input
distribution is given by $f_{\Phi|\Theta}(\phi|\theta)=2(1-p)$ when $\theta,\phi$ are either both smaller or both larger than $\frac{1}{2}$, and $f_{\Phi|\Theta}(\phi|\theta)=2p$ otherwise. Following Theorem \ref{thrm:recursion2} and simple manipulations, the
corresponding normalized posterior matching kernel is given by
\begin{align}\label{eq:BSC_kernel}
F_{\Theta|\Phi}(\theta|\phi) = \left\{\hspace{-4pt}
    \begin{array}{lc}
    2(1-p)\theta & \theta\in(0,\frac{1}{2}), \phi\in(0,\frac{1}{2}) \\
    2p\theta+(1-2p) & \theta\in[\frac{1}{2},1), \phi\in(0,\frac{1}{2})\\
    2p\theta & \theta\in(0,\frac{1}{2}), \phi\in[\frac{1}{2},1) \\
    2(1-p)\theta -(1-2p) & \theta\in[\frac{1}{2},1),\phi\in[\frac{1}{2},1)\\
    \end{array}\right.
\end{align}
and for a fixed $\phi$ is supported on two functions of $\theta$, depending on whether
$\phi\,{\scriptstyle\lessgtr}\,\frac{1}{2}$ which corresponds
to $y=0,1$ in the original discrete setting, see Figure
\ref{fig:BSC}. Therefore, the posterior matching scheme (which is
equivalent to the Horstein scheme in this case) is given by the
following recursive representation:
\begin{equation*}
\Theta_1 = \Theta_0 \,,\qquad \Theta_{n+1} \hspace{-1pt}=\hspace{-1pt} \left\{\hspace{-4pt}
    \begin{array}{ll}
    2(1-p)\Theta_n & \Theta_n\in(0,\frac{1}{2}), \Phi_n\in(0,\frac{1}{2}) \\
    2p\Theta_n+(1-2p) & \Theta_n\in[\frac{1}{2},1), \Phi_n\in(0,\frac{1}{2})\\
    2p\Theta_n & \Theta_n\in(0,\frac{1}{2}), \Phi_n\in[\frac{1}{2},1) \\
    2(1-p)\Theta_n -(1-2p) & \Theta_n\in[\frac{1}{2},1),\Phi_n\in[\frac{1}{2},1)\\
    \end{array}\right.
\end{equation*}
The hidden Markov process describing the original Horstein scheme is recovered from the above by setting
\begin{equation*}
X_k=F_X^{-1}(\Theta_k) = \ind_{[\frac{1}{2},1)}(\Theta_k)\,,\qquad Y_k=F_X^{-1}(\Phi_k) =
\ind_{[\frac{1}{2},1)}(\Phi_k)
\end{equation*}
\begin{figure}
  \begin{center}
    \leavevmode
    \epsfig{file=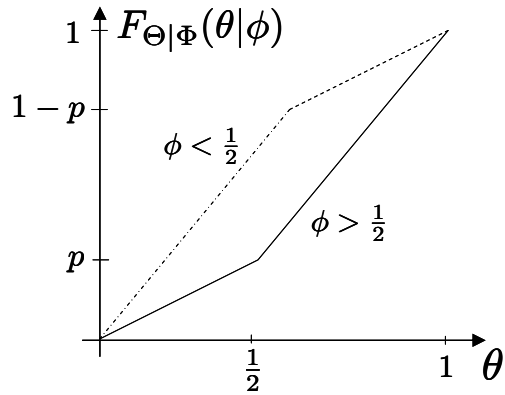, scale = 0.5}
    \caption{The BSC normalized posterior matching kernel}
    \label{fig:BSC}
  \end{center}
\end{figure}

\begin{example}[\it The binary erasure channel (BEC)]
The binary erasure channel is defined over the input alphabet $\m{X} = \{0,1\}$ and the output alphabet $\m{Y} = \{0,1,2\}$. Given any input, the output is equal to that input with probability $p$, and equal to $2$ with probability $1-p$. Using the capacity achieving distribution $P_X = {\rm Bernoulli}\left(\frac{1}{2}\right)$, it is easy to see from the non-recursive representation (\ref{eq:gn_func}) that the posterior matching scheme in this case is exactly the simple repetition rule -- transmit the first bit of $\Theta_0$ until it is correctly received, then continue to the next bit and so on. This scheme clearly achieves the capacity $1-p$. The recursive representation w.r.t. the normalized channel is very simple and intuitive here as well. The normalized posterior matching kernel is supported on three functions -- the identity function corresponding to the erasure output $2$, and the functions $2\theta,2\theta-1$ that correspond to the outputs $0,1$ respectively.
\end{example}

\begin{example}[\it General DMC]\label{ex:gen_DMC}
The case where $P_{Y|X}$ is a DMC and $P_X$ is a corresponding discrete input distribution is a simple extension of the BSC/BEC settings. The normalized posterior matching kernel is supported over a finite number of $|\m{Y}|$ continuous functions, which are all quasi-affine relative to a fixed partition of the unit interval into subintervals corresponding to the input distribution. Precisely, for any $x\in\m{X}$ the normalized posterior matching kernel evaluated at $\theta=F_X(x)$ is given by
\begin{equation}\label{eq:DMC_kernel}
F_{\Theta|\Phi}(F_X(x)|\phi) = F_{X|Y}(x|F_Y^{-1}(\phi))
\end{equation}
and by a linear interpolation in between these points. Hence, the corresponding kernel slopes are given by $\frac{P_{X|Y}(x|y)}{P_X(x)}$.
\end{example}

\begin{example}[\it Exponential noise, input mean constraint]\label{ex:exp_noise_mean}
Consider an additive noise channel $P_{Y|X}$ with $\sim{{\rm Exponential}(b)}$ noise, but now instead of arbitrarily assuming an exponential input distribution as in Example \ref{ex:exp}, let us impose an input mean constraint $(x,a)$ , i.e.,
\begin{align*}
\lim_{n\rightarrow\infty}n^{-1}\sum_{k=1}^n X_k \leq
a\quad\text{a.s.}
\end{align*}
The capacity achieving distribution under this input constraint was determined in~\cite{verdu_exp96} to be a \textit{mixture} of a deterministic distribution and an exponential distribution, with the following generalized p.d.f.:
\begin{equation*}
f_X(x) = \frac{b}{a+b}\,\delta(x) + \frac{a}{(a+b)^2}\exp\left(-\frac{x}{a+b}\right)
\end{equation*}
Under this input distribution the output is $Y\sim{\rm Exponential}(a+b)$, and the capacity can be expressed in closed form
\begin{equation*}
C = I(X;Y) = \log\left(1+\frac{a}{b}\right)
\end{equation*}
in a remarkable resemblance to the AWGN channel with an input power constraint. Interestingly, in this case the posterior matching scheme can also be written in closed form, and as stated later, also achieves the channel capacity under the input mean constraint.

To derive the scheme, we must resort to the normalized representation since the input distribution is not proper. The input's inverse c.d.f. and the output's c.d.f. are given by
\begin{equation*}
F_X^{-1}(\theta) = (a+b)\ln\left(\frac{a}{(a+b)(1-\theta)}\right)\cdot\ind_{[\frac{b}{a+b},1)}(\theta)\,,\quad F_Y(y) = 1-\exp\left(-\frac{y}{a+b}\right)
\end{equation*}
Using the normalized representation and practicing some algebra, we find that the normalized posterior matching kernel is given by
\begin{equation}\label{eq:exp_mean_kernel}
F_{\Theta|\Phi}(\theta|\phi) = (1-\phi)^{\frac{a}{b}}\left(\frac{a+b}{b}\cdot\theta\cdot \ind_{(0,\frac{b}{a+b})}(\theta) +\left(\frac{a}{a+b}\cdot\frac{1}{1-\theta}\right)^{\frac{a}{b}}\ind_{[\frac{b}{a+b},1-\frac{a(1-\phi)}{a+b})}(\theta)\right) + \ind_{(1-\frac{a(1-\phi)}{a+b},\infty)}(\theta)
\end{equation}
Thus the posterior matching scheme in this case is given by
\begin{equation}\label{eq:exp_mean_scheme}
\Theta_1 = \Theta_0\,,\quad \Theta_{n+1} = \left\{\begin{array}{cc}\frac{a+b}{b}\cdot\Theta_n\cdot (1-\Phi_n)^{\frac{a}{b}}& \Theta_n \leq \frac{b}{a+b} \\ \left(\frac{a}{a+b}\cdot \frac{1-\Phi_n}{1-\Theta_n}\right)^{\frac{a}{b}} & \Theta_n > \frac{b}{a+b}\end{array}\right.
\end{equation}
where the original channel's input/output pairs are given by
\begin{equation*}
X_n = (a+b)\ln\left(\frac{a}{(a+b)(1-\Theta_n)}\right)\ind_{[\frac{b}{a+b},1)}(\Theta_n)\,,\quad Y_n = (a+b)\ln\frac{1}{1-\Phi_n}
\end{equation*}
and constitute a hidden Markov process. Note that since we have $\Theta_n\in(0,1-\frac{a(1-\Phi_n)}{a+b})$ a.s., then $\Theta_{n+1}\in\ui$ a.s. and we need not worry about the rest of the thresholds appearing in (\ref{eq:exp_mean_kernel}).
\end{example}

\section{Regularity Conditions for Input/Channel Pairs}\label{sec:fam}
In Section~\ref{sec:main_result_memoryless}, we prove the optimality of the posterior matching scheme. However, to that end we first need to introduce several regularity conditions, and define some well behaved families of input/channel pairs.

For any fixed $\phi\in\ui$, define $\theta^-_\phi$ and $\theta^+_\phi$ to be the unique solutions of
\begin{equation*}
F_{\Theta|\Phi}(\theta^-_\phi|\phi) = \frac{1}{2} F_{\Theta|\Phi}(\theta|\phi)\,,\qquad
1-F_{\Theta|\Phi}(\theta^+_\phi|\phi) = \frac{1}{2}\left(1-F_{\Theta|\Phi}(\theta|\phi)\right)
\end{equation*}
respectively. For any $\eps>0$, define the \textit{left-$\eps$-measure} $\ssc{-}P^\eps_{\Phi|\Theta}(\cdot|\theta)$ of $P_{\Phi|\Theta}(\cdot|\theta)$ to have a density $\ssc{-}f^\eps_{\Phi|\Theta}$ given by
\begin{equation*}
\ssc{-}f^\eps_{\Phi|\Theta}(\phi|\theta)\dfn\inf_{\xi\in
J_{\st\eps}^-(\phi,\theta)} f_{\Phi|\Theta}(\phi|\xi),
\end{equation*}
where the interval $J_\eps^-(\phi,\theta)$ is defined to be
\begin{equation}\label{eq:eps_channel}
J_\eps^-(\phi,\theta) \dfn
\big{(}\max(\theta^-_\phi,\theta-\eps),\theta\big{)}
\end{equation}
Note that the left-$\eps$-measure is not a probability distribution since in general
$\ssc{-}P^\eps_{\Phi|\Theta}(\ui|\theta)<1$. Similarly define \textit{right-$\eps$-measure}
$\ssc{+}P^\eps_{\Phi|\Theta}(\cdot|\theta)$ of $P_{\Phi|\Theta}(\cdot|\theta)$ to have a density
$\ssc{+}f^\eps_{\Phi|\Theta}$ given by
\begin{equation*}
\ssc{+}f^\eps_{\Phi|\Theta}(\phi|\theta)\dfn\inf_{\xi\in
J_{\st\eps}^+(\phi,\theta)} f_{\Phi|\Theta}(\phi|\xi)
\end{equation*}
where the interval $J_\eps^+(\phi,\theta)$ is defined to be
\begin{equation*}
J_\eps^+(\phi,\theta) \dfn \big{(}\theta,\min(\theta+\eps,\theta^+_\phi)\big{)}
\end{equation*}
Note that ${\displaystyle\lim_{\eps\rightarrow 0}}\ssc{-}f^\eps_{\Phi|\Theta} =
{\displaystyle\lim_{\eps\rightarrow 0}}\ssc{+}f^\eps_{\Phi|\Theta} = f_{\Phi|\Theta}$ a.e. over $\ui^2$. Following these definitions, an input/channel pair $(P_X, P_{Y|X})$ is said to be \textit{regular}, if the corresponding normalized channel satisfies
\begin{equation*}
\inf_{\eps>0}\left[D(P_{\Phi|\Theta}\,\|\,\ssc{-}P^\eps_{\Phi|\Theta}\,|\,P_\Theta)+D(P_{\Phi|\Theta}\,\|\,\ssc{+}P^\eps_{\Phi|\Theta}\,|\,P_\Theta)\right]
<\infty
\end{equation*}
Loosely speaking, the regularity property guarantees that the \textit{sensitivity} of the channel law $ P_{Y|X}$ to input perturbations is not too high, or is at least attenuated by a proper selection of the input distribution $ P_X$. Regularity is satisfied in many interesting cases, as demonstrated in the following Lemma.
\begin{lemma}\label{lem:reg_cond}
Each of the following conditions implies that the input/channel pair
$( P_X, P_{Y|X})$  is regular:
\begin{enumerate}[(i)]
\item $h(\Theta,\Phi)$ is finite, $\isupp(\Theta,\Phi)$ is convex in the $\theta$-direction, and $f_{\Theta\Phi}$ is bounded away from zero over $\isupp(\Theta,\Phi)$. \label{cond:reg_bounded_away}

\item $P_{XY}$ is proper, $\;\isupp(X,Y)$ is convex in the $x$-direction, $f_X$ is bounded, and $f_{X|Y}$ has a uniformly bounded max-to-min ratio, i.e.,
\begin{equation*}
\sup_{y\in\isupp(Y)}\left(\frac{\sup_{x\in\isupp(X|Y=y)}f_{X|Y}(x|y)}{\inf_{x\in\isupp(X|Y=y)}f_{X|Y}(x|y)}\right)
< \infty
\end{equation*}\label{cond:reg_min_max}

\item $P_{XY}$ is proper, and $f_{X|Y}(x|y)$ is unimodal with a regular tail and a bounded variance, uniformly over $y\in\isupp{(Y)}$.\label{cond:reg_unimodal}

\item $P_{Y|X}$ is a DMC with nonzero transition probabilities.\label{cond:reg_dmc}

\end{enumerate}
\end{lemma}
\begin{proof}
See Appendix \ref{app:misc}.
\end{proof}

For an input/channnel pair $( P_X, P_{Y|X})$, define the following set of properties:
\begin{enumerate}[\bf(\textsl{A}1)]
\item \textit{$(P_X, P_{Y|X})$ is regular.} \label{cond:omega_reg}
\item \textit{The invariant distribution $P_{\Theta\Phi}$ for the Markov chain $\{(\Theta_n,\Phi_n)\}_{n=1}^\infty$, is ergodic.}\label{cond:omega_inv}
\item \label{cond:omega_fixedp} \textit{$F_{\Theta|\Phi}$ is \textit{fixed-point free}, i.e., for any $\theta\in\ui$.}
\begin{equation}\label{eq:univ_fixed}
\Prob\left(F_{\Theta|\Phi}(\theta|\Phi)=\theta\right) < 1.
\end{equation}
\item \label{cond:omega_cap_ach}\textit{$P_X$ achieves the unconstrained capacity over $P_{Y|X}$, i.e., $I(X;Y) = C(P_{Y|X})$}.\footnote{Since an input/channel pair has finite mutual information, (\textsl{B}\ref{cond:omega_cap_ach}) implies that $C(P_{Y|X})<\infty$. The unconstrained capacity is finite for discrete input and/or output channels, but can be finite under other input alphabet constraints (e.g., an amplitude constraint).}
\end{enumerate}

The following is easily observed.
\begin{lemma}\label{lem:egr_implies_fixed_free}
(\textsl{A}\ref{cond:omega_inv}) $\Rightarrow$ (\textsl{A}\ref{cond:omega_fixedp}).
\end{lemma}
\begin{proof}
See proof of Lemma \ref{lem:fixed_point_no_rate}.
\end{proof}

Let $\Omega_A$ be the family of all input/channel pairs satisfying properties (\textsl{A}\ref{cond:omega_reg}) and (\textsl{A}\ref{cond:omega_inv}). Let $\Omega_B$ be the family of all input/channel pairs satisfying properties (\textsl{A}\ref{cond:omega_reg}), (\textsl{A}\ref{cond:omega_fixedp}) and (\textsl{A}\ref{cond:omega_cap_ach}). In the sequel, we show that for members in $\Omega_A\cup \Omega_B$ the corresponding posterior matching scheme achieves the mutual information. However, while Lemma \ref{lem:reg_cond} provides means to verify the regularity Property (\textsl{A}\ref{cond:omega_reg}), and Properties (\textsl{A}\ref{cond:omega_fixedp}) and (\textsl{A}\ref{cond:omega_cap_ach}) are easy to check, the ergodicity property (\textsl{A}\ref{cond:omega_inv}) may be difficult to verify in general. Therefore, we introduce the following more tractable property:

\begin{enumerate}[\bf(\textsl{A}1)]
\setcounter{enumi}{4}

\item \textit{$f_{XY}$ is bounded and continuous over $\,\isupp(X,Y)$, where the latter is connected and convex in the $y$-direction.}\label{cond:omegac_proper}

\end{enumerate}
We now show that (\textsl{A}\ref{cond:omega_fixedp}) and (\textsl{A}\ref{cond:omegac_proper}) together imply a stronger version of (\textsl{A}\ref{cond:omega_inv}). In fact, to that end a weaker version of (\textsl{A}\ref{cond:omega_fixedp}) is sufficient, which we state (for convenience) in terms of the non-normalized kernel:
\begin{enumerate}[\bf(\textsl{A}1${}^*$)]
\setcounter{enumi}{2}
\item \label{cond:omegac_fixedp} \textit{For any $x\in\isupp (X)$ there exists $y\in\isupp(Y)$, such that $F^{-1}_X\circ F_{X|Y}(x|y)\neq x$.}
\end{enumerate}

\begin{lemma}\label{lem:Omegac_in_Omega}
(\textsl{A}\ref{cond:omegac_fixedp}${}^*$) $\wedge$ (\textsl{A}\ref{cond:omegac_proper}) $\Rightarrow$ $\{(\Theta_n,\Phi_n)\}_{n=1}^\infty$ is p.h.r. and aperiodic $\Rightarrow$  (\textsl{A}\ref{cond:omega_inv}).
\end{lemma}
\begin{proof}
For the first implication, see Appendix \ref{app:pointwise}. The second implication is immediate since p.h.r. implies in particular a unique invariant distribution, which is hence ergodic.
\end{proof}
Following that, let us define $\Omega_C$ to be the family of all input/channel pairs $(P_X, P_{Y|X})$ satisfying properties (\textsl{A}\ref{cond:omega_reg}), (\textsl{A}\ref{cond:omegac_fixedp}${}^*$) and (\textsl{A}\ref{cond:omegac_proper}).

\begin{corollary}
$\Omega_C\subset\Omega_A$.
\end{corollary}
Turning to the discrete case, let $(P_X,P_{Y|X})$  be an input/DMC pair. Without loss of generality, we will assume throughout that $\min_{x\in\m{X}}P_X(x)>0$, as otherwise the unused input can be removed. Define the following set of properties:
\begin{enumerate}[\bf(\textsl{B}1)]
\item \label{cond:omega_dm_nzI} $\displaystyle{\min_{x\in\m{X},y\in\m{Y}}}P_{Y|X}(y|x)>0$.
\item \label{cond:omega_dm_dom} At least one of the following holds:
    \begin{enumerate}[(i)]
    \item There exists some $y\in\m{Y}$ with $P_Y(y)>0$, such that either $P_X\dom P_{X|Y}(\cdot|y)$ or $P_{X|Y}(\cdot|y)\dom P_X$.
    \item There exist some $y_0,y_1\in\m{Y}$ with $P_Y(y_0)>0,P_Y(y_1)>0$, such that $P_{X|Y}(\cdot|y_0)\dom P_{X|Y}(\cdot|y_1)$.
    \end{enumerate}
\item \label{cond:omega_dm_Q} $\forall x\in\m{X}$, $\exists y_0,y_1\in\m{Y}$ s.t. $0>\frac{\beta_0}{\beta_1}\not\in\RatF$, where\footnote{$\RatF$ is the set of rational numbers. Note that there always exists a pair for which $\frac{\beta_0}{\beta_1}<0$, but the quotient is not necessarily irrational.}
    \begin{equation*}
    \beta_i\dfn \log\left(\frac{P_{X|Y}(x|y_i)}{P_X(x)}\right),\quad i\in\{0,1\}.
    \end{equation*}

\end{enumerate}
\begin{lemma}\label{lem:dmc_prop}
Let $(P_X,P_{Y|X})$ be an input/DMC pair. Then:
\begin{enumerate}[(i)]
\item \label{claim:dmc1}(\textsl{B}\ref{cond:omega_dm_nzI}) $\Rightarrow$ (\textsl{A}\ref{cond:omega_reg}).
\item \label{claim:dmc2}(\textsl{B}\ref{cond:omega_dm_dom}) $\Rightarrow$ (\textsl{A}\ref{cond:omega_fixedp}).
\item \label{claim:dmc3}(\textsl{B}\ref{cond:omega_dm_nzI}) $\wedge$ (\textsl{B}\ref{cond:omega_dm_Q}) $\wedge$ (\textsl{A}\ref{cond:omega_fixedp}) $\Rightarrow$ (\textsl{A}\ref{cond:omega_inv}).
\item \label{claim:dmc4}$|\m{X}|=2$ $\Rightarrow$ (\textsl{B}\ref{cond:omega_dm_dom}).
\item \label{claim:dmc5}(\textsl{B}\ref{cond:omega_dm_nzI}) $\wedge$ $I(X;Y)>0$ $\Rightarrow$ there exists an equivalent pair $(P_{X^*}, P_{Y^*|X^*})$ satisfying (\textsl{B}\ref{cond:omega_dm_nzI}) $\wedge$ (\textsl{B}\ref{cond:omega_dm_dom}).
\item \label{claim:dmc6}For any $\eps>0$ there exists $P_X'$, such that $d_{TV}(P_X,P_X') < \eps$, and $(P_X', P_{Y|X})$ is an input/DMC pair satisfying (\textsl{B}\ref{cond:omega_dm_Q}).
\end{enumerate}
\end{lemma}
\begin{proof}
Claim (\ref{claim:dmc1}) follows immediately from condition (\ref{cond:reg_dmc}) of Lemma \ref{lem:reg_cond}. Claim (\ref{claim:dmc4}) holds since any two nonidentical binary distributions can be ordered by dominance. For the remaining claims, see Appendix \ref{app:lemmas}.
\end{proof}

\begin{remark}
The equivalent pair in Lemma \ref{lem:dmc_prop}, claim (\ref{claim:dmc5}), is obtained via an input permutation only, which is given explicitly in the proof and can be simply computed.
\end{remark}

\section{Achieving the Mutual Information}\label{sec:main_result_memoryless}
Our main theorem is presented in Subsection~\ref{subsec:main_res}, establishing the achievability of the mutual information via posterior matching for a large family of input/channel pairs. The examples of Section~\ref{sec:scheme} are then revisited, and the applicability of the theorem is verified in each. Subsection~\ref{subsec:main_proof} is dedicated to the proof of the Theorem.

\subsection{Main Result}\label{subsec:main_res}
\begin{theorem}[\it Achievability]\label{thrm:achv}
Consider an input/channel pair $(P_X, P_{Y|X})\in \Omega_A\cup\Omega_B$ (resp. $\Omega_C$). The corresponding posterior matching scheme with a fixed/variable rate optimal decoding rule, achieves (resp. pointwise achieves) any rate $R<I(X;Y)$ over the channel $P_{Y|X}$. Furthermore, if $(P_X, P_{Y|X})\in \Omega_A$ (resp. $\Omega_C$), then $R$ is achieved (resp. pointwise achieved) within an input constraint $(\eta,\Expt\eta(X))$, for any measurable $\eta:\m{X}\mapsto\RealF$ satisfying $\Expt|\eta(X)|<\infty$.
\end{theorem}

\noindent{\bf Example \ref{ex:AWGN} (AWGN,
continued)}. $P_{XY}$ is proper (jointly Gaussian), and the inverse channel's p.d.f. $f_{X|Y}(x|y)$ is Gaussian with a variance independent of $y$, hence by Lemma \ref{lem:regular_dist} condition (\ref{cond:exp_tail}) has a regular tail uniformly in $y$. Therefore, by condition (\ref{cond:reg_unimodal}) of Lemma \ref{lem:reg_cond}, the Gaussian input/AWGN channel pair $(P_X,P_{Y|X})$ is regular and Property (\textsl{A}\ref{cond:omega_reg}) is satisfied. It is easy to see that the linear posterior matching kernel (\ref{eq:AWGN_kernel}) is fixed-point free, and so Property (\textsl{A}\ref{cond:omegac_fixedp}$^*$) is satisfied as well. Finally, $f_{XY}$ is continuous and bounded over a $\RealF^2$ support, so Property (\textsl{A}\ref{cond:omegac_proper}) is also satisfied. Therefore $(P_X, P_{Y|X})\in\Omega_C$, and Theorem \ref{thrm:achv} verifies the well known fact that the Schalkwijk-Kailath scheme (pointwise) achieves any rate below the capacity $I(X;Y)=\frac{1}{2}\log\left(1+\SNR\right)$.

\vspace{9pt}\noindent{\bf Example \ref{ex:BSC} (BSC, continued)}. The pair of a ${\rm Bernoulli}\left(\frac{1}{2}\right)$ input $P_X$ and a BSC $P_{Y|X}$ with any nontrivial crossover probability $p\neq 0,1$, satisfies properties (\textsl{A}\ref{cond:omega_cap_ach}) and (\textsl{B}\ref{cond:omega_dm_nzI}). Properties (\textsl{A}\ref{cond:omega_reg}) and (\textsl{A}\ref{cond:omega_fixedp}) follow from claims (\ref{claim:dmc1}), (\ref{claim:dmc2}) and (\ref{claim:dmc4}) of Lemma \ref{lem:dmc_prop}. Hence $(P_X,P_{Y|X})\in\Omega_B$ and Theorem \ref{thrm:achv} implies that the posterior matching scheme, which coincides in this case with the Horstein scheme, indeed achieves the capacity $I(X;Y) = 1-h_b(p)$. \textit{This settles in the affirmative a longstanding conjecture}.

\begin{remark}
In the BSC Example above, it also holds (via Lemma \ref{lem:dmc_prop}, claim (\ref{claim:dmc3})) that $(P_X,P_{Y|X})\in\Omega_A$ for a.a. crossover probabilities $p$, except perhaps for the countable set $S=\{p: \frac{1+\log{p}}{1+\log{(1-p)}}\in\RatF\}$ where property (\textsl{B}\ref{cond:omega_dm_Q}) is not satisfied. In these cases the ergodicity property (\textsl{A}\ref{cond:omega_inv}) is not guaranteed, though this may be an artifact of the proof (see Remark \ref{rem:weaken_irr}). Therefore, although capacity is achieved for any $p$ (via $\Omega_B$), Theorem \ref{thrm:achv} guarantees the empirical distribution of the input sequence $X^n$ to approach $P_X$ only for $p\not\in S$. However, since $P_X$ is the \textit{unique} capacity achieving distribution, this sample-path property of the input sequence holds for $p\in S$ nonetheless (see Remark \ref{rem:non_erg}).
\end{remark}

\begin{remark}
Interestingly, for $p\in S$ the Horstein medians exhibit ``regular behavior'', meaning that any median point can always be returned to in a fixed number of steps. In fact, for the subset of $S$ where $\frac{1+\log{p}}{1+\log{(1-p)}} = -k$ for some positive integer $k\geq 2$, the Horstein scheme can be interpreted as a simple finite-state constrained encoder that precludes subsequences of more than $k$ consecutive $0$'s or $1$'s, together with an insertion mechanism repeating any erroneously received bit $k+1$ times. This fact was identified and utilized in~\cite{Schalkwijk71} to prove achievability in this special case.
\end{remark}

\noindent{\bf Example \ref{ex:uniform} (Uniform input/noise, continued)}. $P_{XY}$ is proper with a bounded p.d.f. over the convex support $\isupp(X,Y)=\ui\times(0,2)$, the marginal p.d.f. $f_X$ is bounded, and the inverse channel's p.d.f. is uniform hence has a bounded max-to-min ratio. Therefore, condition (\ref{cond:reg_min_max}) of Lemma \ref{lem:reg_cond} holds, and properties (\textsl{A}\ref{cond:omega_reg}) and (\textsl{A}\ref{cond:omegac_proper}) are satisfied. It is readily verified that the kernel (\ref{eq:uniform_kernel}) is fixed-point free, and so property (\textsl{A}\ref{cond:omegac_fixedp}$^*$) is satisfied as well. Therefore $(P_X, P_{Y|X})\in\Omega_C$, and Theorem \ref{thrm:achv} reverifies that the simple posterior matching scheme (\ref{eq:uniform_scheme}) pointwise achieves the mutual information $I(X;Y) = \frac{1}{2}\log{e}$, as previously established by direct calculation. In fact, we have already seen that (variable-rate) zero-error decoding is possible in this case, and in the next section we arrive at the same conclusion from a different angle.

\vspace{9pt}\noindent{\bf Example \ref{ex:exp} (Exponential input/noise, continued)}.
$P_{XY}$ is proper with a bounded p.d.f. over the convex support $\isupp(X,Y)=\RealF^+\times\RealF^+$, the marginal p.d.f. $f_X$ is bounded, and the inverse channel's p.d.f. is uniform hence has a bounded max-to-min ratio. Therefore, condition (\ref{cond:reg_min_max}) of Lemma \ref{lem:reg_cond} holds, and properties (\textsl{A}\ref{cond:omega_reg}) and (\textsl{A}\ref{cond:omegac_proper}) are satisfied. It is readily verified that the kernel (\ref{eq:exp_kernel}) is fixed-point free, and so property (\textsl{A}\ref{cond:omegac_fixedp}$^*$) is satisfied as well. Therefore $(P_X, P_{Y|X})\in\Omega_C$, and so by Theorem \ref{thrm:achv} the posterior matching scheme (\ref{eq:exp_scheme}) pointwise achieves the mutual information, which is this case is $I(X;Y) \approx 0.8327$.

\vspace{9pt}\noindent{\bf Example \ref{ex:gen_DMC} (General DMC, continued)}. It has already been demonstrated that the posterior matching scheme achieves the capacity of the BSC. We now show that the same holds true for a general DMC, up to some minor resolvable technicalities. Let $P_{Y|X}$ be a DMC with nonzero transition probabilities, and set $P_X$ to be capacity achieving (unconstrained). Hence properties (\textsl{B}\ref{cond:omega_dm_nzI}) and (\textsl{A}\ref{cond:omega_cap_ach}) are satisfied, and by Lemma \ref{lem:dmc_prop}, claim (\ref{claim:dmc1}), property (\textsl{A}\ref{cond:omega_reg}) holds as well. The corresponding posterior matching scheme in this case is equivalent to a generalized Horstein scheme, which was conjectured to achieve the unconstrained capacity when there are no fixed points, namely when property (\textsl{A}\ref{cond:omega_fixedp}) is satisfied \cite[Section 4.6]{Horstein-report}. Since in this case $(P_X,P_{Y|X})\in\Omega_B$, Theorem \ref{thrm:achv} verifies that this conjecture indeed holds. Moreover, the restriction of not having fixed points is in fact superfluous, since by Lemma \ref{lem:dmc_prop}, claim (\ref{claim:dmc5}), there always exists an equivalent input/DMC pair (obtained simply by an input permutation) for which the posterior matching scheme is capacity achieving. This scheme can be easily translated into an equivalent optimal scheme for the original channel $P_{Y|X}$, which is in fact one of the many $\mu$-variants satisfying the posterior matching principle mentioned in Corollary \ref{cor:eq_schemes}, where the u.p.f. $\mu$ plays the role of the input permutation. This observation is further discussed and generalized in Section~\ref{sec:extensions}-\ref{subsec:mu_var}.

More generally, let $P_X$ be any input distribution for $P_{Y|X}$, e.g. capacity achieving under some input constraints. If the associated kernel is fixed-point free ((\textsl{A}\ref{cond:omega_fixedp}) holds) and (\textsl{B}\ref{cond:omega_dm_Q}) is satisfied, then by Lemma \ref{lem:dmc_prop}, claim (\ref{claim:dmc3}), we have that (\textsl{A}\ref{cond:omega_inv}) holds as well. This implies $(P_X,P_{Y|X})\in\Omega_A$, and hence by Theorem \ref{thrm:achv} the associated posterior matching scheme achieves rates up to the corresponding mutual information $I(X;Y)$, within any input constraints encapsulated in $P_X$. Again, the fixed-point requirement is superfluous, and achievability within the same input constraints can be guaranteed via a posterior matching scheme for an equivalent channel (or the corresponding $\mu$-variant), for which the kernel is fixed-point free.

It is worth noting that requiring property (\textsl{B}\ref{cond:omega_dm_Q}) to hold is practically nonrestrictive. For any fixed alphabet sizes $|\m{X}|,|\m{Y}|$, there is only a countable number of input/channel pairs that fail to satisfy this property. Moreover, even if $(P_X,P_{Y|X})$ does not satisfy (\textsl{B}\ref{cond:omega_dm_Q}), then by Lemma \ref{lem:dmc_prop}, claim (\ref{claim:dmc6}), we can find an input distribution $P_X'$ arbitrarily close (in total variation) to $P_X$, such that (\textsl{B}\ref{cond:omega_dm_Q}) does hold for $(P_X',P_{Y|X})$. Hence, the posterior matching scheme (or a suitable variant, if there are fixed points) for $(P_X',P_{Y|X})$ achieves rates arbitrarily close to $I(X;Y)$ while maintaining any input constraint encapsulated in $P_X$ arbitrarily well.

\begin{remark}
For input/DMC pairs such that $(P_X,P_{Y|X})\in\Omega_B$ but where (\textsl{B}\ref{cond:omega_dm_Q}) does not hold, ergodicity is not guaranteed (see also Remark \ref{rem:weaken_irr}). Therefore, although the (unconstrained) capacity is achieved, the empirical distribution of the input sequence $X^n$ will not necessarily approach $P_X$, unless $P_X$ is the \textit{unique} capacity achieving distribution for $P_{Y|X}$ (see Remark \ref{rem:non_erg}).
\end{remark}
\begin{remark}\label{rem:zero_transp}
The nonzero DMC transition probabilities restriction (\textsl{B}\ref{cond:omega_dm_nzI}) is mainly intended to guarantee that the regularity property (\textsl{A}\ref{cond:omega_reg}) is satisfied (although this property holds under somewhat more general conditions, e.g., for the BEC.). However, regularity can be defined in a less restricting fashion so that this restriction could be removed. Roughly speaking, this can be done by redefining the left-$\eps$-measure and right-$\eps$-measure of Section~\ref{sec:fam} so that the neighborhoods over which the infimum is taken shrink near some finite collection of points in $\ui$, and not only near the endpoints, thereby allowing ``holes'' in the conditional densities. For simplicity of exposition, this extension was left out.
\end{remark}

\noindent{\bf Example \ref{ex:exp_noise_mean} (Exponential noise with an input mean constraint, continued)}. This example is not immediately covered by the Lemmas developed. However, studying the input/channel pair $(P_{Y|\Theta},P_\Theta)$ (namely, the normalized pair but without the artificial output transformation), we see that $P_{\Theta Y}$ satisfies property (\textsl{A}\ref{cond:omegac_proper}), and the corresponding posterior matching kernel (which is easily derived from (\ref{eq:exp_mean_kernel})) is fixed-point free, hence property (\textsl{A}\ref{cond:omegac_fixedp}$^*$) is also satisfied. Proving that this is a regular pair is straightforward but requires some work. Loosely speaking, it stems from the fact that $f_{Y|\Theta}(y|\theta)$ is monotonically decreasing in $y$ for any fixed $\theta$, and has a one-sided regular tail. Therefore, the posterior matching scheme (\ref{eq:exp_mean_scheme}) pointwise achieves any rate below the mean-constrained capacity $I(X;Y)=\log{(1+\frac{a}{b})}$.

\subsection{Proof of Theorem \ref{thrm:achv}}\label{subsec:main_proof}
Let us start by providing a rough outline of the proof. First, we show
that zero rate is achievable, i.e., any fixed interval around the
message point accumulates a posterior probability mass that tends
to one. This is done by noting that the time evolution of the
posterior c.d.f. $F_{\Theta_0|\Phi^n}$ can be represented
by an IFS over the space $\mf{F}_{c}$, generated by the inverse
channel's c.d.f. via function composition, and controlled by the
channel outputs. Showing that the inverse channel's c.d.f. is
contractive on the average (Lemma \ref{lem:contraction_cond}), we
conclude that the posterior c.d.f. tends to a unit step function
about $\Theta_0$ (Lemma \ref{lem:diverge_eps}) which verifies
zero-rate achievability. For positive rates, we use the SLLN for
Markov chains to show that the posterior p.d.f. at the message
point is $\approx 2^{nI(X;Y)}$ (Lemma \ref{lem:SLLN_memoryless}).
Loosely speaking, a point that cannot be distinguished from
$\Theta_0$ must induce, from the receiver's perspective, about
the same input sequence as does the true message point. Since
the normalized inputs are just the posterior c.d.f. sequence
evaluated at the message point, this means that such points will
also have about the same c.d.f. sequence as $\Theta_0$ does, hence
also will have a posterior p.d.f. $\approx 2^{nI(X;Y)}$. But that
is only possible within an interval no larger than $\approx
2^{-nI(X;Y)}$ around $\Theta_0$, since the posterior p.d.f.
integrates to unity. Thus, points that cannot be distinguished
from $\Theta_0$ must be $2^{-nI(X;Y)}$ close to it. This is more
of a converse, but essentially the same ideas can be applied
(Lemma \ref{lem:diverge_R}) to show that for any $R<I(X;Y)$, a
$2^{-nR}$ neighborhood of the message point accumulates (with high
probability) a posterior probability mass exceeding some fixed
$\eps>0$ at some point during the first $n$ channel uses. This
essentially reduces the problem to the zero-rate setting, which
was already solved.

We begin by establishing the required technical Lemmas.
\begin{lemma}\label{lem:contraction_cond}
Let $(P_X, P_{Y|X})$ satisfy property (\textsl{A}\ref{cond:omega_fixedp}). Then there exist a contraction $\xi(\cdot)$ and a length function
$\psi_{\sst{\lambda}}(\cdot)$ as in (\ref{eq:def_lenght_func}) over $\mf{F}_{c}$, such that for any
$h\in\mf{F}_{c}$
\begin{equation}\label{eq:contraction_prop}
\Expt\Big{(}\psi_{\sst{\lambda}}\big{[}F_{\Theta|\Phi}(\cdot\,|\Phi)\circ
h\big{]}\Big{)} \,\leq
\;\xi\big{(}\psi_{\sst{\lambda}}(h\,)\,\big{)}
\end{equation}
\end{lemma}
\begin{proof}
See Appendix \ref{app:lemmas}.
\end{proof}

Define the stochastic process $\{\bar{G}_n\}_{n=1}^\infty$,
\begin{equation*}
\bar{G}_n(\cdot) \dfn \bar{g}_n(\cdot,\Phi^{n-1})
\end{equation*}
Since $\bar{g}_n(\theta,\phi^{n-1})=F_{\Theta_0|\Phi^{n-1}}(\theta|\phi^{n-1})$,
$\bar{G}_n$ is the posterior c.d.f. of the message point
after observing the i.i.d. output sequence $\Phi^{n-1}$, and is a
r.v. taking values in the c.d.f. space $\mf{F}_{c}$. Moreover, by
(\ref{eq:gn_rec_norm}) we have that
\begin{equation}\label{eq:Gn_evol}
\bar{G}_{n+1}=F_{\Theta|\Phi}(\cdot|\Phi_n)\circ \bar{G}_n
\end{equation}
and therefore $\{\bar{G}_n\}_{n=1}^\infty$ is an IFS over
$\mf{F}_{c}$, generated by the normalized posterior matching kernel
$F_{\Theta|\Phi}(\cdot|\phi)$ (via function composition) and
controlled by the outputs $\{\Phi_n\}_{n=1}^\infty$. Since the
message point is uniform, the IFS initializes at
$\bar{G}_1(\theta) = \theta\ind_{(0,1)}(\theta)+\ind_{[1,\infty)}(\theta)$ (the uniform c.d.f.).
Recall that the normalized kernel is continuous in $\theta$ for $P_\Phi$-a.a. $\phi$ (Lemma \ref{lem:norm_chan_prop}, claim (\ref{claim:norm4})), hence $\bar{G}_n$ is a.s. continuous.

We find it convenient to define the \textit{$\delta$-positive trajectory}
$\{\ssc{+}\Theta^{\delta}_k\}_{k=1}^\infty$ and
\textit{$\delta$-negative trajectory}
$\{\ssc{-}\Theta^{\delta}_k\}_{k=1}^\infty$, as follows:
\begin{align}\label{eq:trajectory}
\nonumber &\ssc{+}\Theta^{\delta}_k \dfn
\bar{G}_k(\Theta_0+\Delta^+_\delta) \,,\qquad \Delta^+_\delta = \min\left(\delta,\frac{1-\Theta_0}{2}\right)\\
& \ssc{-}\Theta^{\delta}_k \dfn
\bar{G}_k(\Theta_0-\Delta^-_\delta) \,,\qquad \Delta^-_\delta =
\min\left(\delta,\frac{\Theta_0}{2}\right)
\end{align}
These trajectories are essentially the posterior c.d.f. evaluated
after $k$ steps at a $\delta$ perturbation from $\Theta_0$ (up to
edge issues), or alternatively the induced normalized input
sequence for such a perturbation from the point of view of the
receiver. The true normalized input sequence, which corresponds to
the c.d.f. evaluated at the message point itself, is
$\Theta_k=\bar{G}_k(\Theta_0)$.

The next Lemma shows that for a zero rate, the trajectories
diverge towards the boundaries of the unit interval with
probability approaching one, hence our scheme has a vanishing
error probability in this special case.
\begin{lemma}\label{lem:diverge_eps}
Let $(P_X, P_{Y|X})$ satisfy property (\textsl{A}\ref{cond:omega_fixedp}). Then for any $\eps>0,\delta>0$,
\begin{align*}
\Prob\big{(}\ssc{-}\Theta^{\delta}_n>\eps\big{)} = \bigo\left(\sqrt[8]{r(n)}\right)\,,\qquad
\Prob\big{(}\ssc{+}\Theta^{\delta}_n<1-\eps\big{)} = \bigo\left(\sqrt[8]{r(n)}\right)
\end{align*}
where $r(n)$ is the decay profile of the contraction
$\xi(\cdot)$ from Lemma \ref{lem:contraction_cond}.
\end{lemma}
\begin{proof}
Let $\psi_{\sst{\lambda}}$ and $\xi$ be the length function and
contraction from Lemma \ref{lem:contraction_cond} corresponding to
the pair $(P_X,P_{Y|X})$, and let $r(n)$ be the decay profile of
$\xi$. By the contraction property (\ref{eq:contraction_prop}) and
Lemma \ref{lem:IFS_convergence}, we immediately have that for any
$\nu>0$
\begin{equation}\label{eq:length_tends_to_zero}
\Prob\left(\psi_{\sst{\lambda}}(\bar{G}_n)
> \nu\right) \leq \nu^{-1}\,r(n)
\end{equation}
Define the (random) \textit{median point} of $\bar{G}_n$:
\begin{equation*}
\Theta^*_n \dfn \inf\left\{\theta\in\ui : \bar{G}_n(\theta) \geq \frac{1}{2}\right\}
\end{equation*}
Since $\bar{G}_n$ is a.s. continuous, $\bar{G}_n(\Theta^*_n) = \frac{1}{2}$ is a.s. satisfied.
Using the symmetry of the function $\lambda(\cdot)$, we can write
\begin{equation}\label{eq:length_Gn}
\psi_{\sst{\lambda}}(\bar{G}_n) =
\int_0^{\Theta^*_n}\lambda(\bar{G}_n(\theta))d\theta +
\int_{\Theta^*_n}^1\lambda(1-\bar{G}_n(\theta))d\theta \;\quad \text{a.s.}
\end{equation}
and then:
\begin{equation*}
\Prob\left(\bar{G}_n(\Theta^*_n-\delta)> \nu\right) \stackrel{(\rm a)}{\leq} \Prob\left(\lambda\left(\bar{G}_n(\Theta^*_n-\delta)\right)
> \nu\right) \stackrel{(\rm b)}{\leq} \Prob\left(\int_{\Theta^*_n-\delta}^{\Theta^*_n}\lambda\left(\bar{G}_n(\theta)\right)d\theta
> \nu\delta\right)\stackrel{(\rm c)}{\leq} \Prob(\psi_{\sst\lambda}(\bar{G}_n)>\nu\delta)
\end{equation*}
where (a) holds since $\lambda(\theta)>\theta$ for any $\theta\in
(0,\frac{1}{2})$, in (b) we use the monotonicity of
$\bar{G}_n$, and (c) follows from (\ref{eq:length_Gn}). Using
(\ref{eq:length_tends_to_zero}) this leads to
\begin{equation}\label{eq:Gn_left}
\Prob\left(\bar{G}_n(\Theta^*_n-\delta) > \nu\right) \leq
\frac{1}{\nu\delta}\,r(n)
\end{equation}
and similarly
\begin{equation}\label{eq:Gn_right}
\Prob\left(\bar{G}_n(\Theta^*_n+\delta) < 1-\nu\right) \leq
\frac{1}{\nu\delta}\,r(n)
\end{equation}
Now set any $\eta\in(0,\frac{1}{2})$, and write
{\allowdisplaybreaks
\begin{align}\label{eq:long1}
\nonumber\Prob&\left(\int_0^{\Theta_0}\bar{G}_n(\theta)d\theta+\int_{\Theta_0}^1\left(1-\bar{G}_n(\theta)\right)d\theta
> \nu\right)
\\
\nonumber&\stackrel{(\rm a)}{\leq}
\Prob\left(\left\{\int_0^{\Theta^*_n}\bar{G}_n(\theta)d\theta+\int_{\Theta^*_n}^1\left(1-\bar{G}_n(\theta)\right)d\theta
> \frac{\nu}{2}\right\}\cup \left\{|\Theta^*_n-\Theta_0| > \frac{\nu}{2}\right\}\right)
\\
\nonumber&\stackrel{(\rm b)}{\leq}
\Prob\left(\psi_{\sst\lambda}(\bar{G}_n) > \frac{\nu}{2}\right) +
\Prob\left(|\Theta^*_n-\Theta_0| > \frac{\nu}{2}\right) \stackrel{(\rm c)}{\leq}
\frac{2}{\nu}\,r(n) + \Prob\left(|\Theta^*_n-\Theta_0| > \frac{\nu}{2}\right)
\\
\nonumber&\stackrel{(\rm d)}{=} \frac{2}{\nu}\,r(n) +
\Prob\left(\left\{\bar{G}_n(\Theta_0) > \bar{G}_n(\Theta^*_n+\frac{\nu}{2})\right\}\cup\left\{\bar{G}_n(\Theta_0)
< \bar{G}_n(\Theta^*_n-\frac{\nu}{2})\right\}\right)
\\
\nonumber&\stackrel{(\rm e)}{\leq} \frac{2}{\nu}\,r(n) +
\Prob\left(\bar{G}_n(\Theta_0)\not\in(\eta,1-\eta)\right)+
\Prob\left(\bar{G}_n(\Theta^*_n-\frac{\nu}{2})>\eta\right) +
\Prob\left(\bar{G}_n(\Theta^*_n+\frac{\nu}{2})<1-\eta\right)
\\
&\stackrel{(\rm f)}{\leq} \frac{2}{\nu}\,r(n) + 2\eta
+\frac{4}{\nu\eta}\,r(n)
\end{align}}
where in (a) we use the fact that integrals differ only over the
interval between $\Theta_0$ and $\Theta^*_n$ and the integrands
are bounded by unity, in (b) we use the union bound, and then
(\ref{eq:length_tends_to_zero}) by noting that applying
$\lambda(\cdot)$ can only increase the integrands, in (c) we use
(\ref{eq:length_tends_to_zero}) and (d) holds by the continuity
and monotonicity of $\bar{G}_n$. These properties are applied
again together with the union bound in (e), and the inequality
holds for any $\eta\in(0,\frac{1}{2})$. Finally in (f) we use
(\ref{eq:Gn_left}--\ref{eq:Gn_right}) the fact that
$\bar{G}_n(\Theta_0)=\Theta_n$ is uniformly distributed over the
unit interval. Choosing $\eta = \sqrt{r(n)}\;$ we get
\begin{align*}
\Prob\left(\int_0^{\Theta_0}\bar{G}_n(\theta)d\theta+\hspace{-3pt}\int_{\Theta_0}^1(1-\bar{G}_n(\theta))d\theta
> \nu\right) \hspace{-2pt}\leq c\nu^{-1}\sqrt{r(n)}
\end{align*}
for $c>0$. The same bound clearly holds separately for each
of the two integrals above. Define the set
\begin{equation*}
\Pi_\nu \dfn \left\{\theta\in(0,1) : \Prob\left(\int_0^{\Theta_0}\bar{G}_n(\theta)d\theta\ > \nu
\mid \Theta_0 = \theta\right)\geq c\nu^{-1}\sqrt[4]{r(n)}\right\}
\end{equation*}
Then
\begin{equation*}
\Prob\left(\int_0^{\Theta_0}\bar{G}_n(\theta)d\theta > \nu\right)
= \Expt\Prob\left(\int_0^{\Theta_0}\bar{G}_n(\theta)d\theta\
> \nu \mid \Theta_0\right) \geq \Prob(\Theta_0\in\Pi_\nu)\cdot c\nu^{-1}\sqrt[4]{r(n)}
\end{equation*}
and we get $\Prob(\Theta_0\in\Pi_\nu) \leq \sqrt[4]{r(n)}$. Let us
now set $\nu_n=\sqrt[8]{r(n)}$, and suppose $n$ is large enough so
that $\nu_n<\frac{\eps\delta}{2}$. Recalling the definition of the
negative trajectory $\ssc{-}\Theta_n^{\delta}$, we have
{\allowdisplaybreaks
\begin{align*}
\Prob\left(\ssc{-}\Theta_n^{\delta} > \eps\right)
&=\Expt\Prob\left(\ssc{-}\Theta_n^{\delta} > \eps \mid
\Theta_0\right) \leq \int_0^1
\Prob\left(\int_0^{\Theta_0}\bar{G}_n(\theta)d\theta\
> \frac{\eps}{2}\cdot\min\{\delta,\theta\} \,\Big{|}\, \Theta_0 =
\theta\right)d\theta
\\
&\leq \int_0^1
\Prob\left(\int_0^{\Theta_0}\bar{G}_n(\theta)d\theta\
> \min\{\nu_n,\frac{\eps\theta}{2}\} \,\Big{|}\, \Theta_0 =
\theta\right)d\theta
\\
&\leq \Prob(\Theta_0\in\Pi_{\nu_n})  +
\int_0^{2\nu_n\eps^{-1}}d\theta + \int_{2\nu_n\eps^{-1}}^1
c\nu_n^{-1}\sqrt[4]{r(n)}d\theta
\\
&\leq \sqrt[4]{r(n)}+ 2\nu_n\eps^{-1} + c\nu_n^{-1}\sqrt[4]{r(n)}
= \bigo\left(\sqrt[8]{r(n)}\right)
\end{align*}}
The result for $\Prob\left(\ssc{+}\Theta_n^{\delta} <
1-\eps\right)$ is proved via the exact same arguments.
\end{proof}

\begin{lemma}\label{lem:SLLN_memoryless}
Let $(P_X, P_{Y|X})$ satisfy property (\textsl{A}\ref{cond:omega_inv}). Then the posterior p.d.f. evaluated at the message point satisfies
\begin{equation}\label{eq:density_at_theta}
\lim_{n\rightarrow \infty}\frac{1}{n}\log
f_{\Theta_0|\Phi^n}(\Theta_0|\Phi^n) = I(X;Y) \qquad \text{\rm
a.s.}
\end{equation}
\end{lemma}
\begin{proof}
Since the p.d.f.'s involved are all proper, we can use Bayes law
to obtain the following recursion rule:
\begin{equation}\label{eq:posterior_expansion}
f_{\Theta_0|\Phi^n}(\theta|\phi^n) = \frac{f_{\Phi_n|\Theta_0,\Phi^{n-1}}(\phi_n\,|\,\theta,\phi^{n-1})}{f_{\Phi_n|\Phi^{n-1}}(\phi_n\,|\,\phi^{n-1})}\,f_{\Theta_0|\Phi^{n-1}}(\theta|\phi^{n-1})
= f_{\Phi|\Theta}(\phi_n\,|\,\bar{g}_n(\theta,\phi^{n-1}))\cdot
f_{\Theta_0|\Phi^{n-1}}(\theta|\phi^{n-1})
\end{equation}
where in the second equality we have used the memoryless channel
property and the fact that the output sequence $\Phi^\infty$ is an
i.i.d. sequence with marginal $\m{U}$. Applying the recursion rule
$n$ times, taking a logarithm and evaluating at the message point,
we obtain
\begin{equation*}
\frac{1}{n}\log f_{\Theta_0|\Phi^n}(\Theta_0|\Phi^n) = \frac{1}{n}\sum_{k=1}^n\log
f_{\Phi|\Theta}\left(\Phi_k\,|\,g_k\left(\Theta_0,\Phi^{k-1}\right)\right)
=\frac{1}{n}\sum_{k=1}^n\log f_{\Phi|\Theta}(\Phi_k\,|\,\Theta_k)
\end{equation*}
Now by property (\textsl{A}\ref{cond:omega_inv}) the invariant
distribution $P_{\Theta\Phi}$ is ergodic, and so we can use the
SLLN for Markov chains (Lemma \ref{lem:SLLN}) which
asserts in this case that for $P_\Theta$-a.a. $\theta_0\in\ui$
\begin{equation*}
\lim_{n\rightarrow\infty}\frac{1}{n}\log f_{\Theta_0|\Phi^n}(\theta_0|\Phi^n) = \Expt\left(\log
\frac{f_{\Phi|\Theta}(\Phi|\Theta)}{f_\Phi(\Phi)}\right) = I(\Theta;\Phi) = I(X;Y) \qquad \m{P}_{\theta_0}\text{-\rm a.s.}
\end{equation*}
Since $\Theta_0\sim P_\Theta$, (\ref{eq:density_at_theta}) is established.
\end{proof}

For short, let us now define the \textit{$(n,R)$-positive trajectory} $\ssc{+}\Theta^{n,R}_k$ and the \textit{$(n,R)$-negative trajectory} $\ssc{-}\Theta^{n,R}_k$ as the corresponding trajectories in (\ref{eq:trajectory}) with $\delta=2^{-nR}$. Accordingly, we also write $\Delta^+_{n,R},\Delta^-_{n,R}$ in lieu of $\Delta^+_\delta,\Delta^-_\delta$ respectively. The following Lemma uses the SLLN to demonstrate how, for rates lower than the mutual information, these two trajectories eventually move away from some small and essentially fixed neighborhood of the input,
with probability approaching one. This is achieved by essentially proving a more subtle version of Lemma \ref{lem:SLLN_memoryless}, showing that it roughly holds at the vicinity of the message point.
\begin{lemma}\label{lem:diverge_R}
Let $(P_X, P_{Y|X})$ satisfy properties (\textsl{A}\ref{cond:omega_reg}) and (\textsl{A}\ref{cond:omega_inv}). Then for any rate $R<I(X;Y)$ there exists $\eps>0$  small enough such that
\begin{align}\label{eq:prob1}
\nonumber
&\lim_{n\rightarrow\infty}\Prob\left(\;\bigcap_{k=1}^n\left\{\Theta_k-\ssc{-}\Theta^{n,R}_k
< \min\Big{(}\eps,\frac{\Theta_k}{2}\,\Big{)}\right\}
 \right) = 0
\\
&\lim_{n\rightarrow\infty}\Prob\left(\;\bigcap_{k=1}^n\left\{\ssc{+}\Theta^{n,R}_k-\Theta_k
< \min\Big{(}\eps,\frac{1-\Theta_k}{2}\,\Big{)}\right\}\right)
= 0
\end{align}
\end{lemma}
\begin{proof}
We prove the first assertion of (\ref{eq:prob1}), the second
assertion follows through essentially the same way. Let $\delta>0$ be such that $R<I(X;Y)-\delta$. Let
$\ssc{-}P^\eps_{\Phi|\Theta}(\cdot|\theta)$ be the
left-$\eps$-measure corresponding to
$P_{\Phi|\Theta}(\cdot|\theta)$, as defined in Section~\ref{sec:fam}. Define:
\begin{equation}\label{eq:eps_inf}
I_{\eps}^- \dfn \Expt\log \ssc{-}f^\eps_{\Phi|\Theta}(\Phi|\Theta) =\dint{\supp\left(\Theta,\Phi\right)} f_{\Phi|\Theta}(\phi|\theta)\log\ssc{-}f^\eps_{\Phi|\Theta}(\phi|\theta)\,d\theta d\phi
\end{equation}
We have that
\begin{equation*}
0\leq I(X;Y)-I_{\eps}^- = I(\Theta;\Phi)-I_{\eps}^- =
D(P_{\Phi|\Theta}\,\|\,\ssc{-}P^\eps_{\Phi|\Theta}\,|\,P_\Theta)
\end{equation*}
and since by property (\textsl{A}\ref{cond:omega_reg}) the input/channel is
regular then
$\inf_{\eps>0}D(P_{\Phi|\Theta}\,\|\,\ssc{-}P^\eps_{\Phi|\Theta}\,|\,P_\Theta)<\infty$.
hence for any $\eps$ small enough
\begin{equation*}
-\infty < I_{\eps}^- \leq I(X;Y)
\end{equation*}
We have therefore established that the function
$f_{\Phi|\Theta}\log \ssc{-}f^\eps_{\Phi|\Theta}$ is finitely integrable
for any $\eps>0$ small enough, and converges to
$f_{\Phi|\Theta}\log f_{\Phi|\Theta}$ a.e in a monotonically
nondecreasing fashion, as $\eps\rightarrow 0$. Applying Levi's
monotone convergence Theorem~\cite{kolmogorov_fomin}, we can
exchange the order of the limit and the integration to obtain
\begin{equation*}
\lim_{\eps\rightarrow 0}I_{\eps}^- = I(X;Y)
\end{equation*}
Let us set $\eps$ hereinafter so that
\begin{equation*}
I_{\eps}^- -\eps > I(X;Y) -\frac{\delta}{2}
\end{equation*}
Since $I_{\eps}^-$ is finite we can once again apply (using
property (\textsl{A}\ref{cond:omega_inv})) the SLLN for Markov chains
(Lemma \ref{lem:SLLN}) to obtain
\begin{equation}\label{eq:SLLN_substochastic}
\lim_{n\rightarrow\infty}\frac{1}{n}\sum_{k=1}^n \log
\ssc{-}f^\eps_{\Phi|\Theta}(\Phi_k|\Theta_k) =
I_{\eps}^-\qquad \text{a.s.}
\end{equation}
The above intuitive relation roughly means that if the receiver, when considering the likelihood of a wrong message point, obtains an induced input sequence which is always $\eps$-close to the true input sequence, then the posterior p.d.f. at this wrong message point will be close to that of the true message point given in Lemma \ref{lem:SLLN_memoryless}.

Define the following two sequences of events:
\begin{equation}\label{eq:event_seq}
E_{n,\eps} \dfn \bigcap_{k=1}^n\left\{\ssc{-}\Theta^{n,R}_k\in
J_{\st\eps}^-(\Theta_k,\Phi_k)\right\}\,,\quad \widetilde{E}_{n,\eps}\dfn\bigcap_{k=1}^n\left\{\Theta_k-\ssc{-}\Theta^{n,R}_k
< \min\Big{(}\eps,\frac{\Theta_k}{2}\,\Big{)}\right\}
\end{equation}
where the neighborhood $J_{\st\eps}^-$ is defined in
(\ref{eq:eps_channel}). Let us now show that
${\displaystyle\lim_{n\rightarrow \infty}}\Prob(E_{n,\eps})=0$.
This fact will then be shown to imply
${\displaystyle\lim_{n\rightarrow
\infty}}\Prob(\widetilde{E}_{n,\eps})=0$, which is precisely the
first assertion in (\ref{eq:prob1}). Define the following sequence of events:
\begin{equation*}
T_{n,\eps} \dfn \left\{\Theta_0 > 2^{-nR}\right\}
\cap \left\{\frac{1}{n}\sum_{k=1}^n \log
\ssc{-}f^\eps_{\Phi|\Theta}(\Phi_k|\Theta_k) \geq I(X;Y)
-\frac{\delta}{2}\right\}
\end{equation*}
Using (\ref{eq:SLLN_substochastic}) and the fact that the message
point is uniform over the unit interval, it is immediately clear
that $\Prob(T_{n,\eps}) \rightarrow 1$. For short, define the
random interval $J_{n,R}=(\Theta_0-\Delta^-_{n,R},\Theta_0)$, and
consider the following chain of inequalities:
{\allowdisplaybreaks
\begin{align*}
0 &\geq \log\Expt\left(\Theta_n-\ssc{-}\Theta^{n,R}_n\right) =
\log\Expt\left(\bar{G}_n(\Theta_0)-\bar{G}_n(\Theta_0-\Delta^-_{n,R}\,)\right)
\geq \log\Expt\left(\Delta^-_{n,R}\cdot\inf_{\theta\in
J_{n,R}}\left(f_{\Theta_0|\Phi^{n-1}}\left(\theta\,|\,\Phi^{n-1}\right)\right)\right)
\\
&\geq \log\left(\Expt\left(\Delta^-_{n,R}\cdot\inf_{\theta\in
J_{n,R}}f_{\Theta_0|\Phi^{n-1}}\left(\theta\,|\,\Phi^{n-1}\right)\,\Big{|}\,E_{n,\eps}\cap
T_{n,\eps}\right)\cdot\Prob(E_{n,\eps}\cap T_{n,\eps})\right)
\\
&\stackrel{({\rm a})}{\geq}
\Expt(\log\Delta^-_{n,R}\,|\,E_{n,\eps}\cap
T_{n,\eps})+\Expt\left(\log\inf_{\theta\in J_{n,R}}\prod_{k=1}^n
f_{\Phi|\Theta}(\Phi_k|\,\bar{G}_k(\theta))\,\Big{|}\,E_{n,\eps}\cap
T_{n,\eps}\right) + \log\Prob(E_{n,\eps}\cap T_{n,\eps})
\\
&\stackrel{({\rm b})}{\geq} -nR - 1 +
\Expt\left(\,\sum_{k=1}^n\log\inf_{\xi\in
(\ssc{-}\Theta^{n,R}_k,\Theta_k)}f_{\Phi|\Theta}(\Phi_k|\,\xi)\,\Big{|}\,E_{n,\eps}\cap
T_{n,\eps}\right) + \log\Prob(E_{n,\eps}\cap T_{n,\eps})
\\
&\stackrel{({\rm c})}{\geq} -nR - 1 +
\Expt\left(\,\sum_{k=1}^n\log
\ssc{-}f^\eps_{\Phi|\Theta}(\Phi_k|\Theta_k)\,\Big{|}\,E_{n,\eps}\cap
T_{n,\eps}\right) + \log\Prob(E_{n,\eps}\cap T_{n,\eps})
\\
&\stackrel{({\rm d})}{\geq} -nR - 1 + n\big{(}I(X;Y)
-\frac{\delta}{2}\big{)} + \log\Prob(E_{n,\eps}\cap T_{n,\eps})
\end{align*}}
In (a) we use Jensen's inequality and the expansion of
the posterior p.d.f. given in (\ref{eq:posterior_expansion}), in
(b) we use the definition and monotonicity of $\bar{G}_k$, (c)
holds due to $E_{n,\eps}$ and (d) due to $T_{n,\eps}$. Therefore,
\begin{equation}\label{eq:prob_bound1}
\Prob(E_{n,\eps}\cap T_{n,\eps}) \leq
2^{-n(I(X;Y)-\frac{\delta}{2}-R)-1} \leq 2^{-n\frac{\delta}{2}-1}
\;\;\tendsto{n}{\infty}\;\;0
\end{equation}
where the last inequality holds since $R<I(X;Y)-\delta$. Now,
since $\Prob(T_{n,\eps})\rightarrow 1$, then for any $\eta>0$ we
have $\Prob(T_{n,\eps}) > 1-\eta$ for $n$ large enough. Using that
and (\ref{eq:prob_bound1}) we bound $\Prob(E_{n,\eps})$ simply as
follows:
\begin{equation}\label{eq:bound_En}
\Prob(E_{n,\eps})\leq \Prob(E_{n,\eps}\cap T_{n,\eps}) +
\left(1-\Prob(T_{n,\eps})\right) \leq 2^{-n\frac{\delta}{2}-1} +
\eta \leq 2\eta
\end{equation}
where the last two inequalities are true for $n$ large enough.
Since (\ref{eq:bound_En}) holds for any $\eta>0$, we conclude that
$\Prob(E_{n,\eps})\rightarrow 0$, as desired.

To finalize the proof, note that $\widetilde{E}_{n,\eps}$ implies
that for any $1\leq k\leq n-1$
\begin{equation*}
\Theta_k-\ssc{-}\Theta^{n,R}_k < \eps,\quad \ssc{-}\Theta^{n,R}_k  > \frac{\Theta_k}{2},
\end{equation*}
and the rightmost inequality implies
\begin{equation*}
F_{\Theta|\Phi}\left(\ssc{-}\Theta^{n,R}_{k-1}\big{|}\Phi_{k-1}\right)
> \frac{1}{2}F_{\Theta|\Phi}\left(\Theta_{k-1}|\Phi_{k-1}\right)
\end{equation*}
The above constraints imply that for any $1\leq k\leq n-2$
\begin{equation*}
\ssc{-}\Theta^{n,R}_k\in J_\eps^-(\Phi_k,\Theta_k)
\end{equation*}
establishing the implication $\widetilde{E}_{n,\eps}\Rightarrow
E_{n-1,\eps}$. Consequently, $\Prob(\widetilde{E}_{n,\eps}) \leq
\Prob(E_{n-1,\eps})$, thus
${\displaystyle\lim_{n\rightarrow\infty}}\Prob(\widetilde{E}_{n,\eps})
= 0$.
\end{proof}

We are now finally in a position to prove Theorem \ref{thrm:achv} for the family $\Omega_A$. Loosely speaking, we build on the simple fact that since the chain is stationary by construction, one can imagine transmission to have started at any time $m$ with a message point $\Theta_m$ replacing
$\Theta_0$. With some abuse of notations, we define
\begin{equation*}
\ssc{-}\Theta_{k}^\eps(\Theta_m) \dfn F_{\Theta|\Phi}(\cdot|\Phi_{m+k-1})\circ\cdots\circ F_{\Theta|\Phi}(\cdot|\Phi_{m+1})\circ F_{\Theta|\Phi}\left(\Theta_m-\min\left(\eps,\frac{\Theta_m}{2}\right) \mid \Phi_m\right)
\end{equation*}
Namely, the $\eps$-negative trajectory when starting at time $m$ from $\Theta_m$. Note that in particular we have $\ssc{-}\Theta_k^\eps(\Theta_0) = \ssc{-}\Theta_k^\eps$. Since the chain is stationary, the distribution of $\ssc{-}\Theta_{k}^\eps(\Theta_m)$ is independent of $m$. The corresponding positive trajectory $\ssc{+}\Theta_{k}^\eps(\Theta_m)$ can be defined in the same manner.

Now recall the event $\widetilde{E}_{n,\eps}$ defined in (\ref{eq:event_seq}), which by Lemma \ref{lem:diverge_R} satisfies $\Prob(\widetilde{E}_{n,\eps})\rightarrow 0$ for any $\eps>0$ small enough. Note that the complementary event $\widetilde{E}_{n,\eps}^c$ implies that at some time $m\leq n$, the $(n,R)$-negative trajectory $\ssc{-}\Theta_m^{n,R}$ is below the $\eps$-neighborhood of $\Theta_m$, namely $\ssc{-}\Theta_m^{n,R}\leq \Theta_m-\min(\eps,\frac{\Theta_m}{2})$ for some $m$. Using the monotonicity of the transmission functions, this in turn implies that the $(n,R)$-negative trajectory at time $n$ lies below the corresponding $\eps$-negative trajectory starting from $\Theta_m$, namely $\ssc{-}\Theta_n^{n,R}\leq \ssc{-}\Theta_{n-m}^\eps(\Theta_m)$. Thus we conclude that $\Prob\left(\ssc{-}\Theta_n^{n,R}>\max\left(\ssc{-}\Theta_{n-m}^\eps(\Theta_m)\right)\right)\rightarrow 0$ for any fixed $\eps>0$ small enough, where the maximum is taken over $1\leq m\leq n$.

Fixing any $\alpha>0$, we now show that $\ssc{-}\Theta_{(1+\alpha) n}^{n,R}\rightarrow 0$ in probability:
{\allowdisplaybreaks
\begin{align}\label{eq:final_proof}
\nonumber\Prob(\ssc{-}\Theta_{(1+\alpha) n}^{n,R}>\delta) &\leq \Prob\left(\ssc{-}\Theta_n^{n,R}>\max_{1\leq m\leq n}\ssc{-}\Theta_{n-m}^\eps(\Theta_m)\right)
+ \Prob\left(\hspace{-2pt}\left\{\ssc{-}\Theta_{(1+\alpha) n}^{n,R}>\delta\}\right\}\cap\{\ssc{-}\Theta_n^{n,R}\leq \hspace{-2pt}\max_{1\leq m\leq n}\ssc{-}\Theta_{n-m}^\eps(\Theta_m)\}\hspace{-2pt}\right)
\\
\nonumber&\leq \lito(1) + \Prob\left(\max_{1\leq m\leq n}\ssc{-}\Theta_{(1+\alpha) n-m}^\eps(\Theta_m) > \delta\right) \stackrel{\rm (a)}{\leq} \lito(1) + \sum_{m=1}^n \Prob\left(\ssc{-}\Theta_{(1+\alpha) n-m}^\eps(\Theta_m) > \delta\right)
\\
&\stackrel{\rm (b)}{=}  \lito(1) + \sum_{m=1}^n \Prob\left(\ssc{-}\Theta_{(1+\alpha) n-m}^\eps(\Theta_0) > \delta\right) \stackrel{\rm (c)}{=} \lito(1) + \delta^{-1}\bigo(n\sqrt[8]{r(\alpha n)})
\end{align}}
In (a) we used the union bound, in (b) the fact that the chain is stationary, and Lemma \ref{lem:diverge_eps} was invoked in (c) where we recall that (\textsl{A}\ref{cond:omega_inv}) $\Rightarrow$ (\textsl{A}\ref{cond:omega_fixedp}) by Lemma \ref{lem:egr_implies_fixed_free}. Therefore, if $\sqrt[8]{r(n)} = \lito(n^{-1})$ then $\ssc{-}\Theta_{(1+\alpha) n}^{n,R}\rightarrow 0 $ in probability is established. However, for the more general statement we note that this mild constraint\footnote{An exponentially decaying  $r(n)$ can in fact be guaranteed by requiring the normalized posterior matching kernel to be fixed-point free in a somewhat stronger sense than that implied by property (\textsl{A}\ref{cond:omega_inv}), which also holds in particular in all the examples considered in this paper.} is in fact superfluous. This stems from the fact that the union bound in (a) is very loose since the trajectories are all controlled by the same output sequence, and from the uniformity in the initial point in Lemma \ref{lem:IFS_convergence}. In Appendix \ref{app:lemmas}, Lemma \ref{lem:uniform_eps}, we show that in fact
\begin{equation*}
\Prob\left(\max_{1\leq m\leq n} \ssc{-}\Theta_{(1+\alpha)n-m}^\eps(\Theta_m) > \delta\right) = \bigo(\sqrt[8]{r(\alpha n)})
\end{equation*}
According to (\ref{eq:final_proof}), this in turn implies that $\ssc{-}\Theta_{(1+\alpha) n}^{n,R}\rightarrow 0$ in probability without the additional constraint on the decay profile.

The same derivation applies to the positive trajectory, resulting in $\ssc{+}\Theta_{(1+\alpha) n}^{n,R}\rightarrow 1$ in probability. Therefore, for any $\delta>0$,
\begin{equation*}
\Prob\left(\ssc{+}\Theta_{(1+\alpha) n}^{n,R}-\ssc{-}\Theta_{(1+\alpha) n}^{n,R} < 1-2\delta\right) \leq \Prob\left(\ssc{-}\Theta_{(1+\alpha) n}^{n,R} > \delta\right)+\Prob\left(\ssc{+}\Theta_{(1+\alpha) n}^{n,R} < 1-\delta\right) = \lito(1)
\end{equation*}
and so the posterior probability mass within a $2^{-nR}$ symmetric neighborhood of $\Theta_0$ (up to edge issues) after $(1+\alpha) n$ iterations, approaches one in probability as $n\rightarrow\infty$. We can therefore find a sequence $\delta_n\rightarrow 0$ such that the probability this mass exceeds $1-\delta_n$ tends to zero. Using the optimal variable rate decoding rule and setting the target error probability to $p_e(n)=\delta_n$ we immediately have that $\Prob(R_n<(1+\alpha)^{-1}R) \rightarrow 0$. This holds for any $R<I(X;Y)$, and since $\alpha>0$ can be arbitrarily small, any rate below the mutual information is achievable.

To prove achievability using the optimal fixed rate decoding rule, note that any variable-rate rule achieving some rate $R>0$ induces a fixed-rate rule achieving an arbitrarily close rate $R-\eps$, by extending the variable-sized decoded interval into a larger one of a fixed size $2^{-n(R-\eps)}$ whenever the former is smaller, and declaring an error otherwise. Therefore, any rate $R<I(X;Y)$ is achievable using the optimal fixed rate decoding rule. The fact that the input constraint is satisfied follows immediately from the SLLN since the marginal invariant distribution for the input is $P_X$. This concludes the achievability proof for the family $\Omega_A$.

Extending the proof to the family $\Omega_B$ requires reproving Lemma \ref{lem:SLLN_memoryless} and a variation of Lemma \ref{lem:diverge_R}, where the ergodicity property (\textsl{A}\ref{cond:omega_inv}) is replaced with the maximality property (\textsl{A}\ref{cond:omega_cap_ach}). This is done via the ergodic decomposition~\cite{hernandez_lasserre} for the associated stationary Markov chain. The proof appears in Appendix \ref{app:lemmas}, Lemma \ref{lem:diverge_R_B}. Achievability for the family $\Omega_C$ has already been established, since $\Omega_C\subset \Omega_A$. The stronger pointwise achievability statement for $\Omega_C$ is obtained via p.h.r. properties of the associated Markov chain, by essentially showing that Lemmas \ref{lem:diverge_eps}, \ref{lem:SLLN_memoryless} and \ref{lem:diverge_R} hold given any fixed message point. The proof appears in Appendix \ref{app:pointwise}, Lemma \ref{lem:pointwise_ext}.

\begin{remark}\label{rem:non_erg}
For $(P_X,P_{Y|X})\in\Omega_B\setminus\Omega_A$, although the unconstrained capacity $C(P_{Y|X})$ is achieved, there is no guarantee on the sample path behavior of the input, which may generally differ from the expected behavior dictated by $P_X$, and depend on the ergodic component the chain lies in. However, if $P_X$ is the \textit{unique} input distribution\footnote{Uniqueness of the capacity achieving distribution for $P_{Y|X}$ does not generally imply the same for the corresponding normalized channel $P_{\Phi|\Theta}$. For example, the normalized channel for a BSC/${\rm Bernoulli}\left(\frac{1}{2}\right)$ pair, has an uncountably infinite number of capacity achieving distributions.} such that $I(X;Y) = C(P_{Y|X})$, then the sample path behavior will nevertheless follow $P_X$ independent of the ergodic component. This is made precise in Appendix \ref{app:lemmas}, Lemma \ref{lem:diverge_R_B}.
\end{remark}

\section{Error Probability Analysis}\label{sec:error_memoryless}
In this section, we provide two sufficient conditions on the target error probability facilitating the achievability of a given rate using the corresponding optimal variable rate decoding rule. The approach here is substantially different from that of the previous subsection, and the derivations are much simpler. However, the obtained result is applicable only to rates below some thresholds $R^*,R^\dagger$. Unfortunately, it is currently unknown under what conditions do these thresholds equal the mutual information, rendering the previous section indispensable.

Loosely speaking, the basic idea is the following. After having observed $\Phi^n$, say the receiver has some estimate $\widehat{\theta}_{n+1}$ for the next input $\Theta_{n+1}$. Then $(\widehat{\theta}_{n+1},\Phi^n)$ correspond to a unique estimate $\widehat{\theta}_0$ of the message point which is recovered by \textit{reversing} the transmission scheme, i.e., running a RIFS over $(0,1)$ generated by the kernel
$\omega_\phi(\cdot)\dfn F_{\Theta|\Phi}^{-1}(\cdot|\phi)$ (the functional inverse of the normalized posterior matching kernel), controlled by the output sequence $\Phi^n$, and initialized at $\widehat{\theta}_{n+1}$. In practice however, the receiver decodes an interval and therefore to attain a specific target error probability $p_e(n)$, one can tentatively decode a subinterval of $\ui$ in which $\Theta_{n+1}$ lies with probability $1-p_e(n)$, which since $\Theta_{n+1}\sim\m{U}$, is any interval of length $1-p_e(n)$. The endpoints of this interval are then ``rolled back" via the RIFS to recover the decoded interval w.r.t. the message point $\Theta_0$. The target error probability decay which facilitates the achievability of a given rate is determined by the convergence rate of the RIFS, which also corresponds to the maximal information rate supported by this analysis.

This general principle relating rate and error probability to the convergence properties of the corresponding RIFS, facilitates the
use of any RIFS contraction condition for convergence. The only limitation stems from the fact that $\omega_\phi(\cdot)$
generating the RIFS is an inverse c.d.f. over the unit interval and hence never globally contractive, so only contraction on the
average conditions can be used. The Theorems appearing below make use of the principle above in conjunction with the contraction
Lemmas mentioned in Section~\ref{sec:prelim}-\ref{subsec:IFS}, to obtain two different expressions tying error probabilities, rate and
transmission period. The discussion above is made precise in the course of the proofs.

Denote the family of all continuous functions $\rho:\ui\mapsto[1,\infty)$ by $\m{C}$.
\begin{theorem}\label{thrm:memoryless_err1}
Let $(P_X, P_{Y|X})$ be an input/channel pair,
$(P_\Theta,P_{\Phi|\Theta})$ the corresponding normalized pair,
and let $\omega_\phi(\cdot)\dfn F_{\Theta|\Phi}^{-1}(\cdot|\phi)$.
For any $\rho\in\m{C}$, define
\begin{align*}
R^\dagger(\rho) \dfn -\log\sup_{s\in(0,1)}
\Expt\left\{
\frac{\rho(\omega_{\Phi}(s))}{\rho(s)}D_s\big{(}\omega_{\Phi}\big{)}\right\}
\end{align*}
where $D_s(\cdot)$ is defined in (\ref{eq:lip_operator}), and let
\begin{align*}
R^{\dagger} \dfn \sup_{\rho\in\m{C}}R^\dagger(\rho)
\end{align*}
If $R^\dagger>0$, then the posterior matching scheme with an optimal variable rate decoding rule achieves any rate $R<R^\dagger$, by setting the target error probability to satisfy
$p_e(n)\rightarrow 0$ under the constraint
\begin{equation}\label{eq:error_exp1}
\Psi\left((1-\alpha)p_e(n),1-\alpha
p_e(n),2^{-R^\dagger({\rho})}\right)
=\lito\left(2^{\,n(R^\dagger({\rho})-R)}\right)
\end{equation}
for some $\alpha\in\ui$, and some $\rho\in\m{C}$ such that $R^\dagger(\rho)>
R$, where $\Psi$ is defined in Lemma \ref{lem:contraction2}.
\end{theorem}
\begin{proof}
Let $\wt{S}_n(s)$ be the RIFS  generated by $\omega_\phi(\cdot)\dfn F_{\Theta|\Phi}^{-1}(\cdot|\phi)$ and the control sequence $\{\Phi_k\}_{k=1}^\infty$, initialized at $s\in\ui$. Select a fixed interval $J_1 = (s,t)\subseteq\ui$ as the decoded interval w.r.t. $\Theta_{n+1}$. Since
$\Theta_{n+1}\sim\m{U}$, we have that
\begin{equation*}
\Prob(\Theta_{n+1}\in J_1) = |J_1|
\end{equation*}
Define the corresponding interval at the origin to be
\begin{equation*}
J_n \dfn (\wt{S}_n(s),\wt{S}_n(t))
\end{equation*}
and set it to be the decoded interval, i.e., $\Delta_n(\Phi^n)=J_n$. Note that the endpoints of
$J_n$ are r.v.'s. Since $F_{\Theta|\Phi}^{-1}(\cdot|\phi)$ is invertible for any $\phi$,
the interval $J_n$ corresponds to $\Theta_1$, namely,
\begin{equation*}
\Prob(\Theta_1\in J_n) = \Expt\Prob(\Theta_1\in
J_n\,|\, \Phi^n) = \Expt\Prob(\Theta_{n+1}\in
J_1\,|\, \Phi^n)
= \Prob(\Theta_{n+1}\in J_1)
= |J_1|
\end{equation*}
and then in particular (recall that $\Theta_0=\Theta_1$)
\begin{equation*}
p_e(n) = \Prob(\Theta_0\not\in \Delta(\Phi^n)) = 1-|J_1|
\end{equation*}
For a variable rate decoding rule, the target error probability is set in advance. Therefore, given $p_e(n)$ the length of the interval $J_1$ is constrained to be $|J_1| = 1-p_e(n)$,
and so without loss of generality we can parameterize the endpoints of $J_1$ by
\begin{equation*}
(s,t) = ((1-\alpha)p_e(n),1-\alpha p_e(n))
\end{equation*}
for some $\alpha\in(0,1)$.

Now let $\rho\in\m{C}$, and define
\begin{equation*}
r(\rho) \dfn \sup_{s\in(0,1)} \Expt\left\{
\frac{\rho(\omega_{\Phi}(s))}{\rho(s)}D_s\big{(}\omega_{\Phi}\big{)}\right\}
\end{equation*}
Note that the expectation above is taken w.r.t. $\Phi\sim\m{U}$. Using Lemma \ref{lem:contraction2}, if $r(\rho)<1$ then
\begin{equation*}
\Prob\left(\left|J_{1}\right|>\eps\right)
=\Prob\left(\left|\wt{S}_n(s)-\wt{S}_n(t)\right|>\eps\right)
\leq
\eps^{-1}\Psi(s,t,r(\rho))\cdot r^n(\rho)
\end{equation*}
To find the probability that the decoded interval is larger than
$2^{-nR}$, we substitute $\eps=2^{-nR}$ and obtain
\begin{equation*}
\Prob(R_n<R) = \Prob\left(\left|J_{1}\right|>2^{-nR}\right)
\leq
2^{nR}\cdot \Psi((1-\alpha)p_e(n),1-\alpha p_e(n),r(\rho))\cdot
2^{n\log{r(\rho)}}
\end{equation*}
Following the above and defining $R^\dagger(\rho)=-\log r(\rho)$,
a sufficient condition for $\Prob(R_n<R)\rightarrow 0$ for
$R<R^\dagger(\rho)$ is given by (\ref{eq:error_exp1}). The proof
is concluded by taking the supremum over $\rho\in\m{C}$, and noting that if no $\rho$ results in a
contraction then $R^\dagger\leq 0$.
\end{proof}

Theorem \ref{thrm:memoryless_err1} is very general in the sense of
not imposing any constrains on the input/channel pair. It is
however rather difficult to identify a \textit{weight function}
$\rho$ that will result in $R^\dagger(\rho)>0$. Our next error
probability result is less general (e.g., does not apply to
discrete alphabets), yet is much easier to work with. Although it
also involves an optimization step over a set of functions, it is
usually easier to find a function which results in a positive rate
as the examples that follow demonstrate.

The basic idea is similar only now we essentially work with the
original chain and so the RIFS evolves over $\isupp(X)$, generated
by the kernel $\omega_y(\cdot)\dfn F_{X|Y}^{-1}(\cdot|y)\circ F_X$
(the functional inverse of the posterior matching kernel), and
controlled by the i.i.d. output sequence $\{Y_k\}_{k=1}^\infty$.
To state the result we need some definitions first. Let $P_X$ be
some input distribution and let $\rho:\isupp(X)\mapsto (a,b)$ be
differentiable and monotonically increasing ($a,b$ may be
infinite). The family of all such functions $\rho$ for which
$f_{\rho(X)}$ is bounded is denoted by $\m{F}(X)$. Furthermore,
for a proper r.v. $X$ with a support over a (possibly infinite)
interval, we define the \textit{tail function}
$\m{T}_X:\RealF^+\mapsto[0,1]$ to be
\begin{equation*}
\m{T}_X(\ell) \dfn 1-\sup\big{\{}P_X\big{(}(x,x+\ell)\big{)}\,:\,x\in\RealF\big{\}}
\end{equation*}
Namely, $\m{T}_X(\ell)$ is the minimal probability that can be assigned by $P_X$ outside an open interval of length $\ell$.

\begin{theorem}\label{thrm:memoryless_err2}
Let $(P_X,P_{Y|X})$ be an input/channel pair with $f_{XY}$ continuous
over $\isupp(X,Y)$, and let $\omega_y(\cdot)\dfn
F_{X|Y}^{-1}(\cdot|y)\circ F_X$. For any $\rho\in\m{F}(X)$, define
\begin{equation*}
R^*(\rho)\dfn\hspace{-3pt}\lim_{q\rightarrow
0^+}\hspace{-2pt}\inf_{\stackrel{{\scriptstyle s,t\in\range(\rho)}}{s\neq t}}
\hspace{-4pt}\left(- q^{-1}\log\Expt \left[D_{s,t}(\rho\circ
\omega_{\sst{Y}}\circ\rho^{-1})\right]^q\right)
\end{equation*}
and let
\begin{equation*}
R^* \dfn \sup_{\rho\in\m{F}(X)} R^*(\rho)
\end{equation*}
The following statements hold:
\begin{enumerate}[(i)]
\item \label{stat:achv} The posterior matching scheme with an
optimal variable rate decoding rule achieves any rate $R<R^*$, by
setting the target error probability to satisfy $p_e(n)\rightarrow
0$ under the constraint
\begin{equation}\label{eq:target_err1}
p_e(n) = \m{T}_{\rho(X)}\left(\lito\left(2^{n(R^*(\rho)-R)}\right)\right)
\end{equation}
for some $\rho\in\m{F}(X)$ satisfying $R^*(\rho)> R$.

\item \label{stat:zero_err} If $|\range(\rho)|<\infty$ then any
rate $R<R^*(\rho)$ can be achieved with zero error
probability.\footnote{This is not a standard zero-error achievability claim, since the rate is generally random. If a fixed rate must be guaranteed, then the error probability will be equal to the probability of "outage", i.e., the probability that the variable decoding rate falls below the rate threshold.}

\item \label{stat:seperable} If it is possible to write
$\rho\circ\omega_y\circ\rho^{-1}(s) = u(s)v(y)+q(y)$, then
\begin{equation*}
R^*(\rho) = -\Expt\log{|v(Y)|}
-\log\sup_{s\in\range(\rho)}\left|u'(s)\right|
\end{equation*}
whenever the right-hand-side exists.
\end{enumerate}

\end{theorem}

\begin{proof}
We first prove the three statements in the special case where $P_X$ has a support over a (possibly infinite) interval, and considering only the identity function $\rho_1:\isupp(X)\mapsto\isupp(X)$ over the support, i.e., discussing the achievability of $R^*(\rho_1)$ exclusively. We therefore implicitly assume here that $\rho_1\in\m{F}(X)$. Let $\wt{S}_n(s)$ be the RIFS generated by $\omega_y(\cdot)\dfn F_{X|Y}^{-1}(\cdot|y)\circ F_X$ and the control sequence $\{Y_k\}_{k=1}^\infty$, which evolves over the space $\isupp(X)$. Select a fixed interval $J_1 = (s,t)\subseteq\isupp(X)$
as the decoded interval w.r.t. $X_{n+1}$. Since $X_{n+1}\sim P_X$ we have that
\begin{equation*}
\Prob(X_{n+1}\in J_1) = P_X\big{(}J_1\big{)}
\end{equation*}
Define the corresponding interval at the origin to be
\begin{equation*}
J_n \dfn (\wt{S}_n(s),\wt{S}_n(t))
\end{equation*}
and following the same lines as in the proof of the preceding Theorem, $J_n$ is set to be the decoded interval w.r.t. $X_1 = F_X(\Theta_0)$, and so the decoded interval for $\Theta_0$ is set to be $\Delta_n(Y^n) = F_X(J_n)$. Thus,
\begin{equation*}
p_e(n) = \Prob(\Theta_0\not\in F_X(J_n))  = \Prob(X_1\not\in J_n) = 1-
P_X\big{(}J_1\big{)}
\end{equation*}

For any $q>0$ define
\begin{equation*}
r_q \dfn \sup_{s\neq
t\in\isupp(X)}\Expt\left[D_{s,t}(\omega_{\sst{Y}})\right]^q
\end{equation*}
Using Jensen's inequality we have that for any $0<q\leq p$
\begin{equation}\label{r_qp_bound}
r_q = \sup_{s\neq
t}\Expt\left[D_{s,t}(\omega_{\sst{Y}})\right]^q = \sup_{s\neq
t}\Expt\left[D_{s,t}(\omega_{\sst{Y}})\right]^{p\frac{q}{p}}
\leq
\sup_{s\neq
t}\left(\Expt\left[D_{s,t}(\omega_{\sst{Y}})\right]^{p}\right)^\frac{q}{p}
= \left(\sup_{s\neq
t}\Expt\left[D_{s,t}(\omega_{\sst{Y}})\right]^{p}\right)^\frac{q}{p} = \left(r_{p}\right)^\frac{q}{p}
\end{equation}
Now suppose there exists some $q^*>0$ so that $r_{q^*}<1$. Using
(\ref{r_qp_bound}) we conclude that $r_q<1$ for any $0<q\leq q^*$,
and using Lemma \ref{lem:contraction1} we have that for any
$0<q\leq q^*$ and any $\eps>0$
\begin{equation*}
\Prob(|\wt{S}_n(s)-\wt{S}_n(t)|>\eps) \leq
\eps^{-q}|s-t|^qr_q^n
\end{equation*}
and thus
\begin{equation*}
\Prob(R_n<R)  = \Prob\left(
P_X\big{(}(\wt{S}_n(s),\wt{S}_n(t))\big{)}>2^{-nR}\right)
\leq
\Prob(M\cdot|J_n|>2^{-nR}) \leq
M^{-q}2^{nRq}|J_1|^qr_q^n
\end{equation*}
where $M\dfn\sup f_X(x)$. A sufficient condition for
$\,\Prob(R_n<R)\rightarrow 0$ is given by
\begin{equation*}
|J_1| = \lito\left(2^{n(q^{-1}\log{r_q^{-1}}-R)}\right)
\end{equation*}
Since the above depends only on the length of $J_1$, we
can optimize over its position to obtain $p_e(n) = 1-
P_X\big{(}J_1\big{)} = \m{T}_X(|J_1|)$, or arbitrarily close to that.
Therefore, any rate $R<q^{-1}\log{r_q^{-1}}$ is achievable by
setting $p_e(n)\rightarrow 0$ under the constraint
\begin{equation*}
p_e(n) = \m{T}_X(|J_1|) =
\m{T}_X\left(\lito\left(2^{n(q^{-1}\log{r_q^{-1}}-R)}\right)\right)
\end{equation*}

We would now like to maximize the term $q^{-1}\log{r_q^{-1}}$ over
the selection of $0<q\leq q^*$. Using (\ref{r_qp_bound}) we obtain
\begin{equation*}
q^{-1}\log{(r_q)^{-1}} \geq q^{-1}\log{(r_p)^{-\frac{q}{p}}} \geq
p^{-1}\log{(r_p)^{-1}}
\end{equation*}
and so $q^{-1}\log{r_q^{-1}}$ is nonincreasing with $q$, thus
\begin{align*}
\sup_{0<q\leq q^*}q^{-1}\log{r_q^{-1}} = \lim_{q\rightarrow
0^+}q^{-1}\log{r_q^{-1}}  = R^*(\rho_1)
\end{align*}
where $\rho_1$ is the identity function over $\isupp(X)$. From the
discussion above it is easily verified that $R^*(\rho_1)> 0$ iff
$r_{q^*}<1$ for some $q^*>0$. Moreover, if
$|\isupp(X)|=M_0<\infty$ then $\m{T}_X(\ell)=0$ for any
$\ell>M_0$, therefore in this case $p_e(n)=0$ for any $n$ large
enough. Note that since $\rho_1$ is defined only over $\isupp(X)$,
we have that $|\range(\rho_1)|=|\isupp(X)|=M_0$. Thus, statements
(\ref{stat:achv}) and (\ref{stat:zero_err}) are established for an
input distribution with support over an interval, and the specific
selection of the identity function $\rho=\rho_1$.

As for statement (\ref{stat:seperable}), note first that since
$f_{XY}$ is continuous then $\omega_y(s)$ is jointly
differentiable in $y,s$. Suppose that $\omega_y(s)=u(s)v(y)+q(y)$,
and so $u,v,q$ are all differentiable. In this separable case we
have
\begin{align*}
r_q = \sup_{s\neq
t}\Expt\left(|v(Y)|\cdot\frac{|u(t)-u(s)|}{|t-s|}\right)^q
=
\Expt|v(Y)|^q\cdot\sup_{s\neq
t}\left|\frac{u(t)-u(s)}{t-s}\right|^q
=
\Expt|v(Y)|^q\cdot\sup_{s}|u'(s)|^q
\end{align*}
and
\begin{equation*}
q^{-1}\log{r_q^{-1}} = -q^{-1}\log\Expt|v(Y)|^q -\sup_s\log|u'(s)|
\leq  -\Expt\log{|v(Y)|} -\sup_s\log|u'(s)|
\end{equation*}
where we have used Jensen's inequality in the last inequality. We
now show that the limit of the left-hand-side above as $q
\rightarrow 0^+$ in fact attains the right-hand-side bound
(assuming it exists), which is similar to the derivation of the
Shannon entropy as a limit of R\'enyi entropies. Since
$\Expt\log{|v(Y)|}$ is assumed to exist then we have
$\log\Expt|v(Y)|^q\rightarrow 0$ as $q\rightarrow 0^+$, and so to
take the limit we need to use L'Hospital's rule.  To that end, for
any $0<q\leq q^*$
\begin{align*}
\frac{d}{dq}\,\Expt|v(Y)|^q &= \frac{d}{dq}\int f_Y(y)|v(Y)|^q dy
= \int \frac{\partial}{\partial q}f_Y(y)|v(Y)|^q dy
=
\log{e}\cdot\int f_Y(y)|v(Y)|^q\log{|v(Y)|} \,dy
\\
&= \log{e}\cdot\Expt\Big{(}|v(Y)|^q\log{|v(Y)|}\Big{)}
\end{align*}
and thus
\begin{align*}
R^*(\rho_1) &= \lim_{q\rightarrow
0^+}\left(-q^{-1}\log\Expt|v(Y)|^q -\sup_s\log|u'(s)|\right)
=
\lim_{q\rightarrow 0^+}\left(-\frac{d}{dq}\,\log\Expt|v(Y)|^q
\right) -\sup_s\log|u'(s)|
\\
&= \lim_{q\rightarrow
0^+}\left(-\frac{\Expt\big{(}|v(Y)|^q\log{|v(Y)|}\big{)}}{\Expt|v(Y)|^q}\right)
-\sup_s\log|u'(s)|
= -\Expt\log{|v(Y)|} -\sup_s\log|u'(s)|
\end{align*}
Which established statement (\ref{stat:seperable}) in the special case under discussion. The derivations above all hold under the assumption that the right-hand-side above exists.

Treating the general case is now a simple extension. Consider a general input distribution $ P_X$ (with a p.d.f. continuous over its support), and a differentiable and monotonically increasing
function $\rho:\isupp(X)\mapsto (a,b)$. Let us define a \textit{$\rho$-normalized channel} $ P_{Y^\rho|X^\rho}$ by connecting the operator $\rho^{-1}(\cdot)$ to the channel's input.
Let us Consider the posterior matching scheme for the $\rho$-normalized input/channel pair $(P_{\rho(X)},P_{Y^\rho|X^\rho})$. Using the monotonicity of $\rho$, the corresponding input and inverse channel c.d.f's are given by
\begin{equation*}
F_{\rho(X)} = F_X\circ\rho^{-1}\,,\qquad
F_{X^\rho|Y^\rho}(\cdot|y) = F_{X|Y}(\cdot|y)\circ\rho^{-1}
\end{equation*}
The posterior matching kernel is therefore given by
\begin{equation*}
F^{-1}_{\rho(X)}\circ F_{X^\rho|Y^\rho}(\cdot|y) = \rho\circ
\Big{(}F^{-1}_X\circ F_{X|Y}(\cdot|y)\Big{)}\circ\rho^{-1}
\end{equation*}
and the corresponding RIFS kernel is the functional inverse of the above, i.e.,
\begin{equation*}
\Big{(}F^{-1}_{\rho(X)}\circ
F_{X^\rho|Y^\rho}(\cdot|y)\Big{)}^{-1} = \rho\circ
\Big{(}F^{-1}_{X|Y}(\cdot|y)\circ F_X\Big{)}\circ\rho^{-1}
=
\rho\circ\omega_y\circ\rho^{-1}
\end{equation*}
Now, using the monotonicity of $\rho$ it is readily verified that the input/channel pairs $(P_X,P_{Y|X})$ and $(P_{\rho(X)},P_{Y^\rho|X^\rho})$ correspond to the same normalized channel. Hence, the corresponding posterior matching schemes are equivalent, in the sense that $\{(X_k,Y_k)\}_{k=1}^\infty$ and $\{(\rho^{-1}(X^\rho_k),Y^\rho_k)\}_{k=1}^\infty$ have the same joint distribution. Therefore, the preceding analysis holds for the input/channel pair $(P_{\rho(X)}, P_{Y^\rho|X^\rho})$, and the result follows immediately.
\end{proof}

Loosely speaking, the optimization step in both Theorems has a
similar task -- changing the scale by which distances are measured
so that the RIFS kernel appears contractive. In Theorem
\ref{thrm:memoryless_err1}, the weight functions multiply the
local slope of the RIFS. In Theorem \ref{thrm:memoryless_err2} the
approach is in a sense complementing, since the functions are applied
to the RIFS kernel itself, thereby shaping the slopes directly.
These functions will therefore be referred to as \textit{shaping
functions}.

\vspace{9pt}\noindent{\bf Example \ref{ex:AWGN} (AWGN, continued)}. Returning to the AWGN channel setting with a Gaussian input, we can now determine the tradeoff between rate, error
probability and transmission period obtained by the Schalkwijk-Kailath scheme. Inverting the kernel (\ref{eq:AWGN_kernel}) we obtain the RIFS kernel
\begin{equation*}
\omega_y(s)= F_{X|Y}^{-1}(F_X(s)|y) =
\frac{s}{\sqrt{1+\SNR}}+\frac{\SNR}{1+\SNR}\;y
\end{equation*}
Setting the identity shaping function $\rho_1(s) = s$, the
condition of Theorem \ref{thrm:memoryless_err2} statement
(\ref{stat:seperable}) holds and so
\begin{equation*}
R^*(\rho_1) = -\log\sup_{s\in\RealF}
\left|\frac{d}{ds}\left(\frac{s}{\sqrt{1+\SNR}}\right)\right| =
\frac{1}{2}\log (1+\SNR) =C
\end{equation*}
so in this case $R^* = C$, and statement (\ref{stat:achv})
reconfirms that the Schalkwijk-Kailath scheme achieves capacity.
Using standard bounds for the Gaussian distribution, the Gaussian
tail function (for the input distribution) satisfies
\begin{equation*}
\m{T}_X(\ell) = \bigo\left(e^{-\frac{\ell^2}{8P}}\right)
\end{equation*}
Plugging the above into (\ref{eq:target_err1}), we find that a
rate $R<C$ is achievable by setting the target error probability
to
\begin{equation*}
-\log p_e(n)  =
-\log\m{T}_X\left(\lito\left(2^{n(R^*-R)}\right)\right) =
\lito\left(2^{2n(C-R)}\right)
\end{equation*}
recovering the well known double-exponential behavior. Note that since the interval contraction factor in this case is independent of the output sequence, the variable-rate decoding rule is in fact fixed-rate, hence the same double-exponential performance is obtained using a fixed-rate decoding rule.

We mention here the well known fact that for the AWGN channel, the error probability can be made to decay as a higher order exponential in the block length, via adaptations of the Schalkwijk-Kailath scheme~\cite{Kramer69,Gallager08}. These adaptations exploit the discreteness of the message set especially at the last stages of transmission, and are not directly applicable within our framework, since we define error probability in terms of intervals and not discrete messages. They can only be applied to the equivalent standard scheme obtained via Lemma \ref{lem:setting_eqv}.

\vspace{9pt}\noindent{\bf Example \ref{ex:BSC} (BSC, continued)}.
The conditions of Theorem \ref{thrm:memoryless_err2} are not
satisfied in the BSC setting, and we resort to Theorem
\ref{thrm:memoryless_err1}. Inverting the posterior matching
kernel (\ref{eq:BSC_kernel}) pertaining to the corresponding
normalized channel, we obtain the RIFS kernel
\begin{equation*}
\omega_{\phi}(s) = F_{\Theta|\Phi}^{-1}(s|\phi) = \left\{
    \begin{array}{ll}
    \frac{s}{2(1-p)} & s\in (0,1-p),\phi\in(0,\frac{1}{2})\\ \frac{s-(1-2p)}{2p} & s\in[1-p,1),\phi\in(0,\frac{1}{2})
    \\
    \frac{s}{2p} & s\in(0,p),\phi\in[\frac{1}{2},1) \\ \frac{s+(1-2p)}{2(1-p)} & s\in[p,1),\phi\in[\frac{1}{2},1)
    \end{array}\right.
\end{equation*}
and
\begin{equation*}
D_s\left(\omega_{\phi}\right) = \left\{
    \begin{array}{ll}
    \frac{1}{2(1-p)} & s\in (0,1-p),\phi\in(0,\frac{1}{2})\\ \frac{1}{2p} & s\in[1-p,1),\phi\in(0,\frac{1}{2})
    \\
    \frac{1}{2p} & s\in(0,p),\phi\in[\frac{1}{2},1) \\ \frac{1}{2(1-p)} & s\in[p,1),\phi\in[\frac{1}{2},1)
    \end{array}\right.
\end{equation*}
Using a constant weight function (i.e., no weights) does not work
in this case, since the average of slopes for (say) $s\in(0,p)$,
is
\begin{equation*}
\Expt D_s(\omega_\Phi) =
\frac{1}{2}\left(\frac{1}{2p}+\frac{1}{2(1-p)}\right) \geq 1
\end{equation*}
In fact, any bounded weight function will result in the same
problem for $s>0$ small enough, which suggests that the weight
function should diverge to infinity as $s\rightarrow 0$. Setting
$\rho(s)=s^{-\beta}$ for $\beta>1$ is a good selection for
$s\in(0,p)$ since in that case
\begin{equation*}
\Expt \left(\frac{\rho(\omega_{\Phi}(s))}{\rho(s)} \cdot
D_s(\omega_\Phi)\right) =
\frac{1}{2}\left((2p)^{\beta-1}+(2(1-p))^{\beta-1}\right),
\end{equation*}
which can be made smaller than unity by properly selecting
$\beta$. Setting $\rho$ symmetric around $\frac{1}{2}$
duplicates the above to $s\in(1-p,1)$. However, this selection
(and some variants) do not seem to work in the range
$s\in(p,\frac{1}{2})$, for which $\beta\leq 1$ is required.
Finding a weight function $\rho$ for which $R^\dagger(\rho)>0$ (if
exists at all) seems to be a difficult task, which we were unable
to accomplish thus far.

\vspace{9pt}\noindent{\bf Example \ref{ex:uniform} (Uniform
input/noise, continued)}. We have already seen that achieving the
mutual information with zero error decoding is possible in the
uniform noise/input setting. Let us now derive this fact via
Theorem \ref{thrm:memoryless_err2}. The output p.d.f. is given by
\begin{equation*}
f_Y(y) = y\ind_{(0,1]}(y)+(2-y)\ind_{(1,2)}(y)
\end{equation*}
The RIFS kernel is obtained by inverting the posterior matching
kernel (\ref{eq:uniform_kernel}), which yields
\begin{equation*}
\omega_y(s) =F^{-1}_{X|Y}(F_X(s)|y)
=
s\left(y\ind_{(0,1]}(y)+(2-y)\ind_{(1,2)}(y)\right) +
(y-1)\ind_{(1,2)}(y)
= sf_Y(y) + (y-1)\ind_{(1,2)}(y)
\end{equation*}
Using the identity shaping function $\rho_1$ again (but now
restricted to $\isupp(X)=\ui$), the condition of statement
(\ref{stat:seperable}) holds and therefore
\begin{align*}
R^*(\rho_1) &= -\Expt\log f_Y(Y) - \sup_{s\in(0,1)}\log{1} = h(Y)
= I(X;Y)
\end{align*}
and we have $R^*=R^*(\rho_1) = I(X;Y)$, thereby verifying once
again that the mutual information is achievable. Since
$\range(\rho_1) = \ui$ is bounded, statement (\ref{stat:zero_err})
reconfirms that variable-rate zero error decoding is possible.

\vspace{9pt}\noindent{\bf Example \ref{ex:exp} (Exponential
input/noise, continued)}. Let us return to the additive noise
channel with an exponentially distributed noise and input. We have
already seen that the posterior matching scheme
(\ref{eq:exp_scheme}) achieves the mutual information, which in
this case is $I(X;Y) \approx 0.8327$. The p.d.f. of the
corresponding output is
\begin{equation*}
f_Y(y) = ye^{-y}\ind_{(0,\infty)}(y)
\end{equation*}
It is easily verified that $F_{X|Y}^{-1}(s|\,y) = sy$, and so the
RIFS kernel is given by
\begin{equation*}
\omega_y(s)=F^{-1}_{X|Y}(F_X(s)|y) = y(1-e^{-s})
\end{equation*}
Now, using Theorem \ref{thrm:memoryless_err2} with the identity
shaping function $\rho_1$ restricted to $\isupp(X)=(0,\infty)$,
the condition of statement (\ref{stat:seperable}) holds and
therefore
\begin{equation*}
R^*(\rho_1) =  -\Expt\log{Y} -
\log\sup_{s\in(0,\infty)}\left|\frac{d}{ds}(1-e^{-s})\right| =
-\Expt\log{Y}
\approx -0.61 < 0
\end{equation*}
Thus, the identity function is not a good choice in this case, and
we must look for a different shaping function. Let us set
$\rho_2(s) = s^{-\frac{1}{2}}$, which results in a p.d.f. and
c.d.f.
\begin{equation*}
f_{\rho_2(X)}(s) = \frac{1}{2s^3}\,\exp\left(-s^{-2}\right)\qquad
F_{\rho_2(X)}(s) = \exp\left(-s^{-2}\right)
\end{equation*}
and
\begin{equation*}
\rho_2\circ\omega_y\circ\rho_2^{-1}(s) =
\left[y\left(1-\exp\left(-s^{-2}\right)\right)\right]^{-\frac{1}{2}}
\end{equation*}
Since $f_{\rho_2(X)}$ is bounded and the above again satisfies the
condition of statement (\ref{stat:seperable}), we obtain
{\allowdisplaybreaks
\begin{align*}
R^*(\rho_2) &=  \frac{1}{2}\,\Expt\log{Y} \hspace{-3pt}-
\log\hspace{-3pt}\sup_{s\in(0,\infty)}\hspace{-2pt}\left|\frac{d}{ds}\left[\left(1-\exp\left(-s^{-2}\right)\right)\right]^{-\frac{1}{2}}\right|
\\
&=\frac{1}{2}\,\Expt\log{Y}+\inf_{s\in[0,\infty)}\left(\frac{\log{e}}{s^2}+3\log{s}+\frac{3}{2}\log(1-\exp(-s^{-2}))\right)
= \frac{1}{2}\,\Expt\log{Y} \approx 0.305
\end{align*}}
where the infimum above is attained as $s\rightarrow \infty$. The
tail function of $ P_{\rho_2(X)}$ is bounded by
\begin{equation*}
\m{T}_{\rho_2(X)}(\ell) \leq 1-\exp\left(-\ell^{-2}\right)
\leq \ell^{-2}
\end{equation*}
Thus, any rate $R<R^*(\rho_2)\approx 0.305$ is achieved by the
posterior matching scheme \ref{ex:exp} using a variable decoding
rule if the target error probability is set to
\begin{equation*}
p_e(n) = \frac{1}{\lito\left(2^{2n(R^*(\rho_2)-R)}\right)}
\end{equation*}
and so the following error exponent is achievable:
\begin{equation*}
\lim_{n\rightarrow \infty}\frac{1}{n}\log \frac{1}{p_e(n)} =
2(R^*(\rho_2)-R) \approx 0.61-2R
\end{equation*}
Although we know from Theorem \ref{thrm:achv} that any rate up to
the mutual information is achieved in this case, $\rho_2(\cdot)$
is the best shaping function we have found, and so our error
analysis is valid only up to the rate $R^*(\rho_2)\approx 0.305
<I(X;Y)$.

\section{Extensions}\label{sec:extensions}
\subsection{The $\mu$-Variants of the Posterior Matching Scheme}\label{subsec:mu_var}
In this subsection we return to discuss the $\mu$-variants (\ref{eq:PM_mu_variants}) of the baseline posterior matching scheme addressed thus far. To understand why these variants are of interest, let us first establish the necessity of a fixed-point free kernel (thereby also proving Lemma \ref{lem:egr_implies_fixed_free}).
\begin{lemma}\label{lem:fixed_point_no_rate}
If (\textsl{A}\ref{cond:omega_fixedp}) does not hold, then (\textsl{A}\ref{cond:omega_inv}) does not hold either and the corresponding scheme cannot achieve any positive rate.
\end{lemma}
\begin{proof}
By the assumption in the Lemma, there must exists some fixed-point $\theta_f\in\ui$ such that
\begin{equation*}
\Prob\left(F_{\Theta|\Phi}(\theta_f|\Phi)=\theta_f\right) = 1
\end{equation*}
The posterior c.d.f. $F_{\Theta_0|\Phi^n}(\theta|\phi^n)$ is obtained by an iterated composition of the kernel $P_{\Theta|\Phi}(\theta|\phi)$ controlled by the i.i.d. output sequence $\Phi^n$. Thus, the fixed point at $\theta_f$ induces a fixed point for the posterior c.d.f at $\theta_f$ as well, since
\begin{equation*}
\Prob\left(F_{\Theta_0|\Phi^n}(\theta_f|\Phi^n)=\theta_f\right)  \geq \prod_{k=1}^n\Prob\left(F_{\Theta_0|\Phi}(\theta_f|\Phi_k)=\theta_f\right) = 1
\end{equation*}
This immediately implies that no positive rate can be achieved, since the posterior probability of the interval $(0,\theta_f)$ remains fixed at $\theta_f$. Stated differently, this simply means that the output sequence provides no information regarding whether $\Theta_0<\theta_f$ or not. For practically the same reason, the invariant distribution $P_{\Theta\Phi}$ for the Markov chain $\{(\Theta_n,\Phi_n)\}_{n=1}^\infty$ is not ergodic, since the set $(0,\theta_f)\times \ui$ is invariant yet $0<P_{\Theta\Phi}\left((0,\theta_f)\times \ui\right) = \theta_f < 1$.
\end{proof}

Suppose our kernel has $L$ fixed points, and so following the above the unit interval can be partitioned into a total of $L+1$ corresponding \textit{invariant intervals}. One external way to try and handle the fixed-point problem is to decode a disjoint union of $L+1$ exponentially small intervals (one per invariant interval) in which the message point lies with high probability, and then resolve the remaining ambiguity using some simple non-feedback zero-rate code. This seems reasonable, yet there are two caveats. First, the maximal achievable rate in an invariant interval may generally be smaller than the mutual information, incurring a penalty in rate. Second, the invariant distribution $P_{\Theta\Phi}$ is not ergodic, and it is likely that any encapsulated input constraints will not be satisfied (i.e., not pathwise but only in expectation over invariant intervals). A better idea is to map our message into the invariant interval with the maximal achievable rate, which is always at least as high as the mutual information. This corresponds to a posterior matching scheme with a different input distribution (using only some of the inputs), and resolves the rate problem, but not the input constraint problem. We must therefore look for a different type of solution.

Fortunately, it turns out that the fixed points phenomena is in many cases just an artifact of the specific ordering imposed on the inputs, induced by the selection of the posterior c.d.f. in the posterior matching rule. In many cases, imposing a different ordering can eliminate this artifact altogether. We have already encountered that in the DMC setting (Example \ref{ex:gen_DMC} in Section~\ref{sec:main_result_memoryless}, using Lemma \ref{lem:dmc_prop}), where in the case a fixed point exists, a simple input permutation was shown to be sufficient in order for the posterior matching scheme (matched to the equivalent input/channel pair) to achieve capacity. This permutation can be interpreted as inducing a different order over the inputs, and the scheme for the equivalent pair can be interpreted as a specific $\mu$-variant of the original scheme.

These observations provide motivation to extend the notion of equivalence between input/channel pairs from the discrete case to the general case. Two  input/channel pairs $(P_X, P_{Y|X})$ and $(P_{X^*},P_{Y^*|X^*})$ are said to be \textit{equivalent} if there exist u.p.f's $\mu,\sigma$ such that the corresponding normalized channels satisfy
\begin{equation*}
P_{\Phi^*|\Theta^*}(\cdot|\theta) = P_{\Phi|\Theta}(\sigma(\cdot)|\mu(\theta))
\end{equation*}
for any $\theta\in\ui$. This practically means that the asterisked normalized channel is obtained by applying $\mu$ and $\sigma^{-1}$ to the input and output of the asterisk-free normalized channel, respectively, and in this case we also say that the pair $(P_X, P_{Y|X})$ is \textit{$\mu$-related} to the pair $(P_{X^*},P_{Y^*|X^*})$. Again, equivalent input/channel pairs have the same mutual information. Following this, for every u.p.f. $\mu$ and every set of input/channel pairs $\Gamma$, we define $\mu(\Gamma)$ to be the set of all input/channel pairs to which some pair in $\Gamma$ is $\mu$-related.
The following result follows through immediately from the developments in Sections \ref{sec:scheme} and \ref{sec:main_result_memoryless}, and the discussion above.
\begin{theorem}\label{thrm:mu_variants}
For any input/channel pair $(P_X,P_{Y|X})$ and any u.p.f. $\mu$, the corresponding $\mu$-variant posterior matching scheme (\ref{eq:PM_mu_variants}) has the following properties:
\begin{enumerate}[(i)]
\item It admits a recursive representation w.r.t. the normalized channel, with a kernel $\mu\circ F_{\mu^{-1}(\Theta)|\Phi}(\cdot|\phi)\circ \mu^{-1}$, i.e.,
    \begin{equation*}
        \Theta_1 = \mu(\Theta_0)\,,\quad \Theta_{n+1} = \mu\circ F_{\mu^{-1}(\Theta)|\Phi}(\cdot|\Phi_n)\circ \mu^{-1}(\Theta_n)
    \end{equation*}

\item If $(P_X,P_{Y|X})\in\mu\left(\Omega_A\cup\Omega_B\right)$ (resp. $\mu\left(\Omega_C\right)$), the scheme achieves (resp. pointwise achieves) any rate $R<I(X;Y)$ over the channel $P_{Y|X}$. Furthermore, if $(P_X, P_{Y|X})\in \mu\left(\Omega_A\cup\Omega_C\right)$ then this is achieved within an input constraint $(\eta,\Expt\eta(X))$, for any measurable $\eta:\m{X}\mapsto\RealF$ satisfying $\Expt|\eta(X)|<\infty$.
\end{enumerate}
\end{theorem}
Theorem \ref{thrm:mu_variants} expands the set of input/channel pairs for which some variant of the posterior matching scheme achieves the mutual information, by allowing different orderings of the inputs to eliminate the fixed point phenomena. For the DMC case, we have already seen that considering $\mu$-variants is sometimes crucial for achieving capacity. Next we describe perhaps a more lucid (although very synthetic) example, making the same point for continuous alphabets.
\begin{example}\label{ex:fixed_points}
Let the memoryless channel $P_{Y|X}$ be defined by the following input to output relation:
\begin{equation*}
Y = X^2 + Z
\end{equation*}
where the noise $Z$ is statistically independent of the input $X$. Suppose that some input constraints are imposed so that the capacity is finite, and also such that the capacity achieving distribution does not have a mass point at zero. Now assume that an input zero mean constraint is additionally imposed. It is easy to see that the capacity achieving distribution $P_X$ is now symmetric around zero, i.e., $P_X((-\infty,0))=P_X((0,\infty)) = \frac{1}{2}$ . It is immediately clear that the output of the channel provides no information regarding the sign of the input, hence the corresponding posterior matching kernel $F_X^{-1}\circ F_{X|Y}(\cdot|y)$ has a fixed point at the origin, and equivalently, the normalized kernel $F_{\Theta|\Phi}(\cdot|\phi)$ has a fixed point at $\theta=\frac{1}{2}$. Thus, by Lemma \ref{lem:fixed_point_no_rate} the scheme cannot attain any positive rate. Intuitively, this stems from the fact that information has been coded in the sign of the input, or the most-significant-bit of the message point, which cannot be recovered. To circumvent this problem we can change the ordering of the input, which is effectively achieved by using one of the $\mu$-variants of the posterior matching scheme. For example, set
\begin{equation*}
\mu(\theta) = \left\{\begin{array}{cc}\theta+\frac{1}{3} & \theta\in(0,\frac{1}{3}]\\ \theta-\frac{1}{3} & \theta\in(\frac{1}{3},\frac{2}{3}] \\ \theta & \theta\in(\frac{2}{3},1) \end{array}\right.
\end{equation*}
and use the corresponding $\mu$-variant scheme. This maintains the same input distribution while breaking the symmetry around $\frac{1}{2}$, and eliminating the fixed point phenomena. This $\mu$-variant scheme can therefore achieve the mutual information, assuming all the other conditions are satisfied.
\end{example}

\subsection{Channel Model Mismatch}\label{subsec:model_mismatch}
In this subsection we discuss the model mismatch case, where the
scheme is designed according to the wrong channel model. We assume
that the transmitter and receiver are both unaware of the
situation, or at least do not take advantage of it. To that end,
for any pair $( P_X, P_{Y|X})\in\Omega_C$ we define a
\textit{mismatch set} $\Omega_C^{\rm mis}( P_X, P_{Y|X})$
consisting of all input/channel pairs $( P_{X^*}, P_{Y^*|X^*})$,
with a corresponding normalized channel $ P_{\Phi^*|\Theta^*}$,
that admit the following properties:
\begin{enumerate}[\bf(\textsl{C}1)]
\item \label{cond:M_reg}\textit{$(P_{X^*}, P_{Y^*|X^*})$ satisfies
(\textsl{A}\ref{cond:omegac_proper}), and} $\displaystyle
\inf_{\eps>0}\left[D(P_{\Phi^*|\Theta^*}\|\ssc{-}P^\eps_{\Phi|\Theta}\,|\,P_\Theta)+D(P_{\Phi^*|\Theta^*}\|\ssc{+}P^\eps_{\Phi|\Theta}\,|\,P_\Theta)\right]
<\infty$.

\item
$D(P_{Y^*|X^*}\|\,P_{Y|X}\mid P_{X^*})<\infty$.\label{cond:M_div}

\item $F_X^{-1}(F_{X|Y}(X^*|Y^*))\sim P_{X^*}$\label{cond:M_inv}

\item \textit{Let $\{Y^*_n\}_{n=1}^\infty$ be the channel output sequence
when the posterior matching scheme for $(P_{X}, P_{Y|X})$ is used
over $P_{Y^*|X^*}$ and initialized with $X_1\sim P_{X^*}$. There
is a contraction $\xi$ and a length function $\psi_\lambda$ over
$\mf{F}_{c}$, such that for every $h\in\mf{F}_{c}$ and
$n\in\NaturalF$,}\label{cond:M_contr}
\begin{equation*}
\sup_{y^{n-1}}\Expt\Big{(}\psi_{\sst{\lambda}}\big{[}F_{X|Y}(\cdot\,|Y^*_n)\circ
F_X^{-1}\circ h\big{]}\,\Big{|}\,Y^{*\,n-1}=y^{n-1}\Big{)}
\leq \xi\big{(}\psi_{\sst{\lambda}}(h\,)\,\big{)}
\end{equation*}

\item \textit{Let $Z=F_X^{-1}(F_{X|Y}(X^*|Y^*))$. For any
$x^*\in\isupp(X^*)$ the set $\isupp(Z|X^*=x^*)$ contains some open
neighborhood of $x^*$.} \label{cond:M_neighb}

\end{enumerate}

The properties above are not too difficult to verify, with the
notable exception of the contraction condition
(\textsl{C}\ref{cond:M_contr}) which is not ``single letter''. This stems
from the fact that the output distribution under mismatch is
generally not i.i.d. Clearly, for any $( P_X,
P_{Y|X})\in\Omega_C$ we have $( P_X,
P_{Y|X})\in\Omega_C^{\rm mis}( P_X, P_{Y|X})$ in
particular. Moreover, if the posterior matching kernels for the
pairs $( P_X, P_{Y|X})$ and $( P_{X^*}, P_{Y^*|X^*})$ happen to
coincide, then we trivially have $( P_{X^*},
P_{Y^*|X^*})\in\Omega_C^{\rm mis}( P_X, P_{Y|X})$ and any
rate $R<I(X^*;Y^*)=I(X;Y)$ is pointwise achievable, hence there is
no rate loss due to mismatch (although satisfaction of input
constraints may be affected, see below). Note that the
initialization step (i.e., transforming the message point into the
first channel input) is in general different even when the kernels
coincide. Nevertheless, identical kernels imply a common input
support and so using a different initialization amounts to a
one-to-one transformation of the message point, which poses no
problem due to pointwise achievability.

The channel model mismatch does incur a rate loss in general, as
quantified in the following Theorem.

\begin{theorem}[\it Mismatch Achievability]\label{thrm:mismatch}
Let $( P_X, P_{Y|X})\in\Omega_C$, and suppose the
corresponding posterior matching scheme (\ref{eq:gn_recX}) is used
over a channel $ P_{Y^*|X^*}$ (unknown on both terminals). If
there exists an input distribution $ P_{X^*}$ such that $(
P_{X^*}, P_{Y^*|X^*})\in\Omega^{\rm mis}_{\sst{C}}( P_X,
P_{Y|X})$, then $P_{X^*}$ is unique and the mismatched scheme with a fixed/variable rate
optimal decoding rule matched to $( P_X, P_{Y|X})$, pointwise
achieves any rate
\begin{equation}\label{eq:mis_rate}
R < I(X^*;Y^*)-\Big{(}D(P_{Y^*|X^*}\|\,P_{Y|X}\big{|}\,P_{X^*})-
D(P_{Y^*}\|\,P_Y)\Big{)}
\end{equation}
within an input constraint $(\eta,\Expt\eta(X^*))$ provided that
$\Expt|\eta(X^*)|<\infty$.
\end{theorem}
\begin{proof}
See Appendix \ref{app:pointwise}.
\end{proof}
The difference between relative entropies in (\ref{eq:mis_rate})
constitutes the penalty in rate due to the mismatch, relative to what could have been achieved for the induced input distribution $P_{X^*}$. Note that this term is always nonnegative due to the convexity of the
relative entropy, and vanishes when there is no mismatch.

For the next example we need the following Lemma. The proof (by
direct calculation) is left out.
\begin{lemma}\label{lem:divergence_gaussian}
Let $U,V$ be a pair of continuous, zero mean, finite variance r.v.'s, and suppose $V$ is Gaussian. Then
\begin{equation*}
D(P_U\|P_V) = h(V)-h(U) + \frac{\log{e}}{2}\left(\frac{\Expt
\,U^2}{\Expt \,V^2}-1\right)
\end{equation*}
\end{lemma}

\begin{example}[\it Robustness of the Schalkwijk-Kailath scheme]\label{ex:AWGN_mis}
Suppose that the Schalkwijk-Kailath scheme (\ref{eq:AWGN_scheme})
designed for an AWGN channel $ P_{Y|X}$ with noise
$Z\sim\m{N}(0,\rm N)$ and input $X\sim\m{N}(0,\rm P)$, is used
over an AWGN channel with noise variance ${\rm N^*}$. Since the
scheme depends on the channel and input only through the
$\SNR=\frac{\rm P}{\rm N}$, then the scheme's kernel coincides
with the Schalkwijk-Kailath kernel for an input
$X^*\sim\m{N}(0,\rm \frac{ N^*}{N}P)$ over the mismatch channel.
Therefore, following the remark preceding Theorem
\ref{thrm:mismatch}, there is no rate loss, and the input power is
automatically scaled to maintain the same SNR for which the scheme
was designed. This robustness of the Schalkwijk-Kailath scheme to
changes in the Gaussian noise (SNR mismatch) was already mentioned~\cite{weissman_kailath_colq}.

However, Theorem \ref{thrm:mismatch} can be used to demonstrate
how the Schalkwijk-Kailath scheme is robust to more general
perturbations in the noise statistics. Suppose the scheme is used
over a \textit{generally non-Gaussian} additive noise channel $
P_{Y^*|X^*}$ with noise $Z^*$ having zero mean and a variance
${\rm N^*}$. Suppose there exists an input distribution $ P_{X^*}$
such that $( P_{X^*}, P_{Y^*|X^*})\in\Omega_C^{\rm mis}(
P_{X}, P_{Y|X})$. We have $Y=X+Z$ and $Y^*=X^*+Z^*$ for the
original channel and the mismatch channel respectively. Plugging
(\ref{eq:AWGN_scheme}) into the invariance property
(\textsl{C}\ref{cond:M_inv}) and looking at the variance, we have
\begin{equation*}
{\rm P^*} = \Expt\left(\frac{X^*}{\sqrt{1+\SNR}}+\frac{\SNR\cdot
Z^*}{\sqrt{1+\SNR}}\right)^2 = \frac{{\rm P^*}+\SNR^2\cdot {\rm
N^*}}{1+\SNR}
\end{equation*}
which immediately results in ${\rm \SNR^*}\dfn\frac{{\rm
P^*}}{{\rm N^*}} = \SNR$, so the SNR is conserved despite the
mismatch. Now applying Theorem \ref{thrm:mismatch} and some simple
manipulations, we find that the mismatched scheme pointwise
achieves any rate $R$ satisfying
{\allowdisplaybreaks
\begin{align*}
R &< h(Y^*)-h(Z^*)-\left(D(P_{Z^*}\|P_Z)-
D(P_{Y^*}\|P_Y)\right)
\\
& = h(Y^*)-h(Z^*)-\left(\vphantom{\frac{\Expt
(Y^*)^2}{\Expt Y^2}}h(Z)-h(Z^*)-h(Y)+h(Y^*)+\frac{\log
e}{2}\left(\frac{\Expt (Z^*)^2}{\Expt Z^2}-\frac{\Expt
(Y^*)^2}{\Expt Y^2}\right)\right)
\\
&= h(Y)-h(Z) +\frac{\log e}{2}\left(\frac{{\rm
P^*}+{\rm N^*}}{P+N}-\frac{{\rm N^*}}{N}\right)
= I(X;Y) + \frac{\log
e}{2}\cdot\frac{{\rm N^*}}{N}\left(\frac{1+{\rm
\SNR^*}}{1+\SNR}-1\right)
\\
& = I(X;Y) = \frac{1}{2}\log(1+\SNR)
\end{align*}}
where we have used Lemma \ref{lem:divergence_gaussian} in the
first equality. Therefore, the mismatched scheme can
attain any rate below the Gaussian capacity it was designed for,
despite the fact that the noise is not Gaussian, and the input
power is automatically scaled to maintain the same SNR for which
the scheme was designed. Invoking~\cite{ZamirErez2004}, we can now claim that the Schalkwijk-Kailath
scheme is \textit{universal} for communication over a memoryless
additive noise channel (within the mismatch set) with a given
variance and an input power constraint, in the sense of loosing at
most half a bit in rate w.r.t. the channel capacity.
\end{example}

\section{Discussion}\label{sec:discussion}
An explicit feedback transmission scheme tailored to any memoryless channel and any input distribution was developed, based on a novel principle of posterior matching. In particular, this scheme was shown to provide a unified view of the well known Horstein and Schalkwijk-Kailath schemes. The core of the transmission strategy lies in the constantly refined representation of the message point's position \textit{relative} to the uncertainty at the receiver. This is accomplished by evaluating the receiver's posterior c.d.f. at the message point, followed by a technical step of matching this quantity to the channel via an appropriate transformation. A recursive representation of the scheme renders it very simple to implement, as the next channel input is a fixed function of the previous input/output pair only. This function is explicitly given in terms of the channel and the selected input distribution. The posterior matching scheme was shown to achieve the mutual information for pairs of channels and input distributions under very general conditions. This was obtained by proving a concentration result of the posterior p.d.f. around the message point, in conjunction with a contraction result for the posterior c.d.f. over a suitable function space. In particular, achievability was established for discrete memoryless channels, thereby also proving that the Horstein scheme is capacity achieving.

The error probability performance of the scheme was analyzed, by casting the variable-rate decoding process as the evolution of a reversed iterated function system (RIFS), and interpreting the associated contraction factors as information rates. This approach yielded two closed form expressions for the exponential decay of the target error probability which facilitates the achievability of a given rate, then used to provide explicit results in several examples. However, the presented error analysis is preliminary and should be further pursued. First, the obtained expressions require searching for good weight or shaping functions, which in many cases may be a difficult task. In the same vein, it is yet unclear under what conditions the error analysis becomes valid for rates up to the mutual information. Finally, the basic technique is quite general and allows for other RIFS contraction lemmas to be plugged in, possibly to yield improved error expressions.

We have seen that a fixed-point free kernel is a necessary condition for achieving any positive rate. We have also demonstrated how fixed points can sometimes be eliminated by considering an equivalent channel, or a corresponding $\mu$-variant scheme. But can this binary observation be refined? From the error probability analysis of Section~\ref{sec:error_memoryless}, it roughly seems that the ``closer'' the kernel is to having a fixed point, the worst the error performance should be. It would be interesting to quantify this observation, and to characterize the best $\mu$-variant scheme for a given input/channel pair, in terms of minimizing the error probability.

We have derived the rate penalty incurred in a channel model mismatch setting, where a posterior matching scheme devised according to one channel model (and input distribution) is used over a different channel. However, the presence of feedback allows for an adaptive transmission scheme to be used in order to possibly reduce or even eliminate this penalty. When the channel is known to belong to some parametric family, there exist universal feedback transmission schemes that can achieve the capacity of the realized channel if the family is not too rich~\cite{Ooi-Wornell}, and sometimes even attain the optimal error exponent~\cite{Tchamkerten-Telatar}. However, these results involve random coding arguments, and so the associated schemes are neither explicit nor simple. It would therefore be interesting to examine whether an \textit{adaptive} posterior matching scheme, in which the transmitter modifies its strategy online based on channel estimation, can be proven universal for some families of memoryless channels. It seems plausible that if the family is not too rich (e.g., in the sense of~\cite{FederLapidoth98}) then the posterior will have a significant peak only when ``close enough'' to the true channel, and be flat otherwise. Another related avenue of future research is the universal communication problem in an individual/adversarial setting with feedback. For discrete alphabets, it was already demonstrated that the \textit{empirical capacity} relative to a modulo-additive memoryless model can be achieved using a randomized sequential transmission strategy that builds on the Horstein scheme~\cite{empirical_capacity_partI}. It remains to be explored whether this result can be extended to
general alphabets by building on the posterior matching scheme, where the empirical capacity is defined relative to some parametric family of channels.

An extension of the suggested scheme to channels with memory is certainly called for. However, the posterior matching principle needs to be modified to take the channel's memory into account, since it is clear that a transmission independent of previous observations is not always the best option in this case. In hindsight, this part of the principle could have been phrased differently: \textit{The transmission functions should be selected so that the input sequence has the correct marginal distribution, and the output sequence has the correct joint distribution}. In the memoryless case, this is just to say that $X_n\sim P_X$, and $Y^n$ is i.i.d. with the marginal $P_Y$ induced by $(P_X,P_{Y|X})$, which coincides with the original principle. However, when the channel has memory the revised principle seems to lead to the correct generalization. For instance, consider a setting where the channel is Markovian of some order $d$, and the ``designed'' input distribution is selected to be Markovian of order $d$ as well\footnote{By that we mean that $Y_n -  X_{n-d}^nY_{n-d}^{n-1}  -  X^{n-d-1}Y^{n-d-1}$ and $X_n -  X_{n-d}^{n-1}Y_{n-d}^{n-1} -  X^{n-d-1}Y^{n-d-1}$ are Markov chains.}. According to the revised principle, the input to the channel should be generated in such a way that any $d$ consecutive input/output pairs have the correct (designed) distribution\footnote{We interpret ``marginal'' here as pertaining to the degrees of freedom suggested by the designed input distribution.}, and the joint output distribution is the one induced by the designed input distribution and the channel, so the receiver cannot ``tell the difference''. To emulate such a behavior, a $d+1$ order (or higher) kernel is required, since any lower order will result in some deterministic dependence between any $d$ consecutive pairs. This also implies that a $d+1$ dimensional message point is generally required in order to provide the necessary degrees of freedom in terms of randomness. It can be verified that whenever such a procedure is feasible, then under some mild regularity conditions the posterior p.d.f. at the message point is $\approx 2^{I(X^n\rightarrow Y^n)}$, where $I(X^n\rightarrow Y^n)$ is the \textit{directed information} pertaining to the designed input distribution and the channel~\cite{massey90}. This is encouraging, since for channels with feedback the directed information usually plays the same role as mutual information does for channels without feedback~\cite{massey90,tatikonda_phd,kim07,TatikondaMitter09,permuter_fsc_feedback}. Note also that the randomness degrees of freedom argument for a multi-dimensional message point, provides a complementary viewpoint on the more analytic argument as to why the additional dimensions are required in order to attain the capacity of an auto-regressive Gaussian channel via a generalized Schalkwijk-Kailath scheme~\cite{kim_gaussian_feedback_cap}. It is expected that a scheme satisfying the revised principle and its analysis should follow through via a similar approach to that appearing in this paper.

\section*{Acknowledgements}
The authors would like to thank Sergio Verd\'u for suggesting Example \ref{ex:exp_noise_mean}, and Young-Han Kim for some useful comments.

\appendix

\section{Main Proofs}\label{app:lemmas}

\begin{proof}[Proof of Lemma \ref{lem:matching_trans}]
For the first claim, let us find the c.d.f. of $F_X^{-1}(\Theta)$:
\begin{equation*}
\Prob(F_X^{-1}(\Theta)\leq x) =
\Prob(\inf\{z\,:\,F_X(z)>\Theta\}\,\leq x)
\stackrel{(\rm a)}{=}
\Prob(F_X(x)\geq \Theta)=F_X(x)
\end{equation*}
where (a) holds since a c.d.f. is nondecreasing and continuous
from the right, and so the result follows. For the second claim,
define $\Phi=F_X(X)-\Theta\cdot P_X(X)$ and let $\phi\in\ui$ be
such that there exists $x_0\in\supp(X)$ for which $F_X(x_0)=\phi$.
Then
\begin{equation*}
F_\Phi(\phi) \geq \Prob\big{(}F_X(X)\leq F_X(x_0)\big{)} =
\Prob(X\leq x_0)
= F_X(x_0) = \phi
\end{equation*}
and on the other hand
\begin{equation*}
F_\Phi(\phi) \leq \Prob\big{(}F_X(X)-P_X(X)\leq F_X(x_0)\big{)} =
\Prob(X\leq x_0)
= F_X(x_0) = \phi
\end{equation*}
hence $F_\Phi(\phi)=\phi$. If such an $x_0$ does not exists then
there must exist a jump point $x_1$ such that
\begin{equation*}
F_X(x_1)-P_X(x_1)\leq \phi < F_X(x_1) \dfn \phi_1
\end{equation*}
and so
\begin{equation*}
F_\Phi(\phi) =  F_\Phi(\phi_1)-P_\Phi\big{(}(\phi,\phi_1]\big{)}
=\phi_1-\Prob\big{(}X=x_1\,,\Theta\cdot P_X(x_1)\leq
\phi_1-\phi\big{)}
=
\phi_1-P_X(x_1)\cdot\frac{\phi_1-\phi}{P_X(x_1)} = \phi
\end{equation*}
For a proper $X$ there are no mass points hence the simpler result
follows immediately.
\end{proof}

\begin{proof}[Proof of Lemma \ref{lem:setting_eqv}]
Assume we are given a transmission scheme $g_n$ and a decoding rule $\Delta_n$ which are known to achieve a rate $R_0$. For simplicity, we assume the decoding rule is fixed rate, (i.e. $|\Delta(y^n)| = 2^{-nR_0}$ for all $y^n$), since any variable rate rule can be easily mapped into a fixed rate rule that achieves the same rate. It is easy to see that in order to prove the above translates into achievability of some rate $R<R_0$ in the standard framework, it is enough to show we can find a sequence $\Gamma_n=\{\theta_{i,n}\in\ui\}_{i=1}^{\lfloor 2^{nR}\rfloor}$ of message point sets, such that $\theta_{i+1,n}-\theta_{i,n}\geq 2^{-nR_0}$ for any $1\leq i< \lfloor 2^{nR}\rfloor$, and such that we have uniform achievability over $\Gamma_n$, i.e.,
\begin{equation*}
\lim_{n\rightarrow\infty}\max_{\theta\in\Gamma_n}\Prob(\theta\not\in\Delta_n(Y^n)|\Theta_0=\theta) = 0
\end{equation*}

We now show how $\Gamma_n$ can be constructed for any $R<R_0$. Let $p_e(n)$ be the (average) error probability associated with our scheme and the fixed rate $R_0$ decoding rule. Define
\begin{equation*}
A_n = \left\{\theta\in\ui\,:\, \Prob(\Theta_0\not\in\Delta(Y^n)|\Theta_0=\theta) > \sqrt{p_e(n)}\right\}
\end{equation*}
and write
\begin{equation*}
p_e(n) = \int\Prob(\Theta_0\not\in\Delta(Y^n)|\Theta_0=\theta)d\theta
> \sqrt{p_e(n)}\int\ind_{A_n}(\theta)d\theta
\end{equation*}
and so we have that $\int\ind_{A_n}(\theta)d\theta<\sqrt{p_e(n)}$. It is now easy to see that if we want to select $\Gamma_n$ such that $\Gamma_n\cap A_n = \phi$, and also $\theta_{{i+1,n}}-\theta_{i,n} \geq 2^{-nR_0}$, then a sufficient condition is that $\frac{1}{|\Gamma_n|}(1-\sqrt{p_e(n)}-\tau_n)\geq 2^{-nR_0}$ for some positive $\tau_n \rightarrow 0$. This condition can be written as
\begin{equation*}
\frac{1}{n}\log|\Gamma_n| \leq R_0 + \frac{1}{n}\log (1-\sqrt{p_e(n)}-\tau_n) = R_0 + {\rm o}(1)
\end{equation*}
At the same time, we also have by definition
\begin{equation*}
\lim_{n\rightarrow\infty}\max_{\theta\in\Gamma_n}\Prob(\theta\not\in\Delta(Y^n)|\Theta_0=\theta) \leq \lim_{n\rightarrow\infty} \sqrt{p_e(n)} = 0
\end{equation*}
and the proof is concluded.
\end{proof}

\begin{proof}[Proof of Lemma \ref{lem:contraction_profile}]
Since $\xi$ is $\cap$-convex over $[0,1]$, it has a unique maximal value attained at some (not necessarily unique) point $x^*$. Moreover, convexity implies $\xi$ is continuous over $(0,1)$, and since it is nonnegative and upper bounded by $\xi(x)<x$, it is also continuous at $x=0$ and $\xi(0)=0$. Now, define the sequence $s_n=\xi^{(n)}(x^*)$. Since $\xi(x)<x$ the sequence $s_n$ is monotonically decreasing, and since $\xi$ is nonnegative it is also bounded from below. Therefore, $s_n$ converges to a limit $\,s_\infty\in[0,1)$, and we can write
\begin{equation*}
\lim_{n\rightarrow\infty}s_n=s_\infty \,,\qquad
\lim_{n\rightarrow\infty}\xi(s_n) =
\lim_{n\rightarrow\infty}s_{n+1} = s_\infty
\end{equation*}
Since $\xi$ is continuous over $[0,1)$ the above implies that $\xi(s_\infty)=s_\infty$, i.e., $s_\infty$ is a fixed point of $\xi$. Thus, we either have $\xi\equiv 0$ in which case $s_\infty=0$, or $\xi\not\equiv 0$ in which case the only fixed point for $\xi$ is zero and so again $s_\infty=0$. We now note that $\xi(x)\leq \xi(x^*)\leq x^*$ for any $x\in[0,1]$, and also that $\xi$ is nondecreasing over $[0,x^*]$ and hence so is $\xi^{(n)}$. We therefore have
\begin{equation*}
\lim_{n\rightarrow\infty}r(n) =
\lim_{n\rightarrow\infty}\sup_{x\in[0,1]}\xi^{(n)}(x) \leq
\lim_{n\rightarrow\infty}\xi^{(n-1)}(x^*)
=\lim_{n\rightarrow\infty}s_n = 0
\end{equation*}
\end{proof}

\begin{proof}[Proof of Lemma \ref{lem:IFS_convergence}]
For any $\eps>0$,
{\allowdisplaybreaks
\begin{align*}
\Prob\left(\psi(S_n(s))>\eps\right) &\stackrel{\rm (a)}{\leq}
\eps^{-1}\Expt[\psi(S_n(s))]
=
\eps^{-1}\Expt\left(\Expt[\psi(S_n(s))\,|\,Y^{n-1}]\right)
=
\eps^{-1}\Expt\left(\Expt[\psi(\omega_{\sst Y_n}\circ
S_{n-1}(s))\,|\,Y^{n-1}]\right)
\\
&\stackrel{\rm (b)}{\leq}
\eps^{-1}\Expt\,\xi\left(\psi(S_{n-1}(s))\right)
\stackrel{\rm
(c)}{\leq} \eps^{-1}\xi\left(\Expt\psi(S_{n-1}(s))\right)
\leq \cdots\stackrel{\rm (d)}{\leq}
\eps^{-1}\xi^{(n)}\left(\psi(s)\right)
\stackrel{\rm
(e)}{\leq} \eps^{-1}r(n)
\end{align*}}
Markov's inequality was used in (a), the contraction relation
(\ref{eq:contraction_IFS_cond}) in (b) and Jensen's inequality in
(c). Inequality (d) is a recursive application of the preceding
transitions, and the definition of the decay profile was
used in (e).
\end{proof}

\begin{proof}[Proof of Lemma \ref{lem:contraction1}]
For any $\eps>0$,
{\allowdisplaybreaks
\begin{align*}
\Prob(|\wt{S}_n(s)-\wt{S}_n(t)|>\eps) &=
\Prob(|\wt{S}_n(s)-\wt{S}_n(t)|^q>\eps^q)
\stackrel{\rm
(a)}{\leq} \eps^{-q}\Expt|\wt{S}_n(s)-\wt{S}_n(t)|^q
\\
&=
\eps^{-q}\Expt(\Expt(\,|\wt{S}_n(s)-\wt{S}_n(t)|^q\,\big{|}\,Y_2^n))
\\
&= \eps^{-q}\Expt(\Expt(|\omega_{\sst Y_1}\circ\cdots\circ
\omega_{\sst Y_n}(s)-\omega_{\sst Y_1}\circ\cdots\circ
\omega_{\sst Y_n}(t)|^q\big{|}Y_2^n))
\\
&\stackrel{\rm
(b)}{\leq}  \eps^{-q}r\cdot \Expt(|\omega_{\sst
Y_2}\circ\cdots\circ \omega_{\sst Y_n}(s)-\omega_{\sst
Y_2}\circ\cdots\circ \omega_{\sst Y_n}(t)|^q)
\leq \cdots\stackrel{\rm (c)}{\leq} \eps^{-q}r^n|s-t|^q
\end{align*}}
Where in (a) we use Markov's inequality, in (b) we use the
contraction (\ref{eq:contraction_cond}), and (c) is a recursive
application of the preceding transitions.
\end{proof}

\begin{proof}[Proof of Theorem \ref{thrm:post_match_scheme}]
We prove by induction that for any $n\in\NaturalF$, $P_{\Theta_0|Y^n}(\cdot|y^n)$ is proper for $P_{Y^n}$-a.a. $y^n\in\m{Y}^n$, and the rest of the proof remains the same. First, this property is satisfied for $n=0$ since $P_{\Theta_0}$ is proper. Now assume the property holds for any $1\leq n \leq k-1$. By our previous derivations, this implies that $X_n\sim P_X$ for any $1\leq n \leq k$, and thus by the definition of an input/channel pair we have in particular $I(X_n;Y_n)=I(X;Y)<\infty$ for any such $n$. Now suppose the property does not hold for $n=k$. This implies there exists a measurable set $A\subseteq\m{Y}^k$ with $P_{Y^k}(A)>0$ so that $P_{\Theta_0|Y^k}(\cdot|y^k)\not\ll P_{\Theta_0}$ for any $y^k\in A$. Therefore, it must be that $I(\Theta_0;Y^k)=\infty$. However standard manipulations using the fact that the channel is memoryless result in $I(\Theta_0;Y^k)\leq \sum_{n=1}^kI(X_n;Y_n)<\infty$, in contradiction.
\end{proof}

\begin{proof}[Proof of Lemma \ref{lem:norm_chan_prop}, claim (\ref{claim:norm3})]
Since $\Theta\sim\m{U}$, it is enough to show that
$P_{\Phi|\Theta}(\cdot|\theta)$ is proper for $\m{U}$-a.a.
$\theta\in\ui$. Define the discrete part of the output support to
be $\m{Y}_D = \{y\in\supp(Y)\,:\, P_Y(y)>0\}$, which is a
countable set. Define also the set $\mf{Y}_D\dfn
\isupp(\Phi|Y\in\m{Y}_D)$ which is a countable union of disjoint
intervals inside the unit interval, corresponding to the ``jump
spans'' introduced by $F_Y$ over $\m{Y}_D$. Furthermore, for any
$x\in\supp(X)$ define $\m{Y}_{D,x}$ to be the set of mass points
for $P_{Y|X}(\cdot|x)$. Since $I(X;Y)<\infty$, then it must be
that $P_{Y|X}(\cdot|x)\ll P_Y$ for $P_X$-a.a. $x\in\supp(X)$.
Therefore, there exists a set $A\subseteq\supp(X)$ of full measure
$P_X(A)=1$, so that $\m{Y}_{D,x} \subseteq \m{Y}_D$ for any $x\in
A$. Therefore, for any $x\in A$, $P_{Y|X}(\cdot|x)$ restricted to
$\supp(Y)\setdiff \m{Y}_D$ has a proper p.d.f., which implies that
$P_{\Phi|X}(\cdot|x)$ restricted to $\ui\setdiff \mf{Y}_D$
has a proper p.d.f. as well, since $\Phi$ is obtained from $Y$ by
applying a continuous and bounded function. $P_{\Phi|X}(\cdot|x)$
restricted to any one of the countable number of intervals
composing $\mf{Y}_D$ is uniform, hence admits a proper p.d.f.
as well. We therefore conclude that $P_{\Phi|X}(\cdot|x)$ is
proper for any $x\in A$. To conclude, define the set
$B=\{\theta\in\ui\,:\, F_X^{-1}(\theta)\in A\}$, which by Lemma
\ref{lem:matching_trans} is of full measure $\m{U}(B)=1$, and from
the discussion above $P_{\Phi|\Theta}(\cdot|\theta)$ is proper for
any $\theta\in B$.
\end{proof}

\begin{proof}[Proof of Lemma \ref{lem:dmc_prop}, claim (\ref{claim:dmc2})] Suppose there exists some $y_0\in\m{Y}$ so that $P_Y(y_0)>0$ and $P_X\dom P_{X|Y}(\cdot|y_0)$. Define the set $A_0=\{\phi\in\ui\,:\,F_Y^{-1}(\phi)=y_0\}$. For any $x\in\m{X}$ and $\phi\in A_0$, the normalized posterior matching kernel evaluated at $\theta=F_X(x)$ satisfies
\begin{equation*}
F_{\Theta|\Phi}(F_X(x)|\phi) = F_{X|Y}(x|y_0) \geq F_X(x)
\end{equation*}
where the last inequality is due to the dominance assumption above, and is \textit{strict} for $x\in\{0,\ldots,|\m{X}|-2\}$. Moreover, the normalized posterior matching kernel evaluated in between this finite set of points is simply a linear interpolation. Thus, for any $\theta\in\ui$ and any $\phi\in A_0$ we have $F_{\Theta|\Phi}(\theta|\phi)>\theta$, and so
\begin{equation*}
\Prob\big{(}F_{\Theta|\Phi}(\theta|\Phi)=\theta\big{)} \leq
1-P_\Phi(A_0) = 1-P_Y(y_0) <1
\end{equation*}
which implies the fixed-point free property (\textsl{A}\ref{cond:omega_fixedp}). The case where $P_{X|Y}(\cdot|y_0)\dom P_X$ follows by symmetry. The case where $P_{X|Y}(\cdot|y_0)\dom P_X(\cdot|y_1)$ is trivial.
\end{proof}

\begin{proof}[Proof of Lemma \ref{lem:dmc_prop}, claim (\ref{claim:dmc3})] We find it simpler here to consider the normalized input $\Theta$ but the original output $Y$, namely to prove an equivalent claim stating that the invariant distribution $P_{\Theta Y}$ for the Markov chain $(\Theta_n,Y_n)$, is ergodic. To that end, we show that if $S\subseteq\ui\times \m{Y}$ is an invariant set, then $P_{\Theta Y}(S)\in\{0,1\}$. Let us write $S$ as a disjoint union:
\begin{equation*}
S = \bigcup_{y\in\m{Y}} A_y\times \{y\}\,,\quad A_y\subseteq\ui
\end{equation*}

The posterior matching kernel deterministically maps a pair $(\theta,y)$ to the input $\hat{\theta} = F_{\Theta|Y}(\theta|y)$, and then the corresponding output is determined via $P_{Y|\Theta}(\cdot|\hat{\theta})$. Since by (\textsl{B}\ref{cond:omega_dm_nzI}) all transition probabilities are nonzero, then each possible output in $\m{Y}$ is seen with a nonzero probability given any input. Thus, denoting the stochastic kernel of the Markov chain by $\m{P}$, we have that $\m{P}(\cdot|(\theta,y))$ has support on the discrete set $\{F_{\Theta|Y}(\theta|y)\}\times \m{Y}$ for any $(\theta,y)\in S$. Since $S$ is an invariant set, this implies that
\begin{equation*}
S' \dfn \bigcup_{y\in\m{Y}} F_{\Theta|Y}(A_y|y)\times \m{Y} \subseteq S
\end{equation*}
where by $F_{\Theta|Y}(A_y|y)$ we mean the image set of $A_y$ under $F_{\Theta|Y}(\cdot|y)$. This in turn implies that
\begin{equation}\label{eq:A_inclusion}
\bigcup_{y\in\m{Y}}F_{\Theta|Y}(A_y|y) \subseteq \bigcap_{y\in\m{Y}}A_y \dfn A
\end{equation}
Now, defining
\begin{equation*}
\bar{S}\dfn A\times \m{Y}
\end{equation*}
we have that $S'\subseteq \bar{S}\subseteq S$, and hence $\bar{S}$ is also an invariant set. Going through the same derivations as for $S$, the invariance of $\bar{S}$ implies that
\begin{equation}\label{eq:A_incF}
\bigcup_{y\in\m{Y}} F_{\Theta|Y}(A|y) \subseteq A
\end{equation}
and hence
\begin{equation*}
\m{U}(A) \geq \max_{y\in\m{Y}} \m{U}(F_{\Theta|Y}(A|y)) \geq \sum_{y\in\m{Y}} \m{U}(F_{\Theta|Y}(A|y))P_{Y}(y)
= \sum_{y\in\m{Y}} P_{\Theta|Y}(A|y)P_{Y}(y) = P_\Theta(A) = \m{U}(A)
\end{equation*}
To avoid contradiction, it must be that $\m{U}(F_{\Theta|Y}(A|y)) = \m{U}(A)$ for all $y\in\m{Y}$, and together with (\ref{eq:A_incF}) it immediately follows that for all $y\in\m{Y}$
\begin{equation}\label{eq:Ainses}
F_{\Theta|Y}(A|y) = A\setdiff N_y\,,\quad  \m{U}(N_y)=0
\end{equation}
Namely, for any output value, the set $A$ remains the same after applying the posterior matching kernel, up to a $\m{U}$-null set.

Let us now prove the implication
\begin{equation}\label{eq:sbar_implies_s}
\m{U}(A)\in\{0,1\} \;\Rightarrow\; P_{\Theta Y}(S)\in\{0,1\}
\end{equation}
To that end, we show that $0<P_{\Theta Y}(S)<1$ implies $0<\m{U}(A)<1$. The upper bound follows from $\m{U}(A) = P_\Theta(A)=P_{\Theta Y}(\bar{S})\leq P_{\Theta Y}(S)<1$ . For the lower bound, we note that $P_{\Theta Y}(S)>0$ implies there exists at least one $y_0\in\m{Y}$ such that $\m{U}(A_{y_0}) >0$. Recall that for a DMC, the normalized posterior matching kernel for any fixed output is a quasi-affine function with slopes given by $\frac{P_{X|Y}(x|y)}{P_X(x)} = \frac{P_{Y|X}(y|x)}{P_Y(y)}$. Since by (\textsl{B}\ref{cond:omega_dm_nzI}) all the transition probabilities are nonzero, these slopes are all positive, and denote their minimal value by $\alpha > 0$. Therefore, it must be that $\m{U}(F_{\Theta|Y}(A_{y_0}|y_0)) > \alpha\m{U}(A_{y_0})>0$ , which by (\ref{eq:A_inclusion}) implies $\m{U}(A)>0$.

After having established (\ref{eq:sbar_implies_s}), we proceed to show that $\m{U}(A)\in\{0,1\}$ which will verify Property (\textsl{A}\ref{cond:omega_inv}). It is easily observed that if $A$ is an interval, (\ref{eq:Ainses}) holds if and only if the endpoints of the interval are both either fixed points of the kernel or endpoints of $\ui$. For $A$ a finite disjoint union of intervals, (\ref{eq:Ainses}) holds if and only if all non-shared endpoints are both either fixed points of the kernel or endpoints of $\ui$. Hence for such $A$, since we assumed the kernel does not have any fixed points, (\ref{eq:Ainses}) holds if and only if $\m{U}(A)\in\{0,1\}$.

Let us now extend this argument to any $A\in\Borel$. Under (\textsl{B}\ref{cond:omega_dm_Q}), there exist two output symbols $y_0,y_1\in\m{Y}$ such that
\begin{equation*}
0>\frac{\beta_0}{\beta_1}\not\in\RatF,
\end{equation*}
where
\begin{equation*}
\beta_i = \log\left(\frac{P_{X|Y}(0|y_i)}{P_X(0)}\right), \quad i\in\{0,1\}
\end{equation*}
Define the set
\begin{equation*}
B \dfn \left\{b\in\ui : \exists n_0,n_1\in\NaturalF, b=2^{n_0\beta_0+n_1\beta_1}\right\}
\end{equation*}
\begin{lemma}\label{lem:dense}
$B$ is dense in $\ui$.
\end{lemma}
\begin{proof}
Without loss of generality, we assume $\beta_0<0<\beta_1$. We prove equivalently that the set $\log B$ is dense in $(-\infty,0)$. Let $b\in(-\infty,0)$. Define
\begin{equation*}
b_n \dfn n\beta_0 + \left\lfloor \frac{b-n\beta_0}{\beta_1}\right\rfloor\beta_1, \quad n\in\NaturalF
\end{equation*}
It is easy to see that $\{b_n\}_{n=n'}^\infty\subset \log B$, if $n'$ is taken to be large enough. Let $\{x\}\dfn x - \lfloor x\rfloor$ be the fractional part of $x$. Write:
\begin{equation*}
r_n\dfn \frac{b-b_n}{\beta_1} = \left\{\frac{b}{\beta_1}+n\left(-\frac{\beta_0}{\beta_1}\right)\right\}
\end{equation*}
Since $\frac{\beta_1}{\beta_0}\not\in\RatF$, $r_n$ can be though of as an irrational rotation on the unit circle, hence is dense in $\ui$~\cite{walters}. In particular, this implies that $r_n$ has a subsequence $r_{k_n}\rightarrow 0$, hence $b_{k_n}\rightarrow b$.
\end{proof}

For $\theta\in\ui$, let $A(\theta)\dfn A\cap (0,\theta)$. For brevity, let $p\dfn P_X(0)$. Define $A_{n_0,n_1}$ be the set obtained starting from $A(p)$ and applying $F_{\Theta|Y}(\cdot|y_0)$ $n_0$ times, and then applying $F_{\Theta|Y}(\cdot|y_1)$ $n_1$ times. $F_{\Theta|Y}(\cdot|y_i)$ is linear over $(0,p)$ with a slope $2^{\beta_i}$, hence assuming that $2^{n_0\beta_0+n_1\beta_1}\leq 1$, we have
\begin{equation}\label{eq:dense_approx}
\m{U}(A_{n_0,n_1}) = 2^{n_0\beta_0+n_1\beta_1}\cdot\m{U}(A(p))
\end{equation}
On the other hand, (\ref{eq:Ainses}) together with the aforementioned linearity imply that $A_{n_0,n_1}$ and $A\left(p\cdot 2^{n_0\beta_0+n_1\beta_1}\right)$ are equal up to a $\m{U}$-null set. Combining this with (\ref{eq:dense_approx}) and Lemma \ref{lem:dense}, we find that for any $\theta\in(0,p)$
\begin{equation*}
\m{U}(A(\theta)) = \theta p^{-1}\m{U}(A(p))
\end{equation*}
We note that $\m{U}(A(\theta))$ is the indefinite Lebesgue integral of $\ind_{A(p)}(\theta)$. Invoking the Lebesgue differentiation Theorem~\cite{kolmogorov_fomin}, the  derivative $\frac{d\m{U}(A(\theta))}{d\theta} = p^{-1}\m{U}(A(p))$ must be equal to $\ind_{A(p)}(\theta)$ for a.a. $\theta\in(0,p)$, which implies $p^{-1}\m{U}(A(p))\in\{0,1\}$. Hence $A(p)$ is either of full measure or a null set.

Let us now show that this implies the same for $A=A(1)$. Define the function
\begin{equation*}
\overline{F}(\theta) \dfn \max_{y\in\m{Y}} F_{\Theta|Y}(\theta)
\end{equation*}
Let us establish some properties of $\overline{F}$.
\begin{enumerate}[(a)]
\item \textit{$\overline{F}$ is Lipschitz, monotonically increasing, and maps $(0,1)$ onto $(0,1)$:} Trivial.\label{prop:fbar1}
\item \textit{$\overline{F}^{(n)}(\theta)\rightarrow 1$ monotonically as $n\rightarrow\infty$ for any $\theta\in\ui$:} Observe that
    \begin{equation*}
    \Expt\left(F_{\Theta|Y}(\theta|Y)\right)=\Expt\Prob(\Theta\leq \theta|Y) = \Prob(\Theta\leq \theta) = \theta,
    \end{equation*}
    Hence $\overline{F}(\theta)\geq \theta$  with equality if and only if $\theta$ is a fixed point, which contradicts property (\textsl{A}\ref{cond:omegac_fixedp}). Thus it must hold that $\overline{F}(\theta)>\theta$ for any $\theta\in\ui$, hence $\overline{F}^{(n)}(\theta)$ is increasing with $n$. $\overline{F}\leq 1$ and therefore a limit exists and is at most $1$. $\overline{F}$ is continuous, hence the limit cannot be smaller than $1$ as this will violate $\overline{F}(\theta)>\theta$.\label{prop:fbar2}
\item \textit{$\overline{F}(A) = A$ up to a $\m{U}$-null set:} it is easily observed that
\begin{align*}
\bigcap_{y\in\m{Y}}F_{\Theta|Y}(A|y) \subseteq \overline{F}(A) \subseteq \bigcup_{y\in\m{Y}}F_{\Theta|Y}(A|y)
\end{align*}
The property follows by applying (\ref{eq:Ainses}).\label{prop:fbar3}
\end{enumerate}

Combining (\ref{prop:fbar1}) and (\ref{prop:fbar3}) it follows that for any $n\geq 1$, $\overline{F}^{(n)}(A(p))=A(\overline{F}^{(n)}(p))$ up to a $\m{U}$-null set. Furthermore, since $A(p)$ is either of full measure or null, then property (\ref{prop:fbar1}) implies the same for $\overline{F}^{(n)}(A(p))$, and so either $\m{U}(A(\overline{F}^{(n)}(p)))=0$ for all $n$, or $\m{U}(A(\overline{F}^{(n)}(p))) = \overline{F}^{(n)}(p)$. Using (\ref{prop:fbar2}), we get:
\begin{equation*}
\m{U}(A) = \m{U}\left(\bigcup_{n=1}^\infty A(\overline{F}^{(n)}(p))\right)\in \left\{0,\lim_{n\rightarrow\infty}\overline{F}^{(n)}(p)\right\} = \{0,1\}
\end{equation*}
Hence (\textsl{A}\ref{cond:omega_inv}) holds.

\begin{remark}\label{rem:weaken_irr}
The proof only requires an irrational ratio to be found for $x=0$ (or similarly, for $x=|\m{X}|-1$), hence a weaker version of property (\textsl{B}\ref{cond:omega_dm_Q}) suffices. It is unclear if even this weaker property is required for ergodicity to hold. The proof fails whenever the leftmost interval $(0,P_X(0))$ cannot be densely covered by a repeated application the posterior matching kernel (starting from the right endpoint), without ever leaving the interval. This argument leans only on the linearity of the kernel within that interval, and does not use the entire non-linear structure of the kernel. It therefore seems plausible that condition (\textsl{B}\ref{cond:omega_dm_Q}) could be further weakened, or perhaps even completely removed.
\end{remark}
\end{proof}

\begin{proof}[Proof of Lemma \ref{lem:dmc_prop}, claim (\ref{claim:dmc5})] (\textsl{B}\ref{cond:omega_dm_nzI}) trivially holds for any equivalent input/channel pair. Let us show there exists one satisfying (\textsl{B}\ref{cond:omega_dm_dom}). To that end, the following Lemma is found useful.
\begin{lemma}\label{lem:perm}
Let $p^n,q^n$ be two distinct probability vectors. Then there exists a permutation operator $\sigma:\RealF^n\mapsto\RealF^n$ such that $\sigma(q^n)\dom \sigma(p^n)$.
\end{lemma}
\begin{proof}
Let $\delta^n$ be the element-wise difference of $p^n$ and $q^n$, i.e., $\delta_k=p_k-q_k$. Define $\sigma$ to be a permutation operator such that $\sigma(\delta^n)$ is in descending order. Then since $p^n\neq q^n$ and $\sum_{i=1}^n\delta_i=0$ we have that any partial sum of $\sigma(\delta^n)$ is positive, i.e., $\sum_{i=1}^k\{\sigma(\delta^n)\}_i> 0$ for any $k<n$, which implies the result.
\end{proof}
Now, since $I(X;Y)>0$ there must exist some $y_0\in\m{Y}$ so that $P_{X|Y}(\cdot|y_0)\neq P_X$. Viewing distributions as probability vectors, then by Lemma \ref{lem:perm} above there exists a permutation operator $\sigma$ such that $\sigma(P_X)\dom \sigma(P_{X|Y}(\cdot|y_0))$. Thus, applying $\sigma$ to the input results in an equivalent input/channel pair for which (\textsl{B}\ref{cond:omega_dm_dom}) holds.
\end{proof}

\begin{proof}[Proof of Lemma \ref{lem:dmc_prop}, claim (\ref{claim:dmc6})] Let $(\mathsf{P}(\m{X}),d_{TV}$) be the space of probability distributions over the alphabet $\m{X}$, equipped with the total variation metric. For a fixed channel $P_{Y|X}$, the set $S$ of input distributions not satisfying property (\textsl{B}\ref{cond:omega_dm_Q}) is clearly of countable cardinality. Since any non-singleton open ball centered at any point in $\mathsf{P}(\m{X})$ is of uncountable cardinality, then $\mathsf{P}(\m{X})\setdiff S$  must be dense in $(\mathsf{P}(\m{X}),d_{TV}$), and the claim follows.

\end{proof}

\begin{proof}[Proof of Lemma \ref{lem:contraction_cond}]
Let $\lambda:[0,1]\mapsto[0,1]$ be any surjective, strictly
$\cap$-convex function symmetric about $\frac{1}{2}$. This implies in particular that $\lambda(\cdot)$ is continuous, its restriction to
$[0,\frac{1}{2}]$ is injective, and $\lambda(0)=\lambda(1)=0,\lambda(\frac{1}{2})=1$. Let $\lambda^{-1}:[0,1]\mapsto [0,\frac{1}{2}]$ be the inverse of $\lambda$ restricted to the $[0,\frac{1}{2}]$ branch. Let $\psi_{\sst\lambda}$ be the corresponding length function over $\mf{F}_{c}$, as defined in (\ref{eq:def_lenght_func}). Define the function $\xi^*:[0,1]\mapsto[0,1]$ as follows:
\begin{equation*}
\xi^*(\theta) \dfn \max\left\{\Expt\lambda\left(F_{\Theta|\Phi}(\lambda^{-1}(\theta)|\Phi)\right),
\Expt\lambda\left(F_{\Theta|\Phi}(1-\lambda^{-1}(\theta)|\Phi)\right)\right\}
\end{equation*}
We now establish the following two properties:
\begin{enumerate}[(a)]
\item \label{prop:xi*cont}\textit{$\xi^*(\cdot)$ is continuous over $[\hspace{1pt}0,1]$:} Fix any $\theta'\in[0,1]$, and let $\{\theta_n\}_{n=1}^\infty$ be a sequence in $[0,1]$ such that $\theta_n\rightarrow \theta'$. Define $q(\theta,\phi)\dfn \lambda\left(F_{\Theta|\Phi}(\lambda^{-1}(\theta)|\phi)\right)$, and  $q_n(\phi)\dfn q(\theta_n,\phi)$. By Lemma \ref{lem:norm_chan_prop} claim (\ref{claim:norm4}), $F_{\Theta|\Phi}(\theta|\phi)$ is continuous in $\theta$ for $P_\Phi$-a.a. $\phi\in\ui$. Since $\lambda(\cdot), \lambda^{-1}(\cdot)$ are continuous, we have that $q(\theta,\phi)$ is continuous in $\theta$ for $P_\Phi$-a.a. $\phi\in\ui$, and therefore $q_n(\phi)\rightarrow q(\theta',\phi)$ for $P_\Phi$-a.a. $\phi\in\ui$. Furthermore, $|q_n(\phi)|\leq 1$. Thus, invoking the bounded convergence Theorem~\cite{kolmogorov_fomin} we get
    \begin{equation*}
    \lim_{n\rightarrow\infty}\Expt(q_n(\Phi)) = \Expt(q(\theta',\Phi))
    \end{equation*}
    Reiterating for $q(\theta,\phi)\dfn \lambda\left(F_{\Theta|\Phi}(1-\lambda^{-1}(\theta)|\phi)\right)$, we conclude that $\xi^*(\theta_n)\rightarrow \xi^*(\theta')$.

\item \label{prop:xi*pos}\textit{$0\leq\xi^*(\theta)<\theta$ for $\theta\in(0,1]$:} The lower bound is trivial. For the upper bound, we note again that
    \begin{equation*}
    \Expt\left(F_{\Theta|\Phi}(\theta|\Phi)\right)=\Expt\Prob(\Theta\leq \theta|\Phi) = \Prob(\Theta\leq \theta) = \theta,
    \end{equation*}
    and since by the fixed-point free property (\textsl{A}\ref{cond:omega_fixedp}) we also have $\Prob(F_{\Theta|\Phi}(\theta|\Phi)=\theta) <1$ for any $\theta\in\ui$, then $F_{\Theta|\Phi}(\theta|\Phi)$ is not a.s. constant. Combining that with the fact that $\lambda(\cdot)$ is \textit{strictly} $\cap$-convex, a \textit{strict} Jensen's inequality holds:
    \begin{equation*}
    \Expt\lambda\big{(}F_{\Theta|\Phi}(\lambda^{-1}(\theta)|\Phi)\big{)} <
    \lambda\left(\Expt\left(F_{\Theta|\Phi}(\lambda^{-1}(\theta)|\Phi)\right)\right)
    = \lambda(\lambda^{-1}(\theta)) = \theta
    \end{equation*}
    Similarly, using the symmetry of $\lambda(\cdot)$,
    \begin{equation*}
    \Expt\lambda\big{(}F_{\Theta|\Phi}(1-\lambda^{-1}(\theta)|\Phi)\big{)} <
    \lambda(1-\lambda^{-1}(\theta))
    = \lambda(\lambda^{-1}(\theta))
    = \theta
    \end{equation*}
\end{enumerate}

Now, define $\xi(\cdot)$ to be the upper convex envelope of $\xi^*(\cdot)$. Let us show that $\xi(\cdot)$ is a contraction. $\xi(\cdot)$ is trivially $\cap$-convex and nonnegative, hence it remains to prove that $\xi(\theta)<\theta$ for $\theta\in(0,1]$. Define the function
\begin{equation*}
\delta(\theta)\dfn \inf_{\phi\in[\theta,1]}\left(\phi-\xi^*(\phi)\right)
\end{equation*}
Property (\ref{prop:xi*pos}) implies that $\delta(0) = 0$. Combining properties (\ref{prop:xi*cont}) and (\ref{prop:xi*pos}), we observe that $\phi-\xi^*(\phi)$ is continuous and positive over $[\theta,1]$ for any fixed $\theta\in(0,1]$, hence attains a positive infimum over that interval. We conclude that $\delta(\theta)$ is continuous and monotonically nondecreasing over $[0,1]$, and positive over $(0,1]$. Fixing any $\theta'\in(0,1]$, we use the definition of the upper convex hull and the above properties of $\delta(\cdot)$ to write
\begin{align}\label{eq:xi*xi}
\nonumber\xi(\theta') &= \sup\left\{\alpha\xi^*(\theta_0) + (1-\alpha)\xi^*(\theta_1)\right\}
\leq \sup\left\{\alpha (\theta_0-\delta(\theta_0))+(1-\alpha)(\theta_1-\delta(\theta_1))\right\}
\\
&\leq \theta'-\inf\left\{\alpha\delta(\theta_0)+(1-\alpha)\delta(\theta')\right\}
\end{align}
where the supremums and the infimum are taken over all $\{\theta_0,\theta_1,\alpha\}$ such that $0\leq\theta_0\leq \theta'\leq \theta_1\leq 1$, and such that $\theta'$ is the convex combination $\theta'=\alpha \theta_0+(1-\alpha)\theta_1$. Thus, since $\delta(\theta')>0$, a necessary condition for $\xi(\theta')\geq \theta'$ is for the infimum in (\ref{eq:xi*xi}) to be attained as $\alpha\rightarrow 1$ and $\delta(\theta_0)\rightarrow 0$. By continuity and positivity, the latter implies $\theta_0\rightarrow 0$. However, the convex combination for $\theta'$ can be maintained as $\alpha\rightarrow 1$ and $\theta_0\rightarrow 0$ if and only if $\theta'=0$, in contradiction. Hence $\xi(\theta')<\theta'$.

To conclude the proof, we demonstrate that $\xi(\cdot)$ and
$\psi_{\sst\lambda}$ satisfy (\ref{eq:contraction_prop}):
{\allowdisplaybreaks
\begin{align*}
\Expt\Big{(}\psi_{\sst{\lambda}}\big{[}F_{\Theta|\Phi}(\cdot\,|\Phi)\circ
h\big{]}\Big{)} &=
\int_0^1\Expt\lambda\left(F_{\Theta|\Phi}(h(\theta)|\Phi)\right)d\theta
\stackrel{(\rm a)}{\leq} \int_0^1\xi^*(\lambda(h(\theta)))d\theta
\stackrel{(\rm b)}{\leq} \int_0^1\xi(\lambda(h(\theta)))d\theta
\\
&\stackrel{(\rm c)}{\leq}
\xi\left(\int_0^1\lambda(h(\theta))d\theta\right)
=
\xi\big{(}\psi_{\sst{\lambda}}(h\,)\,\big{)}
\end{align*}}
where (a) holds by the definition of $\xi^*$ and the symmetry of $\lambda(\cdot)$, (b) holds since
$\xi\geq \xi^*$, and (c) holds by Jensen's inequality.
\end{proof}

\begin{lemma}\label{lem:uniform_eps}
Let $(P_X, P_{Y|X})$ satisfy property (\textsl{A}\ref{cond:omega_fixedp}). Then for any $\alpha>0$, $\eps>0$ and $\delta>0$,
\begin{equation*}
\Prob\left(\max_{1\leq m\leq n}\ssc{-}\Theta_{(1+\alpha)n-m}^\eps(\Theta_m) > \delta\right) = \bigo(\sqrt[8]{r(\alpha n)})
\end{equation*}
where $r(n)$ is the decay profile of the contraction $\xi(\cdot)$ from Lemma \ref{lem:contraction_cond}.
\end{lemma}
\begin{proof}
For any $g\in\mf{F}_{c}$ and any $m,n\in\NaturalF$ where $m\leq n$, define
\begin{equation*}
\bar{G}_{m,m}(\cdot) \dfn  g(\cdot) ,\quad \bar{G}_{m,n}^g(\cdot) \dfn F_{\Theta|\Phi}(\cdot|\Phi_n)\circ \bar{G}_{m,n-1}(\cdot)
\end{equation*}
Then for any fixed $m$ and $g$, $\{G_{m,n}^g\}_{n=m}^\infty$ is an IFS over $\mf{F}_{c}$.  Let $g_u(\theta)=\theta$ be the uniform c.d.f., and define the following r.v.'s:
\begin{equation*}
L_{m,n} \dfn \psi_{\sst{\lambda}}\left(\bar{G}_{m,n}^{g_u}\right)\quad L_{m,n}^* \dfn \sup_{g\in\mf{F}_{c}}\psi_{\sst{\lambda}}\left(\bar{G}_{m,n}^{g}\right)
\end{equation*}
where $\psi_{\sst{\lambda}}$ is the associated length function from Lemma \ref{lem:contraction_cond}. Clearly, $L_{m,n}\leq L_{m,n}^*$. Furthermore, $L_{m,n}^* \leq L_{m+1,n}^*$ for any $m\leq n-1$. To see that, we note that $L_{m,n}^*$ is a deterministic function of $\Phi_m^n\dfn (\Phi_m,\ldots,\Phi_n)$, hence there exists a sequence of functions $\{g_k(\theta;\phi_m^n)\}_{k=1}^\infty$ such that $g_k(\cdot;\phi_m^n)\in\mf{F}_{c}$ for any sequence $\phi_m^n\in\ui^{n-m+1}$, and
\begin{equation*}
L_{m,n}^* = \lim_{k\rightarrow\infty} \psi_{\sst{\lambda}}\left(\bar{G}_{m,n}^{g_k(\cdot;\Phi_m^n)}\right)
= \lim_{k\rightarrow\infty} \psi_{\sst{\lambda}}\left(\bar{G}_{m+1,n}^{F_{\Theta|\Phi}(\cdot|\Phi_m)\circ g_k(\cdot;\Phi_m^n)}\right)
\leq \sup_{g\in\mf{F}_{c}}\psi_{\sst{\lambda}}\left(\bar{G}_{m+1,n}^g\right) = L_{m+1,n}^*
\end{equation*}
Therefore, for any $\nu>0$ we have
\begin{equation*}
\Prob\left(\max_{1\leq m\leq n}L_{m,(1+\alpha)n}>\nu\right) \leq \Prob\left(\max_{1\leq m\leq n}L_{m,(1+\alpha)n}^*>\nu\right)
= \Prob\left(L_{n,(1+\alpha)n}^*>\nu\right)
\leq \nu^{-1}r(\alpha n)
\end{equation*}
where we have used Lemmas \ref{lem:IFS_convergence} and \ref{lem:contraction_cond} for the last inequality, noting that the former holds for any IFS initialization. The proof now follows that of Lemma \ref{lem:diverge_eps}, with the proper minor modifications.
\end{proof}

\begin{lemma}\label{lem:diverge_R_B}
Let $(P_X, P_{Y|X})$ satisfy (\textsl{A}\ref{cond:omega_reg}) and (\textsl{A}\ref{cond:omega_cap_ach}). Then (\ref{eq:density_at_theta}) holds, and for any rate $R<I(X;Y)$
\begin{align}\label{eq:limsupI}
\nonumber\lim_{\eps\rightarrow 0}\limsup_{n\rightarrow\infty}&\,\Prob\left(\;\bigcap_{k=1}^n\left\{\Theta_k-\ssc{-}\Theta^{n,R}_k
< \min\Big{(}\eps,\frac{\Theta_k}{2}\,\Big{)}\right\}
 \right) = 0
\\
\lim_{\eps\rightarrow 0}\limsup_{n\rightarrow\infty}&\,\Prob\left(\;\bigcap_{k=1}^n\left\{\ssc{+}\Theta^{n,R}_k-\Theta_k
< \min\Big{(}\eps,\frac{1-\Theta_k}{2}\,\Big{)}\right\}\right) = 0
\end{align}
Furthermore, if $P_X$ is also the unique input distribution for $P_{Y|X}$ such that $I(X;Y)=C(P_{Y|X})$, then
\begin{equation}\label{eq:unique_cap}
\lim_{n\rightarrow\infty}n^{-1}\sum_{k=1}^n\eta(X_k)  = \Expt(\eta(X)) \quad \text{a.s.}
\end{equation}
for any measurable $\eta:\m{X}\mapsto\RealF$ satisfying $\Expt(|\eta(X)|)<\infty$.
\end{lemma}
\begin{proof}
Without the ergodicity property (\textsl{A}\ref{cond:omega_inv}), we cannot directly use the SLLN which was a key tool in deriving (\ref{eq:density_at_theta}) and (\ref{eq:SLLN_substochastic}). Instead, we use the \textit{ergodic decomposition} for Markov chains\footnote{The chain has at least one invariant distribution, and evolves over a locally compact state space $\ui^2$, hence admits an ergodic decomposition.}~\cite[Section 5.3]{hernandez_lasserre} to write the invariant distribution $P_{\Theta\Phi}$ as a mixture of ergodic distributions. We then apply the SLLN to each ergodic component, and use the maximality property (\textsl{A}\ref{cond:omega_cap_ach}) to control the behavior of the chain within each component. For clarity of exposition, we avoid some of the more subtle measure theoretic details for which the reader is referred to~\cite{hernandez_lasserre}.

Let $\m{P}$ denote the Markov stochastic kernel associated with the posterior matching scheme. The ergodic decomposition implies that there exists a r.v. $\Gamma$ taking values in $\ui$, such that $\Gamma=\chi(\Theta)$ for some measurable function $\chi:\ui\mapsto\ui$, and $P_{\Theta\Phi|\Gamma}(\cdot|\gamma)$ is ergodic for $\m{P}$, for $P_\Gamma$-a.a. $\gamma$. Let us first show that $\Phi$ and $\Gamma$ are statistically independent. For any $S\in\Borel$, it is clear that $P_{\Theta|\Gamma}(\cdot|S)$ is an invariant distribution for $\m{P}$, being a mixture of ergodic distributions. Hence the set $\chi^{-1}(S)$ must be invariant for the posterior matching kernel, i.e.,
\begin{equation}\label{eq:decomp_inv}
F_{\Theta|\Phi}(\chi^{-1}(S)|\phi) = \chi^{-1}(S)
\end{equation}
up to a $P_\Theta$-null set, for $P_\Phi$-a.a. $\phi$. Define $Z\dfn F_{\Theta|\Phi}(\Theta|\Phi)$. For any $S,T\in\Borel$:
\begin{align*}
P_{\Gamma\Phi}(S,T) &= P_{\Theta\Phi}(\chi^{-1}(S),T) \stackrel{(\rm
a)}{=} P_{Z\Phi}(\chi^{-1}(S),T)
\stackrel{(\rm
b)}{=} P_{Z}(\chi^{-1}(S))\cdot P_{\Phi}(T)
\\
&\stackrel{(\rm
c)}{=} P_{\Theta}(\chi^{-1}(S))\cdot P_{\Phi}(T)
= P_{\Gamma}(S)P_{\Phi}(T)
\end{align*}
where (a) follows from  (\ref{eq:decomp_inv}) and the fact that $(Z,\Phi)$ a.s. determines $\Theta$, (b) holds since $Z$ is independent of $\Phi$, and (c) holds since $Z\sim P_{\Theta}$ (i.e., uniform).

We can now apply the SLLN (Lemma \ref{lem:SLLN}) to each ergodic component $\chi^{-1}(\gamma)$. For $P_\Gamma$-a.a. $\gamma$ and $P_{\Theta|\Gamma}(\cdot|\gamma)$-a.a. message points $\theta_0\in\chi^{-1}(\gamma)$
{\allowdisplaybreaks
\begin{align}\label{eq:density_at_theta_nonerg}
\nonumber\lim_{n\rightarrow \infty}\frac{1}{n}\log f_{\Theta_0|\Phi^n}(\theta_0|\Phi^n)
&= \Expt\left(\log \frac{f_{\Phi|\Theta}(\Phi|\Theta)}{f_\Phi(\Phi)}\mid \Gamma=\gamma\right) \qquad \text{\rm a.s.}
\\
&= \Expt\left(\log \frac{f_{\Phi|\Theta \Gamma}(\Phi|\Theta,\gamma)}{f_{\Phi|\Gamma}(\Phi|\gamma)}\mid \Gamma=\gamma\right)
= I(\Theta;\Phi|\Gamma=\gamma)
\end{align}}
Now, for any $\gamma$
\begin{equation*}
I(\Theta;\Phi|\Gamma=\gamma) \leq C(P_{\Phi|\Theta}) =  C(P_{Y|X})
\end{equation*}
where the inequality holds by the definition of the unconstrained capacity and since $\Gamma-\Theta-\Phi$ is a Markov chain, and the equality holds since the normalized channel preserves the mutual information (Lemma \ref{lem:norm_chan_prop}). Furthermore, using the independence of $\Phi$ and $\Gamma$ and the Markov relation above again, together with property (\textsl{A}\ref{cond:omega_cap_ach}), leads to
\begin{equation*}
I(\Theta;\Phi|\Gamma) = I(\Theta;\Phi) = C(P_{Y|X})
\end{equation*}
Combining the above we conclude\footnote{Note that (\ref{eq:erg_comp}) does not hold in general if property (\textsl{A}\ref{cond:omega_cap_ach}) is not satisfied, as there may be variations in the limiting values between ergodic components.} that for $P_\Gamma$-a.a. $\gamma$
\begin{equation}\label{eq:erg_comp}
I(\Theta;\Phi|\Gamma=\gamma) = C(P_{Y|X})
\end{equation}
Substituting the above into (\ref{eq:density_at_theta_nonerg}) yields
\begin{equation*}
\lim_{n\rightarrow \infty}\frac{1}{n}\log f_{\Theta_0|\Phi^n}(\theta_0|\Phi^n) = C(P_{Y|X}) \qquad \text{\rm a.s.}
\end{equation*}
for $P_\Theta$-a.a. $\theta_0$. This in turn implies (\ref{eq:density_at_theta}).

Establishing (\ref{eq:limsupI}) follows the same line of argument, proving a weaker version of (\ref{eq:SLLN_substochastic}). By the ergodic decomposition, for $P_\Gamma$-a.a. $\gamma$ and $P_{\Theta|\Gamma}(\cdot|\gamma)$-a.a. message points $\theta_0\in\chi^{-1}(\gamma)$
\begin{align}\label{eq:chi_def}
\nonumber \lim_{n\rightarrow\infty}\frac{1}{n}\sum_{k=1}^n \log \ssc{-}f^\eps_{\Phi|\Theta}(\Phi_k|\Theta_k)
&= \Expt\left(\log \ssc{-}f^\eps_{\Phi|\Theta}(\Phi|\Theta) \mid \Gamma=\gamma\right) \qquad \text{\rm a.s.}
\\
&\dfn \m{L}_\eps(\gamma)
\end{align}
The function $\m{L}_\eps(\gamma)$ satisfies
\begin{equation*}
\Expt\m{L}_\eps(\Gamma) = \Expt\left(\log \ssc{-}f^\eps_{\Phi|\Theta}(\Phi|\Theta)\right) = I_{\eps}^-,
\end{equation*}
and since $\ssc{-}f^\eps_{\Phi|\Theta}\leq f_{\Phi|\Theta}$, then
\begin{equation*}
\m{L}_\eps(\gamma) \leq C(P_{Y|X})
\end{equation*}
for $P_\Gamma$-a.a. $\gamma$. Now since $I(X;Y) = C(P_{Y|X})$ under property (\textsl{A}\ref{cond:omega_cap_ach}), then
\begin{equation}\label{eq:I_eps_C}
\lim_{\eps\rightarrow 0}I_{\eps}^- = C(P_{Y|X})
\end{equation}
It is therefore clear that for small $\eps$ values $\m{L}_\eps(\gamma)$ must be close to $I_{\eps}^-$ for a set of high $P_\Gamma$ probability. Precisely:
\begin{equation*}
A_{\eps,\nu} \dfn \left\{\gamma\in\supp(\Gamma) : \m{L}_\eps(\gamma) > I_{\eps}^- - \nu^{-1}(C(P_{Y|X})-I_{\eps}^-)\right\}
\end{equation*}
Then
\begin{equation*}
I_{\eps}^- = \Expt\m{L}_\eps(\Gamma) \leq P_\Gamma(A_{\eps,\nu})C(P_{Y|X})
+ (1-P_\Gamma(A_{\eps,\nu}))\left(I_{\eps}^- - \nu^{-1}(C(P_{Y|X})-I_{\eps}^-)\right)
\end{equation*}
Rearranging, we get
\begin{equation}\label{eq:A_eps_nu_bound}
P_\Gamma(A_{\eps,\nu}) \geq \frac{1}{1+\nu}
\end{equation}
Combining (\ref{eq:chi_def}), (\ref{eq:I_eps_C}) and (\ref{eq:A_eps_nu_bound}), we conclude that for any $\nu>0$ and any $\eps>0$ small enough,
\begin{equation}
\Prob\left(\lim_{n\rightarrow\infty}\frac{1}{n}\sum_{k=1}^n \log \ssc{-}f^\eps_{\Phi|\Theta}(\Phi_k|\Theta_k) > I_{\eps}^- - \frac{\delta(\eps)}{\nu}\right) \geq \frac{1}{1+\nu}
\end{equation}
for $P_\Theta$-a.a. message points $\theta_0$, where $\delta(\eps)\rightarrow 0$ as $\eps\rightarrow 0$. The remainder of the proof follows that of Lemma \ref{lem:SLLN_memoryless}, with some minor adaptations.

Finally, suppose $P_X$ is the unique capacity achieving input distribution for $P_{Y|X}$. For $P_\Gamma$-a.a. $\gamma$,
\begin{equation}
I(X;Y|\Gamma=\gamma)  = I(\Theta;\Phi|\Gamma=\gamma) = C(P_{Y|X})
\end{equation}
Thus, since $\Gamma-X-Y$ is a Markov chain and from the uniqueness of $P_X$ as capacity achieving, it must be that $P_{X|\Gamma}(\cdot|\gamma) = P_X(\cdot)$ for $P_\Gamma$-a.a. $\gamma$. Applying the SLLN to each ergodic component, we find that for $P_\Gamma$-a.a. $\gamma$ and $P_{\Theta|\Gamma}(\cdot|\gamma)$-a.a. message points $\theta_0\in\chi^{-1}(\gamma)$
\begin{align*}
\lim_{n\rightarrow\infty}n^{-1}\sum_{k=1}^n\eta(X_k)  &= \lim_{n\rightarrow\infty}n^{-1}\sum_{k=1}^n\eta(F_X^{-1}(\Theta_k))
=\Expt(\eta(F_X^{-1}(\Theta))|\Gamma=\gamma) \qquad \text{a.s.}
\\
&= \Expt(\eta(X)|\Gamma=\gamma)  = \Expt\eta(X)
\end{align*}
establishing (\ref{eq:unique_cap}).

\begin{remark}
It is instructive to point out that the proof of the Lemma holds also when property (\textsl{A}\ref{cond:omega_fixedp}) is not satisfied, namely when the posterior matching kernel has fixed points. In that case, each ergodic component must lie strictly inside an invariant interval (i.e., an interval between adjacent fixed points), which results in a decoding ambiguity as the receiver cannot distinguish between the ergodic components. As discussed in Section~\ref{sec:extensions}-\ref{subsec:mu_var}, this exact phenomena prevents any positive rate from being achieved, and generally requires using a posterior matching variant. The fact that capacity is nonetheless achieved under (\textsl{A}\ref{cond:omega_cap_ach}) in the absence of fixed-points even when the chain is not ergodic, suggests that in this case almost any ergodic component, in addition to being capacity achieving in the sense of (\ref{eq:erg_comp}), is also dense in $\ui$. The intuitive interpretation is that in that case any interval intersects with almost all of the ergodic components, hence the receiver, interested in decoding intervals, is ``indifferent'' to the specific component the chain lies in.
\end{remark}
\end{proof}

\section{Pointwise Achievability Proofs}\label{app:pointwise}

\begin{proof}[Proof of Lemma \ref{lem:Omegac_in_Omega}]
Property (\textsl{A}\ref{cond:omegac_proper}) implies in particular that $F_X(x),F_Y(y)$ are continuous and bijective over $\,\isupp(X),\isupp(Y)$ respectively, and that $F_{XY}(x|y)$ is jointly continuous in $x,y$ over $\isupp(X,Y)$. The normalized posterior matching kernel is therefore given by
\begin{equation*}
F_{\Theta|\Phi}(\theta|\phi) =
F_{X|Y}(F_X^{-1}(\theta)|F_Y^{-1}(\phi))
\end{equation*}
and is jointly continuous in $\theta,\phi$ over $\isupp(\Theta,\Phi)$ (note that for the family $\Omega_A$ this does not hold in general). Thus, property (\textsl{A}\ref{cond:omegac_fixedp}${}^*$) implies (by continuity) that for any $\theta\in\ui$ there exists some $\phi\in\ui$ so that $F_{\Theta|\Phi}(\theta|\phi)\neq\theta$.\footnote{Note that continuity also implies there is an interval for which this holds, and since $\Phi\sim\m{U}$, the stronger property (\textsl{A}\ref{cond:omega_fixedp}) holds.}

We first show that the chain is $P_{\Theta\Phi}$-irreducible. Let $\Borel$ here denote the usual Borel $\sigma$-algebra corresponding to the open unit interval. Since $\Theta_n$ is a deterministic function of $(\Theta_{n-1},\Phi_{n-1})$, and since $\Phi_n$ is generated from $\Theta_n$ via a memoryless channel, it follows (by arguments similar to those given in the proof of Lemma \ref{lem:dmc_prop}) that to establish irreducibility it suffices to consider only the $\Theta_n$ component of the chain, and (since $P_{\Theta\Phi}$ has a proper p.d.f.) to show that any set $\Delta\in\Borel$ with $\m{U}(\Delta)>0$ is reached in a finite time with a positive probability starting from any fixed message point $\Theta_0=\theta_0\in\ui$.

Define the set mapping $\pi:\Borel\mapsto\Borel$
\begin{equation*}
\pi(A) \dfn \left\{\xi\in \ui\,:\, \xi =
F_{\Theta|\Phi}(\theta|\phi)\,,\theta\in A,
\phi\in
\isupp\left(\Phi|\Theta=\theta)\right)\right\}
\end{equation*}
namely, the set of all points that are ``reachable'' from the set $A$ in a single iteration. If $A$ is an interval (or a single point), then $\pi(A)$ is also an interval, since it is a continuous image of the set $A'=\isupp(\Theta,\Phi)\cap\left\{A\times \ui\right\}$, which by property (\textsl{A}\ref{cond:omegac_proper}) is a connected set\footnote{This is proved as follows: Since $F_X,F_Y$ are continuous, the set $\isupp(\Theta,\Phi)$ inherits the properties of $\isupp(X,Y)$, namely it is connected (and open, hence path-connected) and convex in the $\phi$-direction. Therefore, any two points in $a,b\in A'$ can be connected by a path in $\isupp(\Theta,\Phi)$. If this path does not lie entirely in $A'$, then consider a new path that starts from $a$ in a straight line connecting to the last point in the original path which has the same $\theta$ coordinate as $a$, then merges with the original path until reaching the first point with the same $\theta$ coordinate as $b$, and continuing in a straight line to $b$. Since $\isupp(\Theta,\Phi)$ is convex in the $\phi$-direction this new path is completely within $A'$.}. For any $\theta_0\in\ui$ it holds that $\Expt\left(F_{\Theta|\Phi}(\theta_0|\Phi)\right)=\theta_0$, and together with property (\textsl{A}\ref{cond:omegac_fixedp}${}^*$) it must also be that
\begin{equation}\label{eq:expnading}
\inf_{\phi\in\isupp(\Phi|\Theta=\theta_0)}\hspace{-0.2cm}F_{\Theta|\Phi}(\theta_0|\phi)
< \theta_0 \;<
\hspace{-0.2cm}\sup_{\phi\in\isupp(\Phi|\Theta=\theta_0)}\hspace{-0.2cm}F_{\Theta|\Phi}(\theta_0|\phi)
\end{equation}
Thus, $\theta_0$ is an interior point of the interval $\pi(\{\theta_0\})$. The arguments above regarding $\pi$ can be applied to all points within the set $\pi(\{\theta_0\})$, and then recursively to obtain
\begin{equation}\label{eq:expanding_intervals}
\theta_0\in \pi(\{\theta_0\})\subseteq \pi^{(2)}(\{\theta_0\}) \subseteq
\cdots\subseteq \pi^{(n)}(\{\theta_0\})\subseteq \cdots
\end{equation}
where $\pi^{(n)}$ is the $n$-fold iteration of $\pi$. Therefore, $\{\pi^{(n)}(\{\theta_0\})\}_{n=1}^\infty$ is a sequence of \textit{expanding intervals} containing $\theta_0$ as an interior point. Note also that $\pi^{(n)}(\{\theta_0\})=\isupp(\Theta_n|\Theta_{0}=\theta_0)$. Consider the set
\begin{equation*}
A_{\theta_0} = \bigcup_{n=0}^\infty \pi^{(n)}(\{\theta_0\})
\end{equation*}
Let us show that $A_{\theta_0}=\ui$. First, it is easy to see that $A_{\theta_0}$ is an open interval, since it is a union of nested intervals, and if it had contained one of its endpoints then that endpoint would have been contained in $\pi^{(n)}(\{\theta_0\})$ for some $n$, which by the expansion property above is an interior point of $\pi^{(n+1)}(\{\theta_0\})\subseteq A_{\theta_0}$, in
contradiction. Now, suppose that $A_{\theta_0}=(\theta_1,\theta_2)$ for $\theta_1>0$. Using
(\ref{eq:expnading}) and the continuity of $F_{\Theta|\Phi}(\theta|\phi)$ once again, we have
\begin{equation*}
\lim_{\theta\rightarrow\theta_1^+}\inf_{\phi\in\isupp(\Phi|\Theta=\theta)}\hspace{-0.35cm}F_{\Theta|\Phi}(\theta|\phi)
=
\inf_{\phi\in\isupp(\Phi|\Theta=\theta_1)}\hspace{-0.3cm}F_{\Theta|\Phi}(\theta_1|\phi)
< \theta_1
\end{equation*}
in contradiction. The same argument applies for $\theta_2$, establishing $A_{\theta_0}=\ui$. As a result, for any set $\Delta\in\Borel$ with $\m{U}(\Delta)>0$ we have that $\m{U}(\Delta\cap \pi^{(n)}(\{\theta_0\}))\rightarrow\m{U}(\Delta)$ as $n\rightarrow\infty$. Therefore, there exists a finite $n$ for which $\m{U}(\Delta\cap \pi^{(n)}(\{\theta_0\}))>0$, and since $\m{U}\ll P_{\Theta_n|\Theta_0}$ when restricted to $\pi^{(n)}(\{\theta_0\})$, it must be that $P_{\Theta_n|\Theta_0}(\Delta|\theta_0)>0$. Thus, the normalized chain is $P_{\Theta\Phi}$-irreducible. It was already verified that $P_{\Theta\Phi}$ is an invariant distribution, hence by Lemma \ref{lem:irreducible_to_recurrent} the chain is also recurrent, $P_{\Theta\Phi}$ is unique and ergodic, and so property (\textsl{A}\ref{cond:omega_inv}) holds.

Let $\m{P}$ denote the stochastic kernel of our Markov chain. To establish p.h.r., we would like to use condition (\ref{cond:PHR_abs_cont}) of Lemma \ref{lem:PHR_cond}. However, $\Theta_{n+1}$ is a deterministic function of $(\Theta_n,\Phi_n)$, and thus $\m{P}(\cdot|(\theta,\phi))\not\ll P_{\Theta\Phi}$ (as the former is supported on a $P_{\Theta\Phi}$-null set). Nevertheless, it is easy to see that due to the expansion property, the $2$-skeleton of the chain (which is also recurrent with the same invariant distribution) admits a proper p.d.f. over a subset of $\isupp(\Theta,\Phi)$ and therefore $\m{P}^2(\cdot|(\theta,\phi))\ll P_{\Theta\Phi}$ for any $(\theta,\phi)\in\isupp(\Theta,\Phi)$. Thus, by condition (\ref{cond:PHR_abs_cont}) of Lemma \ref{lem:PHR_cond} the $2$-skeleton is p.h.r., which in turn implies the chain itself is p.h.r. via condition (\ref{cond:PHR_skeleton}) of Lemma \ref{lem:PHR_cond}.

To establish aperiodicity, we use the expansion property (\ref{eq:expanding_intervals}) once again. Suppose the chain has period $d>1$ and let $\{D_i\}_{i=0}^{d-1}$ be the corresponding partition of the state space $\isupp(\Theta,\Phi)$. From our previous discussion we already know that for any
$(\theta_0,\phi_0)$, the set $\isupp(\Theta_n|\Theta_0=\theta_0,\Phi_0=\phi_0)$ is an interval that expands into $\ui$ as $n\rightarrow\infty$. Since we have the Markov relation $\Phi_n - \Theta_n - \Theta^{n-1}\Phi^{n-1}$, the set $\isupp(\Theta_n,\Phi_n|\Theta_0=\theta_0,\Phi_0=\phi_0)$ expands into $\isupp(\Theta,\Phi)$ in the sense that it contains any open subset of $\isupp(\Theta,\Phi)$ for any $n$ large enough. Therefore, by definition of periodicity for any $n\in\NaturalF$ and $i\in\{0,\ldots,d-1\}$ we have $\Prob((\Theta_{nd+i},\Phi_{nd+i})\in D_i|(\Theta_0,\Phi_0)\in D_0)=1$, and since $P_{\Theta\Phi}\ll \m{U}\times\m{U}$ , then it must be that $(\m{U}\times\m{U})\left(\isupp(\Theta,\Phi)\setdiff D_i\right)=0$ for any $i\in\{0,1,\ldots,d-1\}$. However, this cannot be satisfied by $d>1$ disjoint sets.

\end{proof}

\begin{lemma}\label{lem:pointwise_ext}
Suppose $(P_X,P_{Y|X})\in\Omega_C$. Then Lemmas
\ref{lem:SLLN_memoryless} and \ref{lem:diverge_R} hold for any
fixed message point $\Theta_0=\theta_0\in\ui$. Furthermore, for any $\eps>0,\delta>0$ and $\theta_0\in\ui$:
\begin{equation*}
\lim_{n\rightarrow\infty}\Prob\big{(}\ssc{-}\Theta^{\delta}_n>\eps|\Theta_0=\theta_0\big{)}
\hspace{-1pt}= \lim_{n\rightarrow\infty}\Prob\big{(}\ssc{+}\Theta^{\delta}_n<1-\eps|\Theta_0=\theta_0\big{)}
=0
\end{equation*}
\end{lemma}
\begin{proof}
The proofs of Lemmas \ref{lem:SLLN_memoryless} and \ref{lem:diverge_R} remain virtually the same, only now using the SLLN for
p.h.r. chains (Lemma \ref{lem:SLLN}) to obtain convergence for any fixed message point.

Since by Lemma \ref{lem:Omegac_in_Omega} the normalized chain is p.h.r. and aperiodic, Lemma \ref{lem:PHR_convergence} guarantees that the marginal distribution converges to the invariant distribution $P_{\Theta\Phi}$ in total variation, for any initial condition and hence any fixed message point. Loosely speaking, we prove the result by reducing the fixed message point setting for large
enough $n$, to the already analyzed case of a uniform message point in Lemma \ref{lem:diverge_eps}.

First, let $\{\wt{\Phi}_n\}_{n=1}^\infty$ be a sequence of r.v.'s such that $P_{\wt{\Phi}_n}$ tends to $\m{U}$ in total variation. Then the result of Lemma \ref{lem:contraction_cond} can be rewritten as
\begin{equation}\label{eq:contr_lim1}
\lim_{n\rightarrow\infty}\Expt\Big{(}\psi_{\sst{\lambda}}\big{[}F_{\Theta|\Phi}(\cdot\,|\wt{\Phi}_n)\circ
h\big{]}\Big{)} \,\leq
\;\xi\big{(}\psi_{\sst{\lambda}}(h\,)\,\big{)}
\end{equation}
which holds since the expectation is taken over a bounded function.

Now, consider the $k$-fold chain $\{\Theta_n^{n+k-1},\Phi_n^{n+k-1}\}_{n=1}^{\infty}$ for some fixed $k$. It is immediately seen that this chain is also p.h.r., and its invariant distribution is $P_{\Theta\Phi}^k$, the $k$-fold cartesian product of $P_{\Theta\Phi}$. Thus, by Lemma \ref{lem:PHR_convergence} the $k$-fold chain approaches this invariant distribution in total variation for any initial condition. In particular, this implies that
\begin{equation*}
\lim_{n\rightarrow\infty}d_{TV}(P_{\Phi_n^{n+k-1}|\Theta_0}(\cdot|\theta_0),\m{U}^k) = 0
\end{equation*}
where $\m{U}^k$ is the $k$-fold cartesian product of $\m{U}$. Namely, the distribution of $k$ consecutive outputs tends to i.i.d. uniform in total variation. Using (\ref{eq:contr_lim1}) and a trivial modification of Lemma \ref{lem:IFS_convergence} for an asymptotically i.i.d. control sequence, we have that for any fixed $k$
\begin{equation}\label{eq:decay_PHR}
\lim_{n\rightarrow\infty}\Prob\left(\psi_{\sst{\lambda}}(\bar{G}_n(\theta)) > \nu\,|\Theta_0=\theta_0\right) \leq \frac{1}{\nu}\,r(k)
\end{equation}
where $r(\cdot)$ is the decay profile of $\xi$. Let $n_k$ be the smallest integer such that for any $n\geq n_k$
\begin{equation*}
\Prob\left(\psi_{\sst{\lambda}}(\bar{G}_n(\theta)) > \nu\,|\Theta_0=\theta_0\right) \leq \frac{1}{\nu}\,\sqrt{r(k)}
\end{equation*}
holds, which must exist by (\ref{eq:decay_PHR}). Thus,
\begin{equation*}
\lim_{k\rightarrow \infty}\Prob\left(\psi_{\sst{\lambda}}(\bar{G}_{n_k}(\theta)) > \nu\,|\Theta_0=\theta_0\right) \leq \frac{1}{\nu}\,\lim_{k\rightarrow \infty}\sqrt{r(k)} = 0
\end{equation*}
Now, the proof of the Lemma follows through by working with
$(k,n_k)$ in lieu of $n$, and in (\ref{eq:long1}) using the fact
that the distribution of $\bar{G}_n(\theta_0)$ tends to $\m{U}$ in
total variation.
\end{proof}

\begin{proof}[Proof of Theorem \ref{thrm:mismatch}]
Let us first make the distinction between the Markov chain
generated by the posterior matching scheme for $(P_X, P_{Y|X})$
when operating over the channel $P_{Y|X}$, according to whose law
the transmitter and receiver encode/decode, and the chain
generated by the same scheme when operating over the channel
$P_{Y^*|X^*}$, which describes what actually takes place during
transmission. We refer to the former as the \textit{primary chain}
denoting its input/output sequence as usual by $(X_n,Y_n)$, and to
the latter as the \textit{mismatch chain}, denoting its
input/outptut sequence by $(X^*_n,Y^*_n)$. The same monikers and
notations are used for the normalized counterparts.

Property (\textsl{C}\ref{cond:M_neighb}) guarantees that the expansion
property holds for the mismatch chain, and since by Property
(\textsl{C}\ref{cond:M_inv}) $P_{X^*Y^*}$ is an invariant distribution, a
similar derivation as in Lemma \ref{lem:Omegac_in_Omega} implies
that the mismatch chain is p.h.r., which in particular also
guarantees the uniqueness of $P_{X^*Y^*}$. We would now like to
obtain an analogue of Lemma \ref{lem:SLLN_memoryless}. Let us
expand posterior p.d.f. w.r.t. the  primary chain, using the fact
that it induces an i.i.d. output distribution is  (this does not
necessarily hold for the mismatch chain) and the channel is
memoryless.
\begin{equation*}
f_{\Theta_0|Y^n}(\theta|y^n) =
\frac{f_{\sst{Y_n|\Theta_0,Y^{n-1}}}(y_n\,|\,\theta,y^{n-1})}{f_{\sst{Y_n|Y^{n-1}}}(y_n\,|\,y^{n-1})}\,f_{\Theta_0|Y^{n-1}}(\theta|y^{n-1})
= \frac{f_{Y|X}(y_n\,|\,g_n(\theta,y^{n-1}))}{f_Y(y_n)}\,f_{\Theta_0|Y^{n-1}}(\theta|y^{n-1})
\end{equation*}
Applying the recursion rule $n$ times, taking a logarithm and
evaluating the above at the message point, we obtain
\begin{align*}
\frac{1}{n}\log f_{\Theta_0|Y^n}(\theta|y^n) =
\frac{1}{n}\sum_{k=1}^n\log \frac{f_{Y|X}(y_k\,|\,g_k(\theta,y^{k-1}))}{f_Y(y_k)}
\end{align*}
Now we can evaluate this posterior of the primary chain using the inputs/outputs of the mismatch chain, and apply the p.h.r. SLLN (Lemma \ref{lem:SLLN_memoryless}) for the mismatch chain using its
invariant distribution $P_{\Theta^*\Phi^*}$:
{\allowdisplaybreaks
\begin{align*}
\lim_{n\rightarrow\infty}\frac{1}{n}\log f_{\Theta_0|Y^n}(\Theta_0|Y^{*n})
&=\lim_{n\rightarrow\infty}\frac{1}{n}\sum_{k=1}^n\log\frac{f_{Y|X}(Y^*_k\,|\,g_k(\Theta_0,Y^{* k-1}))}{f_Y(Y^*_k)}
\stackrel{(\rm
a)}{=} \lim_{n\rightarrow\infty}\frac{1}{n}\sum_{k=1}^n\log\frac{f_{Y|X}(Y^*_k\,|\,X^*_k)}{f_Y(Y^*_k)}
\\
&\quad= \Expt\left(\log \frac{f_{Y|X}(Y^*|X^*)}{f_Y(Y^*)}\right)\qquad \m{P}^*_{\theta_0}{\text-a.s.}
\\
&\quad\stackrel{(\rm
b)}{=}\Expt\left(\log\frac{f_{Y|X}(Y^*|X^*)}{f_{Y^*|X^*}(Y^*|X^*)}+\log\frac{f_{Y^*}(Y^*)}{f_Y(Y^*)}
+\log\frac{f_{Y^*|X^*}(Y^*|X^*)}{f_{Y^*}(Y^*)}\right)
\\
&\quad= I(X^*;Y^*) - \left(D(P_{Y^*|X^*}\|P_{Y|X}|P_{X^*}) - D(P_{Y^*}\| P_Y)\right)
\\
&\quad\dfn\; R^{\rm mis}(X,Y;X^*,Y^*)
\end{align*}}
where in (a) we used the definition of the channel input, and in (b) we used Property (\textsl{C}\ref{cond:M_div}) and the convexity of the relative entropy which together guarantee that $D( P_{Y^*}\| P_Y)\leq D(P_{Y|X}^*\| P_{Y|X}\,|\, P_{X^*})<\infty$. The same analysis using normalized chains results in
\begin{equation*}
\lim_{n\rightarrow\infty}\frac{1}{n}\log
f_{\Theta_0|\Phi^n}(\Theta_0|\Phi^{*n}) = \Expt\log
f_{\Phi|\Theta}(\Phi^*|\Theta^*)
= R^{\rm mis} \qquad
P_{\theta_0}{\text-a.s.}
\end{equation*}
where the last equality is due to the invertibility of the chain
normalization, which is guaranteed by property
(\textsl{A}\ref{cond:omegac_proper}). Now we can define the analogue of
$I_{\eps}^-$ in (\ref{eq:eps_inf}) as follows:
\begin{equation*}
R^{\rm mis}_{\,\eps}\dfn \Expt\log
\ssc{-}f^\eps_{\Phi|\Theta}(\Phi^*|\Theta^*)
\end{equation*}
Therefore,
\begin{equation*}
0 \leq R^{\rm mis}-R^{\rm mis}_{\,\eps}
= D\left(P_{\Phi^*|\Theta^*}\|\,\ssc{-}P^\eps_{\Phi|\Theta}\,|\,P_\Theta\right)-D\left(P_{\Phi^*|\Theta^*}\|P_{\Phi|\Theta}\,|\,P_\Theta\right)
\end{equation*}
The second term on the right-hand-side above is finite due to
Property (\textsl{C}\ref{cond:M_div}), and by the Property
(\textsl{C}\ref{cond:M_reg}) we have that
$\inf_{\eps>0}D(P_{\Phi^*|\Theta^*}\|\ssc{-}P^\eps_{\Phi|\Theta}\,|\,P_\Theta)<\infty$.
Thus, for any $\eps$ small enough
\begin{equation*}
-\infty < R^{\rm mis}_{\,\eps} \leq R^{\rm mis}
\end{equation*}
We can now continue as in the proof of Lemma
(\ref{lem:diverge_R}), to show that (\ref{eq:prob1}) holds in this
case for any rate $R<R^{\rm mis}$.

The contraction Property (\textsl{C}\ref{cond:M_contr}) implies the
equivalent of Lemma \ref{lem:diverge_eps} for the mismatch chain,
since although the output sequence $Y^*_n$ is not necessarily
i.i.d. even when we start in the invariant distribution, we have a
contraction uniformly given any conditioning. Tied together with
the above and repeating the last steps of Theorem \ref{thrm:achv},
the achievability of (\ref{eq:mis_rate}) is established .
\end{proof}

\section{Miscellaneous Proofs}\label{app:misc}

\begin{proof}[Proof of Lemma \ref{lem:regular_dist}]
For simplicity we assume that $f_X$ is symmetric around its
maximum, the general unimodal case follow through essentially the
same way. Since the property of having a regular tail is shift
invariant, we can further assume without loss of generality that
$f_X$ attains its maximum at (and is symmetric around) $x=0$.
\begin{enumerate}[(i)]
\item By the assumption, there exist $m_0,m_1>0$, $b\geq a>1$ and
$x_0>1$ so that for any $|x|>x_0$
\begin{equation*}
m_0|x|^{-b} \leq f_X(x) \leq m_1|x|^{-a}
\end{equation*}
Thus, for any $x>x_0$
\begin{equation*}
1-F_X(x) \leq m_1\int_x^\infty y^{-a}dy = \frac{m_1}{a-1}\,x^{1-a}
\leq \frac{m_0^{\frac{1-a}{b}}m_1}{1-a}\,f_X^{\frac{a-1}{b}}(x)
\end{equation*}
and similarly
\begin{align*}
1-F_X(x) \geq
\frac{m_1^{\frac{1-b}{a}}m_0}{1-b}\,f_X^{\frac{b-1}{a}}(x)
\end{align*}
Identical derivations hold for $F_X(x)$ and $x<-x_0$, and thus
setting $\gamma=1-F_X(x_0)$ the tail regularity is established.

\item By the assumption, there exist $0<m_0<m_1$, $a\geq 1$, $b>0$
and $x_0>1$ so that for any $|x|>x_0$
\begin{equation*}
m_0e^{-b|x|^a} \leq f_X(x) \leq m_1e^{-b|x|^a}
\end{equation*}
Thus, for any $x>x_0$
\begin{align*}
1-F_X(x) &\leq m_1\int_x^\infty e^{-by^a}dy
\leq m_1\int_x^\infty
\left(\frac{y}{x}\right)^{a-1}e^{-by^a}dy
\stackrel{(z=y^a)}{\leq}
m_1\int_{x^a}^\infty \frac{1}{ax^{a-1}}\,e^{-bz}dz
=
\frac{m_1}{abx^{a-1}}\;e^{-bx^a}
\\
& \leq \frac{m_1}{m_0ab}\,f_X(x)
\end{align*}
and on the other hand
\begin{align*}
\left(1-F_X(x)\right)(ab+(a-1)x^{-a})
&\geq m_0\int_x^\infty
(ab+(a-1)x^{-a}) e^{-by^a}dy
\\
&\geq m_0\int_x^\infty
(ab+(a-1)y^{-a}) e^{-by^a}dy
\stackrel{(\rm a)}{=}
-m_0\frac{e^{-by^a}}{y^{a-1}}\;\bigg{|}_x^\infty
=m_0\frac{e^{-bx^a}}{x^{a-1}}
\end{align*}
where (a) is easily verified by differentiation. Thus for any
$x>x_0$
\begin{align*}
1-F_X(x) \geq \frac{m_0x}{abx^a+a-1}\,e^{-bx^a} \geq
m_0m_1^{-\frac{\beta}{b}}f_X^{\frac{\beta}{b}}(x)
\end{align*}
where the last inequality holds for $x>x_0$ with suitable
selection of $\beta>b$. Identical derivations hold for $F_X(x)$
and $x<-x_0$, and thus setting $\gamma=1-F_X(x_0)$ the tail
regularity is established.
\end{enumerate}

\end{proof}

\begin{proof}[Proof of Lemma \ref{lem:reg_cond}]
\begin{enumerate}[(i)]
\item Let $0<M\dfn
{\displaystyle\inf_{\isupp(\Theta,\Phi)}}f_{\Phi|\Theta}(\phi|\theta)$. Since $\isupp(\Theta,\Phi)$ is convex in the $\theta$-direction and $f_{\Phi|\Theta}(\phi|\theta) = f_{\Theta|\Phi}(\theta|\phi)$, we have that $\ssc{-}f^\eps_{\Phi|\Theta}(\phi|\theta) \geq M$ over $\isupp(\Theta,\Phi)$. Therefore:
{\allowdisplaybreaks
\begin{align*}
0&\leq
D(P_{\Phi|\Theta}\,\|\,\ssc{-}P^\eps_{\Phi|\Theta}\,|\,P_\Theta)
=
\hspace{-0.2cm}\dint{\isupp(\Theta,\Phi)}\hspace{-0.1cm}f_{\Phi|\Theta}(\phi|\theta)\log\frac{f_{\Phi|\Theta}(\phi|\theta)}{\ssc{-}f^\eps_{\Phi|\Theta}(\phi|\theta)}
\,d\theta d\phi
\\
&\leq
\dint{\isupp(\Theta,\Phi)}\hspace{-0.1cm}f_{\Phi|\Theta}(\phi|\theta)\log\frac{f_{\Phi|\Theta}(\phi|\theta)}{M}
\,d\theta d\phi
= -\left(h(\Phi|\Theta)+\log{M}\right) <\infty
\end{align*}}
where in the last inequality we used the finiteness of the joint entropy and $\Theta\sim\m{U}$. The same
holds for
$D(P_{\Phi|\Theta}\,\|\,\ssc{+}P^\eps_{\Phi|\Theta}\,|\,P_\Theta)$,
concluding the proof.

\item Since both $F_X,F_Y$ are now bijective, we have that
\begin{equation*}
F_{\Phi|\Theta}(\phi|\theta) = F_{\sst
Y|X}(F^{-1}_Y(\phi)|F^{-1}_X(\theta))
\end{equation*}
and thus
\begin{equation*}
f_{\Phi|\Theta}(\phi|\theta) =
\frac{\partial}{\partial\phi}\left(F_{Y|X}(F^{-1}_Y(\phi)|F^{-1}_X(\theta))\right)
=\frac{f_{Y|X}(F^{-1}_Y(\phi)|F^{-1}_X(\theta))}{
f_Y(F^{-1}_Y(\phi))}
= \frac{
f_{X|Y}(F^{-1}_X(\theta))|F^{-1}_Y(\phi))}{f_X(F^{-1}_X(\theta))}
\end{equation*}
We can therefore write
\begin{equation*}
\ssc{-}f^\eps_{\Phi|\Theta}(\phi|\theta) = \inf_{\xi\in
J_\eps^-(\phi,\theta)}\frac{
f_{X|Y}(F^{-1}_X(\xi))|F^{-1}_Y(\phi))}{f_X(F^{-1}_X(\xi))}
\geq m^{-1}\cdot\inf_{\xi\in J_\eps^-(\phi,\theta)}
f_{X|Y}(F^{-1}_X(\xi)|F^{-1}_Y(\phi))
\end{equation*}
where $m \dfn \sup f_X(x)<\infty$. Denote the max-to-min ratio bound by
\begin{equation*}
M = \sup_{y\in\isupp(Y)}\left(\frac{\sup_{x\in\isupp(X|Y=y)}
f_{X|Y}(x|Y=y)}{\inf_{x\in\isupp(X|y)} f_{X|Y}(x|y)}\right)
\end{equation*}
The relative entropy
$D(f_{\Phi|\Theta}\,\|\,\ssc{-}f^\eps_{\Phi|\Theta})$ is now upper
bounded as follows:
{\allowdisplaybreaks
\begin{align}\label{eq:divergence_bound}
\nonumber
D(P_{\Phi|\Theta}&\,\|\,\ssc{-}P^\eps_{\Phi|\Theta}\,|\,P_\Theta)
\leq
\\
\nonumber &\leq \dint{\isupp(\Theta,\Phi)}\hspace{-0.1cm}\frac{
f_{X|Y}(F^{-1}_X(\theta)|F^{-1}_Y(\phi))}{
f_X(F^{-1}_X(\theta))}\log\frac{
f_{X|Y}(F^{-1}_X(\theta)|F^{-1}_Y(\phi))}{f_X(F^{-1}_X(\theta))\cdot
m^{-1}\cdot{\hspace{-0.4cm}\displaystyle\inf_{\xi\in
J_\eps^-(\phi,\theta)}}f_{X|Y}(F^{-1}_X(\xi))|F^{-1}_Y(\phi))}
\,d\theta d\phi
\\
\nonumber &=
\log{(m)}+\hspace{-0.2cm}\dint{\isupp(X,Y)}\hspace{-0.1cm}
f_{X|Y}(x|y) f_Y(y)\log\left(\frac{1}{f_X(x)}\cdot\frac{
f_{X|Y}(x|y)}{\inf_{z\in \widehat{J}_\eps^-(y,x)}
f_{X|Y}(z|y)}\right) \,dxdy
\\
&\leq  \log{(m)}+h(X) + \log{M} < \infty
\end{align}}
where a straightforward change of variables was performed, and $\widehat{J}_\eps^-(y,x)$ is the counterpart of $J_\eps^-(\phi,\theta)$. In the last inequality we used the fact that $\isupp(X,Y)$ is convex in the $y$-direction, which implies that $\widehat{J}_\eps^-(y,x)\subseteq \isupp(X|Y=y)$. Furthermore, $h(X)$ is finite since $f_X$ is proper and bounded.

\item We prove the claim under the lenient assumption that
$f_{X|Y}(x|y)$ is also symmetric for any fixed $y$. The argument
for the general claim is a similar yet more tedious version of
this proof. We need the following Lemma:
\begin{lemma}\label{lem:aux_div}
Suppose $X$ is proper with a symmetric unimodal p.d.f., a finite
variance $\sigma^2$, and a regular tail with parameters
$\gamma,c_i,\alpha_i$. Define
\begin{equation*}
f_X^*(x) \dfn \inf_{z\in (e(x),x)}f_X(x) \,,\quad
e(x)\dfn F_X^{-1}\left(\frac{1}{2}F_X(x)\right)
\end{equation*}
and let
\begin{equation*}
\gamma^*\dfn \min(\gamma,\frac{1}{3})\,,\quad M \dfn \sup
f_X(x)\,,\quad M_1 \dfn\frac{\gamma^*}{4}\left(\sqrt{\frac{2}{\gamma^*}}\,\sigma-\frac{1-\gamma^*}{2M}\right)
\end{equation*}
Then
\begin{equation*}
D(f_X\|f_X^*)
\leq \alpha_1^{-1}\log{\frac{2c_1}{c_0}} +
(1+\alpha_1-\alpha_0)\log{M}+\log{M_1}
\end{equation*}
\end{lemma}
\begin{proof}
Without loss of generality we can assume that
$\gamma<\frac{1}{3}$, since a larger value implies a regular tail
for any smaller value. Define $x_2<x_1<x_0<0$ to be
\begin{equation*}
x_0 = F_X^{-1}\left(\gamma\right)\,,\quad  x_1 =
F_X^{-1}\left(\frac{\gamma}{2}\right) \,,\quad  x_2 =
F_X^{-1}\left(\frac{\gamma}{4}\right)
\end{equation*}
It is easy to see that $e(x_0)=x_1$ and $e(x_1)=x_2$. Defining
$M=\sup f_X(x)$ we can lower bound $|x_1|$ using symmetry:
\begin{equation*}
2|x_1|M \geq 1-\gamma \quad\Rightarrow\quad  |x_1| \geq
\frac{1-\gamma}{2M}
\end{equation*}
Using Chebyshev's inequality and symmetry, we can upper bound
$|x_2|$ by
\begin{equation*}
2\int_{|x_2|}^\infty f_X(x)dx = \frac{\gamma}{2} \leq
\frac{\sigma^2}{x_2^2} \quad \Rightarrow\quad |x_2|\leq
\sqrt{\frac{2}{\gamma}}\,\sigma
\end{equation*}
Combining the above and using the monotonicity of $f$ for $x<0$, we have
\begin{equation*}
f_X(x_1)\cdot (|x_2|-|x_1|) \geq \frac{\gamma}{4}
\end{equation*}
which yields a lower bound for $f_X(x_1)$:
\begin{equation*}
f_X(x_1) \geq \frac{\gamma}{4(|x_2|-|x_1|)} \geq \frac{\gamma}{4}\left(\sqrt{\frac{2}{\gamma}}\,\sigma-\frac{1-\gamma}{2M}\right)
= M_1
\end{equation*}
and since $f_X$ is symmetric and unimodal and by the assumption
$\gamma<\frac{1}{3}$, it is readily verified that
\begin{align}\label{eq:f*bound}
\nonumber &f_X^*(x) = f_X(e(x)) &x\in(-\infty,x_0)
\\
\nonumber &f_X^*(x) \geq f_X(x_1)\geq M_1 &x\in(x_0,|x_0|)
\\
&f_X^*(x) = f_X(x) &x\in(|x_0|,\infty)
\end{align}
Now, recall that $f$ has a regular tail, which is this symmetric
case means that (recall that $x_0<0$)
\begin{equation*}
c_0f^{\alpha_0}(x)\leq F(x) \leq  c_1f^{\alpha_1}(x)\,\qquad
|x|>|x_0|
\end{equation*}
Let us upper bound the relative entropy between $f_X,f_X^*$ using
the above together with (\ref{eq:f*bound}):
\begin{align*}
D(f_X\|f_X^*) &= \int_{-\infty}^{x_0}
f_X(x)\log\frac{f_X(x)}{f_X^*(x)}\,dx+\int_{x_0}^{|x_0|}
f_X(x)\log\frac{f_X(x)}{f_X^*(x)}\,dx+\int_{|x_0|}^\infty
f_X(x)\log\frac{f_X(x)}{f_X^*(x)}\,dx
\\
&\leq  \int_{-\infty}^{x_0}
f_X(x)\log\frac{f_X(x)}{f_X(e(x))}\,dx+\int_{x_0}^{|x_0|}f_X(x)\log\frac{M}{M_1}\,dx
+ \int_{|x_0|}^\infty f_X(x)\log{1}\,dx
\\
& \leq \alpha_1^{-1}\int_{-\infty}^{x_0}
f_X(x)\log\left(\frac{f_X^{\alpha_1}(x)}{F_X(x)}\frac{2F_X(e(x))}{f_X^{\alpha_1}(e(x))}\right)\,dx
+\log{M}-\log{M_1}
\\
&\leq \alpha_1^{-1}\int_{-\infty}^{x_0}
f_X(x)\log\left(\frac{f_X^{\alpha_1}(x)}{c_0f_X^{\alpha_0}(x)}\frac{2c_1f_X^{\alpha_1}(e(x))}{f_X^{\alpha_1}(e(x))}\right)\,dx
+\log{M}-\log{M_1}
\\
&\leq \alpha_1^{-1}\log{\frac{2c_1}{c_0}} +
(1+\alpha_1-\alpha_0)\log{M}-\log{M_1} \;<\; \infty
\end{align*}
\end{proof}

Returning to the pursued claim, let $\gamma,c_i,\alpha_i$ be the
common tail parameters of $f_{X|Y}(\cdot|y)$, let $M=\sup
f_{X|Y}(x|y)$ and let $\sigma^2$ be an upper bound on the variance
of $f_{X|Y}(\cdot|y)$ for all $y$. It follows from definition that
for any $y$
\begin{equation*}
\inf_{z\in \widehat{J}_\eps^-(y,x)} f_{X|Y}(z|y) \geq
f_{X|Y}^*(x|y)
\end{equation*}
where $f_{X|Y}^*$ is defined as in Lemma \ref{lem:aux_div}. We now follow the derivations of the previous claim (\ref{cond:reg_min_max}) up to (\ref{eq:divergence_bound}), and use the above inequality and Lemma \ref{lem:aux_div} to obtain:
{\allowdisplaybreaks
\begin{align*}
D(P_{\Phi|\Theta}\,\|\,\ssc{-}P^\eps_{\Phi|\Theta}\,|\,P_\Theta)
&\leq \log{(m)}+h(X)+\int_{\isupp(Y)}\hspace{-0.5cm}
f_Y(y)dy\int_{\isupp(X|Y=y)}\hspace{-0.8cm} f_{X|Y}(x|y)\log\frac{
f_{X|Y}(x|y)}{ f_{X|Y}^*(x|y)} \,dx
\\
& = \log{(m)}+h(X)+ \int_{\isupp(Y)}\hspace{-0.5cm}
f_Y(y)D(f_{X|Y}(\cdot|y)\,\|\, f_{X|Y}^*(\cdot|y))dy
\\
& = \log{(m)}+h(X)+ \alpha_1^{-1}\log{\frac{2c_1}{c_0}} +
(1+\alpha_1-\alpha_0)\log{M}-\log{M_1} <\infty
\end{align*}}
The same proof holds for
$D(P_{\Phi|\Theta}\,\|\,\ssc{+}P^\eps_{\Phi|\Theta}\,|\,P_\Theta)$.

\item A direct consequence of (\ref{cond:reg_bounded_away}).

\end{enumerate}
\end{proof}

\bibliographystyle{IEEEbib}

\end{document}